\documentclass[12pt]{article}
\usepackage[english]{babel}
\usepackage[utf8]{inputenc}
\usepackage{amsmath,amsthm}
\usepackage{amssymb}
\usepackage{enumitem}
\usepackage{graphicx, subfig}
\usepackage{rotating}
\usepackage{paralist}
\usepackage{placeins}
\usepackage[hidelinks]{hyperref}
\usepackage[ruled,vlined,linesnumbered]{algorithm2e}
\usepackage[round]{natbib}
\usepackage{multibib}
\newcites{app}{References for Supplementary Material}
\usepackage{tikz}
\usetikzlibrary{positioning, shapes, calc, arrows.meta}
\usepackage{xspace}
\usepackage{multirow}
\usepackage{supertabular}

\usepackage[textwidth=1.8cm, textsize=scriptsize]{todonotes}

\newcommand{\du}{\textrm{d}u}
\newcommand{\dv}{\textrm{d}v}
\newcommand{\dx}{\textrm{d}x}
\newcommand{\dy}{\textrm{d}y}
\newcommand{\dz}{\textrm{d}z}
\newcommand{\dt}{\textrm{d}t}

\let\oldnl\nl
\newcommand{\nonl}{\renewcommand{\nl}{\let\nl\oldnl}}

\renewcommand{\descriptionlabel}[1]{%
  \hspace\labelsep \upshape\bfseries #1%
}
 
\theoremstyle{definition}

\makeatletter
\let\orgdescriptionlabel\descriptionlabel
\renewcommand*{\descriptionlabel}[1]{%
  \let\orglabel\label
  \let\label\@gobble
  \phantomsection
  \edef\@currentlabel{#1\unskip}%
  \let\label\orglabel
  \orgdescriptionlabel{#1}%
}
\makeatother

\newtheorem{lemma}{Lemma}
\newtheorem{prop}{Proposition}

\newtheorem{corollary}{Corollary}
\newcommand{\HH}{\ensuremath{\mathbb{H}}\xspace}
\newcommand{\X}{\ensuremath{\mathbb{X}}\xspace}
\newcommand{\Y}{\ensuremath{\mathbb{Y}}\xspace}
\newcommand{\measurable}{\ensuremath{\mathcal{B}_b}}

\newcommand{\ESS}{\textrm{ESS}}

\DeclareMathOperator{\Exp}{\mathbb{E}}
\DeclareMathOperator{\pr}{\mathbb{P}}
\DeclareMathOperator{\KL}{KL}
\DeclareMathOperator{\mise}{MISE}
\DeclareMathOperator{\ise}{ISE}
\DeclareMathOperator{\mse}{MSE}
\DeclareMathOperator{\BL}{BL}
\DeclareMathOperator{\N}{\mathcal{N}}
\DeclareMathOperator{\Kop}{\operatorname{K}}
\DeclareMathOperator{\emsmap}{\operatorname{F}_{EMS}}
\DeclareMathOperator{\emmap}{\operatorname{F}_{EM}}
\def\real{\mathbb{R}}
\def\weight{G_n}
\def\weightN{G_n^N}
\def\predictive{\eta_n}
\def\predictiveN{\eta^N_n}

\def\updateN{\hat{\eta}_n^N}
\def\update{\hat{\eta}_n}
\def\bgN{\Psi_{G_n^N}}
\def\bg{\Psi_{G_n}}
\def\testfn{\varphi}
\newcommand{\lp}{\ensuremath{\mathbb{L}_p}}
\newcommand{\norm}[2]{\ensuremath{\Vert #1 \Vert_{#2}}}
\newcommand{\supnorm}[1]{\norm{#1}{\infty}}

\def\bwN{s_N}
\def\predictiveX{\eta_{n}\vert_{\X}}
\def\predictiveNX{\eta^N_{n}\vert_{\X}}

\newcommand{\ourtitle}{A Particle Method for Solving\\ Fredholm Equations of the First Kind}
\graphicspath{{eps/}}

\def\version{2}

\if\version2 
	\usepackage[a4paper, margin = 2.5cm]{geometry}
	\fontsize{10pt}{12}\selectfont
\else
\fi

\begin{document}

\def\spacingset#1{\renewcommand{\baselinestretch}%
{#1}\small\normalsize}
\spacingset{1}

\if1\version 
{
  \bigskip
  \bigskip
  \bigskip
  \begin{center}
    {\LARGE\bf \ourtitle}
\end{center}
  \medskip
}
\else 
{
  \title{\ourtitle}
	\author{Francesca R. Crucinio\thanks{This work was supported by funding from the EPSRC and MRC OXWASP Centre for Doctoral Training EP/L016710/1, EPSRC grants EP/R034710/1 and EP/T004134/1, and the Lloyd’s Register Foundation Programme on Data-Centric Engineering at the Alan Turing Institute.}\hspace{.2cm}\\
Department of Statistics, University of Warwick\\
and\\
Arnaud Doucet\\
University of Oxford\\
and\\
Adam M. Johansen\\
University of Warwick \& The Alan Turing Institute}
\maketitle
} \fi
\begin{abstract}
Fredholm integral equations of the first kind are the prototypical example of ill-posed linear inverse problems. They model, among other things, reconstruction of distorted noisy observations and indirect density estimation and also appear in instrumental variable regression. However, their numerical solution remains a challenging problem. Many techniques currently available require a preliminary discretization of the domain of the solution and make strong assumptions about its regularity. For example, the popular expectation maximization smoothing (EMS) scheme requires the assumption of piecewise constant solutions which is inappropriate for most applications. We propose here a novel particle method that circumvents these two issues. This algorithm can be thought of as a Monte Carlo approximation of the EMS scheme which not only performs an adaptive stochastic discretization of the domain but also results in smooth approximate solutions.
We analyze the theoretical properties of the EMS iteration and of the corresponding particle algorithm. Compared to standard EMS, we show experimentally that our novel particle method provides state-of-the-art performance for realistic systems, including motion deblurring and reconstruction of cross-section images of the brain from positron emission tomography.
 
\end{abstract}

\noindent%
\vfill


\if\version2 
        \spacingset{1}
	\fontsize{10pt}{12}\selectfont
\else
	\spacingset{1.5} 
\fi

\section{Introduction}
We consider Fredholm equations of the first kind of the form 
\begin{equation}
\label{eq:fe}
h(y) = \int_{\X} f(x) g(y \mid x) \dx\qquad \forall y \in \Y,
\end{equation}
with $f(x)$ and $h(y)$ probability densities on $\X$ and $\Y$, respectively, and $g(y \mid x)$ the density of a Markov kernel from $\X$ to $\Y$. Given $g$ and (some characterization of) $h$, we aim to estimate $f$. Our particular interest is the setting in which we have access to a collection of samples from $h$, rather than the function itself.

This class of equations has numerous applications in statistics and applied mathematics. For example, $h$ might correspond to a mixture model for which we wish to estimate its mixing distribution, $f$, from samples from $h$. This problem is known as density deconvolution or indirect density estimation \citep{delaigle2008alternative,ma2011indirect,pensky2017minimax,yang2020density}. In epidemiology, \eqref{eq:fe} links the incidence curve of a disease to the observed number of cases \citep{goldstein2009reconstructing, gostic2020practical, marschner2020back}. In instrumental variable regression and causal inference, Fredholm equations can be used to estimate a nonlinear regression function or identify causal effects in the presence of confounders  \citep{hall2005nonparametric,miao2018identifying}. Since the seminal work of 
\citet{vardi1985statistical,vardi1993image}, Fredholm equations have also been widely used in positron emission tomography. In this and similar contexts, $f$ corresponds to an image which needs to be inferred from noisy measurements \citep{aster2018parameter,clason2019regularization,snyder1992deblurring,jin2019expectation}.

In most interesting cases, Fredholm integral equations of the first kind are ill-posed and it is necessary to introduce a regularizer to obtain a unique solution. Solving the regularized problem remains computationally very challenging. 
For certain subclasses of this problem, such as density deconvolution \citep{delaigle2008alternative} good methods exist and can achieve optimal convergence rates as the number of observations increases \citep{carroll1988optimal}. However, generally applicable approaches which do not assume a particular form of $g$ typically require discretization of the domain, $\X$, which restricts their applications to low-dimensional scenarios, and often assume a piecewise-constant solution \citep{burger2019entropic, koenker2014convex,ma2011indirect,tanana2016approximate,yang2020density}. This is the case for the popular Expectation Maximization Smoothing (EMS) scheme \citep{silverman1990smoothed}, a smoothed version of the infinite dimensional expectation maximization algorithm of \cite{kondor1983method}.

In this paper, our contributions are three-fold. First, we provide novel theoretical results for the EMS scheme on continuous spaces, establishing that it admits a fixed point under weak assumptions. Second, we propose a novel particle version of EMS which does not suffer from the limitations of the original scheme. This Monte Carlo algorithm provides an adaptive stochastic discretization of the domain and outputs a sample approximation of $f$ through which a smooth approximation can be obtained via a natural kernel density estimation procedure. Although this algorithm is related to sequential Monte Carlo (SMC) methods which have been widely used to perform inference for complex Bayesian models \citep{chopin2020,del2013mean,douc2014nonlinear,doucet2011tutorial,liu1998sequential,liu2001monte}, standard SMC convergence results do not apply to this scheme so we also provide an original theoretical analysis of the algorithm. Third, we demonstrate this algorithm on both illustrative examples and realistic image processing applications. 

The rest of this paper is organized as follows. In Section~\ref{sec:background}, we review Fredholm integral equations of the first kind and the EMS algorithm, and establish existence of a fixed point for the continuous version. In Section~\ref{sec:mf}, we introduce a particle approximation of the EMS recursion and provide convergence results for this scheme. We demonstrate the application of the algorithm in Section~\ref{sec:examples} and then briefly conclude.

\section{Fredholm equations and EMS}
\label{sec:background}

\subsection{Fredholm integral equations of the first kind}
\label{sec:fe}
We recall that we consider equations of the form~\eqref{eq:fe}.
We concern ourselves in particular with the case in which
\begin{compactdesc}
\item[(A0)\label{a:space}] $\X\subset \real^{d_{\X}}$ and $\Y\subset \real^{d_{\Y}}$ are compact subsets of Euclidean spaces, $g$ can be evaluated pointwise and a sample, $\mathbf{Y}$, from $h$ is available.
\end{compactdesc}

In most applications the space $\HH= \X \times \Y \subset \real^{d_{\X}\times d_{\Y}}$ is closed and bounded and~\ref{a:space} is satisfied. For instance, in image processing both $\X$ and $\Y$ are typically of the form $[-a, a]\times[-b, b]$ for $a, b>0$, $f$ and $h$ are continuous densities on $\X$ and $\Y$, respectively. 
In applications the analytic form of $h$ is often unknown, and the available data arise from discretization of $h$ over $\Y$, as in, e.g., \cite{vardi1993image}, or from sampling, as in, e.g. \cite{ma2011indirect}. In the image processing context, the available data are usually either the values of $h$ over the discretization of $\Y$ induced by the pixels of the image (e.g. an image with $10\times 10$ pixels induces a discretization on $\Y$ in which the intervals $[-a, a]$ and $[-b, b]$ are each divided into 10 bins) or samples from $h$. We focus here on the sampling case.

Considering~\eqref{eq:fe} in the context of probability densities is not too restrictive. A wider class of integral equations can be recast in this framework by appropriate normalizations and translations, provided that $f$ and $h$ are bounded below \citep[Section 6]{chae2018algorithm}.

As the set of probability densities on $\X$ is not finite, if the kernel $g$ is not degenerate then the resulting integral equation is in general ill-posed \citep[Theorem 15.4]{kress2014linear}.
Fredholm's alternative (see, e.g., \citet[Corollary 4.18]{kress2014linear}) gives a criterion to assess the existence of solutions of~\eqref{eq:fe}; however, the lack of continuous dependence on $h$ causes the solutions to be unstable and regularization techniques are needed \citep{kress2014linear, groetsch2007integral}.
Common methods are Tikhonov regularization \citep{tihonov1963solution} and iterative methods \citep{landweber1951iteration, kondor1983method}.
See \cite{yuan2019overview} for a recent review.
\subsection{Expectation Maximization and Related Algorithms}
\subsubsection{Expectation Maximization}
\label{sec:em}

From a statistical point of view,~\eqref{eq:fe} describes an indirect density estimation problem: the mixing density $f$ has to be recovered from the mixture $h$. 
This can in principle be achieved by maximizing an incomplete data likelihood for $f$ through the Expectation Maximization (EM) algorithm \citep{dempster1977maximum}.
Nevertheless, the maximum likelihood estimator is not consistent, as the parameter to be estimated (i.e. $f$) is infinite dimensional \citep{laird1978nonparametric}; a problem aggravated by the ill-posedness of~\eqref{eq:fe} \citep{silverman1990smoothed}.

We briefly review a number of iterative schemes based on the EM algorithm which aim to find approximate solutions of~\eqref{eq:fe} through regularization. The starting point is the iterative method of \cite{kondor1983method}, an infinite dimensional EM algorithm,
\begin{equation}
\label{eq:em}
f_{n+1}(x) = f_n(x) \int \frac{g(y \mid x)}{\int f_n(z) g(y \mid z)
  \dz} h(y) \dy,
\end{equation}
which minimizes the Kullback--Leibler divergence,
\begin{equation}
\label{eq:minimization}
\KL\left(h, \int_{\X}f(x)g(\cdot \mid x)\ \dx\right) = \int_{\Y} h(y)\log\left( \frac{h(y)}{\int_{\X}f(x)g(y \mid x)\ \dx}\right) \dy,
\end{equation}
with respect to $f$ over the set of probability densities on $\X$ \citep{multhei1989properties}.
Minimizing~\eqref{eq:minimization} is equivalent to maximizing
\begin{align*}
\Lambda(f) := \int_{\Y} h(y)\log \int_{\X}f(x)g(y \mid x)\ \dx \ \dy, 
\end{align*}
a continuous version of the incomplete data log-likelihood for the function $f$ \citep{multhei1989properties}.
This scheme has several good properties, iterating~\eqref{eq:em} monotonically decreases~\eqref{eq:minimization}  \citep[Theorem 7]{multhei1987iterative} and if the iterative formula converges, then the limit is a minimizer of~\eqref{eq:minimization} \citep[Theorem 8]{multhei1987iterative} --- but the minimizer need not be unique.
Convergence of the EM iteration~\eqref{eq:em} to a fixed point has been proved under the existence of a sequence $(f^\star_s)_{s\geq 1}$ with $h^\star_s(y) = \int_{\X} f^\star_s(x)g(y \mid x)\ \dx$, such that $\KL(h, h^\star_s)$ converges to the infimum of~\eqref{eq:minimization} and additional integrability conditions \citep{chae2018algorithm}.

In general, implementing the recursive formula~\eqref{eq:em} analytically is not possible and discretization schemes are needed. Under the assumption of piecewise constant densities $f$, $h$ and $g$, with the discretization grid fixed in advance, the EM recursion~\eqref{eq:em} reduces to the EM algorithm for Poisson data \citep{vardi1993image}, known as the Richardson--Lucy (RL) algorithm in the image processing field \citep{richardson1972bayesian, lucy1974iterative}, where the intensities of pixels are modeled as Poisson counts,
\begin{equation}
\label{eq:em_discrete}
f_{b}^{(n+1)} = f_{b}^{(n)}\sum_{d=1}^D\left(\frac{h_d g_{bd}}{\sum_{k=1}^B f_k^{(n)}g_{kd}}\right),
\end{equation}
here $f_b$ for $b=1,\ldots, B$ and $h_d$ for $d=1, \ldots, D$ are the constant values over the deterministic discretization of the space for $f$ and $h$ respectively.

The Iterative Bayes (IB) algorithm of \cite{ma2011indirect} considers the case in which only samples from $h$ are available. These samples are used to build a kernel density estimator (KDE) for $h$, which is then plugged into the discretized EM iteration~\eqref{eq:em_discrete}.

As discussed earlier, despite being popular and easy to implement, the EM algorithm~\eqref{eq:em_discrete} has a number of drawbacks: after a certain number of iterations the EM approximations deteriorate resulting in unstable estimates that lack smoothness and give spiky estimates of $f$ \citep{silverman1990smoothed, nychka1990some}; in fact minimizing~\eqref{eq:minimization} does not deal with the ill-posedness of the problem and regularization is needed \citep{byrne2015algorithms}.

A natural way to introduce regularization is via maximum penalized likelihood estimation (MPLE; see, e.g. \cite{green1990use}), maximizing, for some penalty term, $P$:
\begin{align*}
\Lambda^\prime(f) := \int_{\Y} h(y)\log \int_{\X}f(x)g(y \mid x)\ \dx \ \dy - P(f).
\end{align*}
In most cases, an updating formula like~\eqref{eq:em_discrete} cannot be obtained straightforwardly for MPLE because the derivative of $P(f)$ usually involves several derivatives of $f$. A possible solution is to update the estimate of $f$ from iteration $f_n$ to $f_{n+1}$ evaluating the penalty term at $f_n$, rather than at the new value $f_{n+1}$. This is known as the one-step late (OSL) algorithm \citep{green1990use}.
The resulting update formula is usually easier to compute but there is no guarantee that each iteration will increase the penalized log-likelihood. However, if convergence occurs, the OSL algorithm converges more quickly than the corresponding EM for the penalized likelihood.

\subsubsection{Expectation Maximization Smoothing}
An easy-to-implement regularized version of the EM recursion~\eqref{eq:em_discrete} is the EMS algorithm of \cite{silverman1990smoothed}, an EM-like algorithm in which a smoothing matrix $\Kop$ is applied to the EM estimates at each iteration
\begin{equation}
\label{eq:ems_discrete}
f_{b}^{(n+1)} = \sum_{\kappa=1}^B \Kop_{b\kappa}f_{\kappa}^{(n)}\sum_{d=1}^D\left(\frac{h_d g_{\kappa d}}{\sum_{k=1}^B f_k^{(n)}g_{kd}}\right).
\end{equation}
The EMS algorithm has long been attractive from a practical point of view as the addition of the smoothing step to the EM recursion~\eqref{eq:em_discrete} gives good empirical results, with convergence occurring empirically in a relatively small number of iterations (e.g. \citet{silverman1990smoothed, li2017prober, becker1991method}).

Under mild conditions on the smoothing matrix the discretized EMS recursion~\eqref{eq:ems_discrete} has a fixed point \citep{latham1992hyperplane}. In addition, with a particular choice of smoothing matrix, the fixed point of~\eqref{eq:ems_discrete} minimizes a penalized likelihood with a particular roughness penalty \citep{nychka1990some}. With this choice of penalty, the OSL and the EMS recursion have the same fixed point \citep{green1990use}.
\cite{fan2011local} establish convergence of~\eqref{eq:ems_discrete} to local-EM, an EM algorithm for maximum local-likelihood estimation, when the smoothing kernel is a symmetric positive convolution kernel with positive bandwidth and bounded support. If the space on which the EMS mapping is defined is bounded, the discrete EMS mapping is globally convergent for sufficiently large bandwidth.

The focus of this work is a continuous version of the EMS recursion, in  which we do not discretize the space and use smoothing convolutions 
$K f(\cdot):=\int_{\X} K(u, \cdot)f(u)\ \du$ in place of smoothing matrices, i.e.
\begin{align}
\label{eq:ems}
f_{n+1}(x) 
&= \int_{\X} K(x^\prime, x) f_n(x^\prime) \int_{\Y} \frac{g(y \mid x^\prime) h(y)}{\int_{\X} f_n(z) g(y \mid z) \dz}\ \dy \ \dx^\prime.
\end{align}

\subsection{Properties of the Continuous EMS Recursion}
\label{sec:cems}
Contrary to the discrete EMS map~\eqref{eq:ems_discrete}, relatively little is known about the continuous EMS mapping. We prove, under the following assumptions, that it also admits a fixed point in the space of probability distributions:
\begin{description}
\item[(A1)\label{a:g}] The density of the kernel $g(y \mid x)$ is continuous and bounded away from 0 and $\infty$:
\begin{equation*}
 \exists m_g>0 \text{ such that }\qquad 0<  m_g^{-1} \leq  g(y \mid x) \leq m_g < \infty \qquad \forall (x, y)\in \X\times\Y.
\end{equation*}
\item[(A2)\label{a:k}] The smoothing 
kernel is specified via a continuous bounded density, $T$, over $\real^{d_{\X}}$, such that
  $\inf_{v \in \X} \int_\X T(u-v) \du > 0$ as:
\begin{align*}
K(v, u) = \frac{T(u-v)\mathbb{I}_{\X}(u)}{\int_{\X} T(u^\prime-v)\ \du^\prime}.
\end{align*}
\end{description}

Assumption \ref{a:g} is common in the literature on Fredholm integral equations as continuity of $g$ rules out degenerate integral equations which require special treatment \citep[Chapter 5]{kress2014linear}. The boundedness condition on $g$ ensures the existence of a minimizer of~\eqref{eq:minimization} \citep[Theorem 1]{multhei1992iterative}. Assumption \ref{a:k} on $T$ is mild and is satisfied by most commonly used kernels for density estimation \citep{silverman1986density} and implies that $K(v, \cdot)$ is a density over $\X$ for any fixed $v$. We can draw samples from $K(v, \cdot)$, e.g. by rejection sampling whenever $T$ is proportional to a density from which sampling is feasible.

The EMS map describes one iteration of this algorithm, for a probability density $f$, \begin{equation*}
\emsmap : f \mapsto \emsmap f := \int_{\X} f(x^\prime)  K(x^\prime, \cdot)  \int_{\Y} \frac{g(y \mid x^\prime) h(y)}{\int_{\X} f(z)g(y \mid z) \dz  }\ \dy\ \dx^\prime.
\end{equation*}
It is the composition of linear smoothing by the kernel $K$ defined in~\ref{a:k}
and the non-linear map corresponding to the EM iteration, $\emmap$,
\begin{align}
\label{eq:emop}
  \emmap(f)(x) &= \frac{\bar{G}_{f}(x) f(x)}{f(\bar{G}_{f})}
&\textrm{ where }
\bar{G}_{f}(\cdot) := \int_{\Y} \frac{g(y \mid \cdot)h(y)}{\int_{\X}f(z) g(y \mid z) \dz}\ \dy
\end{align}
and we introduce the normalizing constant $f(\bar{G}_{f})\equiv 1$ to highlight the connection with the particle methods introduced in Section~\ref{sec:mf} (here and elsewhere we adopt the convention that for any suitable integrable function, $\varphi$, and probability or density, $f$, $f(\varphi)=\int f(x) \varphi(x) \dx$). That is, $\emmap$ corresponds to a simple reweighting of a probability, with the weight being given by $\bar{G}_f$.

The existence of the fixed point of $\emsmap$ is established in Appendix~\ref{app:ems} using results from non-linear functional analysis. This result is obtained taking $h$ to be any probability distribution over $\Y$, and shows that a fixed point exists both in the case in which $h$ admits a density and that in which $h$ is the empirical distribution of a sample $\mathbf{Y}$ --- the latter is common in applications, and is the setting we are concerned with. 

\begin{prop}
\label{prop:eu}
Under~\ref{a:space}, \ref{a:g} and~\ref{a:k}, the EMS map, $\emsmap$, has a fixed point in the space of probability distributions over $\X$.
\end{prop}

\section{Particle implementation of the EMS Recursion}
\label{sec:mf}
In order to make use of the continuous EMS recursion in practice, it is necessary to approximate the integrals which it contains. To do so, we develop a particle method specialized to our context via a stochastic interpretation of the recursion.

\subsection{Particle methods}
\label{sec:smc}
Particle methods also known as Sequential Monte Carlo (SMC) methods are a class of Monte Carlo methods that sequentially approximate a sequence of target probability densities $\lbrace \eta_n(z_{1:n})\rbrace_{n\geq1}$ defined on the product spaces $\HH^n$ of increasing dimension, whose evolution is described by Markov transition kernels $M_n$ and positive potential functions $G_n$ \citep{ del2013mean}
\begin{equation}
\label{eq:smc}
\eta_{n+1}(z_{1:n+1}) \propto \eta_n(z_{1:n})G_n(z_n)M_{n+1}(z_{n+1}\mid z_{n}).
\end{equation}
These sequences naturally arise in state space models (e.g. \cite{liu1998sequential,doucet2011tutorial, li2016}) and many inferential problems can be described by~\eqref{eq:smc} (see, e.g., \cite{liu2001monte,chopin2020}, and references therein).

The approximations of $\eta_n$ for $n\geq 1$ are obtained through a population of Monte Carlo samples, called particles. The population consists of a set of $N$ weighted particles $\lbrace Z_n^i, W_n^i\rbrace_{i=1}^N$ which evolve in time according to the dynamic in~\eqref{eq:smc}.
Given the equally weighted population at time $n-1$, $\lbrace \widetilde{Z}_{n-1}^i, \frac{1}{N}\rbrace_{i=1}^N$, new particle locations $Z_n^i$ are sampled from $M_n(\cdot \mid \widetilde{Z}_{n-1}^i)$ to obtain the equally weighted population at time $n$, $\lbrace Z_n^i, \frac{1}{N}\rbrace_{i=1}^N$.
Then, the fitness of the new particles is measured through $\weight$, which gives the weights $W_n^i$. The new particles are then replicated or discarded using a resampling mechanism, giving the equally weighted population at time $n$, $\lbrace \widetilde{Z}_{n}^i, \frac{1}{N}\rbrace_{i=1}^N$.
Several resampling mechanisms have been considered in the literature (\citet[page 336]{douc2014nonlinear}, \cite{gerber2019negative}) the simplest of which consists of sampling the number of copies of each particle from a multinomial distribution with weights $\lbrace W_n^i\rbrace_{i=1}^N$ \citep{gordon1993novel}.

At each $n$, the empirical distribution of the particle population provides an approximation of the marginal distribution of $Z_n$ under $\predictive$ via $\eta_n^N = N^{-1}\sum_{i=1}^N \delta_{Z_{n}^i}.$
Throughout, in the interests of brevity, we will abuse notation slightly and treat $\eta_n^N$ as a density, allowing $\delta_{x_0}(x)\dx$ to denote a probability concentrated at $x_0$. These approximations possess various convergence properties (e.g. \cite{ del2013mean}), in particular $\lp$ error estimates and a strong law of large numbers for the expectations $\eta_n^N(\testfn) := \int_{\HH} \eta_n^N(u)\testfn(u)\ \du = N^{-1} \sum_{i=1}^N \testfn(Z_n^i)$ of sufficiently regular test functions $\testfn$ \citep{crisan2002survey, miguez2013convergence}.

\subsection{A stochastic interpretation of EMS}
The EMS recursion~\eqref{eq:ems} can be modeled as a sequence of densities satisfying~\eqref{eq:smc} by considering an extended state space.
Denote by $\predictive$ the joint density at $(x, y)\in\HH$ defined by $\predictive(x,y)=f_n(x)h(y)$ so that $f_n(x) = \eta_n\vert_{\X}(x) = \int_{\Y} \eta_n(x, y)\ \dy$. This density satisfies a recursion similar to that in~\eqref{eq:ems}
\begin{align}
\label{eq:ems_extended}
\eta_{n+1}(x, y) = \int_{\X} \int_{\Y} \eta_n(x^\prime, y^\prime) K(x^\prime, x) h(y)\frac{g(y^\prime\mid x^\prime)}{\int_{\X}f_n(z)g(y^\prime \mid z)\dz}
 \ \dy^\prime \ \dx^\prime.
\end{align}
With a slight abuse of notation, we denote by $\eta_n$ the joint density of $(x_{1:n}, y_{1:n})\in\HH^n$ obtained by iterative application of~\eqref{eq:ems_extended} with the integrals removed. 
\begin{prop}
\label{prop:smcems}
The sequence of densities $\lbrace \eta_n\rbrace_{n\geq1}$ defined over the product spaces $\HH^n=(\X\times\Y)^n$ by~\eqref{eq:smc} with $z_n:=(x_n, y_n)$,
\begin{equation}
\label{eq:mutation}
M_{n+1}\left((x_{n+1}, y_{n+1})\mid (x_n, y_n)\right) = K(x_n, x_{n+1})h(y_{n+1})
\end{equation}
and 
\begin{equation}
\label{eq:potential}
\weight(x_n, y_n) = \frac{g(y_n\mid x_n)}{\int_{\X} \eta_{n}\vert_{\X}(z)g(y_n\mid z)\ \dz}
\end{equation}
satisfies, marginally, recursion~\eqref{eq:ems_extended}. 
In particular, the marginal distribution over $x_{n}$ of $\predictive$,
\begin{align}
\label{eq:marginal}
\eta_{n}\vert_{\X}(x_{n}) = \int_{\Y}\int_{\HH^{n-1}} \eta_{n}\left(x_{1:n}, y_{1:n}\right)\ \dx_{1:n-1}\dy_{1:n}=\int_{\Y} \eta_{n}\left(x_{n}, y_{n}\right)\ \dy_{n}
\end{align}
satisfies recursion~\eqref{eq:ems} with the identification $f_n(x)=\eta_{n}\vert_\X(x)$.
\end{prop}
\begin{proof}

Starting from~\eqref{eq:smc} with $M_{n+1}$ and $\weight$ as in~\eqref{eq:mutation}-\eqref{eq:potential}
\begin{align}
\label{eq:fkflow}
\eta_{n+1}(x_{1:n+1}, y_{1:n+1}) &=  \frac{\predictive(x_{1:n}, y_{1:n})G_n(x_n,y_n)}{\predictive(G_n)}M_{n+1}\left((x_{n+1}, y_{n+1})\mid (x_n, y_n)\right),
\end{align}
where $\eta_n(G_n):= \int_{\HH}\predictive(x_{n}, y_{n})G_n(x_n,y_n)\ \dx_{n}\dy_{n} = 1$, and integrating out $(x_{1:n}, y_{1:n})$
\begin{align*}
\eta_{n+1} (x_{n+1}, y_{n+1})=& \int_{\HH^n}\frac{\predictive(x_{1:n}, y_{1:n})G_n(x_n,y_n)}{\predictive(G_n)}M_{n+1}\left((x_{n+1}, y_{n+1})\mid (x_n, y_n)\right)\ \dx_{1:n}\dy_{1:n}\\
=&  \int_{\HH}\int_{\HH^{n-1}} \Big\{ \predictive(x_{1:n}, y_{1:n})\ \dx_{1:n-1} \dy_{1:n-1} \\
&\qquad \times\frac{g(y_n\mid x_n)}{\int \eta_n\vert_\X(z) g(y_n\mid z) \dz}K(x_n, x_{n+1})h(y_{n+1})\ \dx_n \dy_n\Big\}\\
=&  \int_{\HH}\predictive(x_{n}, y_{n})\frac{g(y_n\mid x_n)}{\int \eta_n\vert_\X(z) g(y_n\mid z) \dz}K(x_n, x_{n+1})h(y_{n+1})\ \dx_n \dy_n.
\end{align*}
We can then compute the marginal over $\X$, $\eta_{n+1}\vert_{\X}$
\begin{align*}
\eta_{n+1}\vert_{\X}(x_{n+1}) &= \int_{\Y}\eta_{n+1} (x_{n+1}, y_{n+1})\ \dy_{n+1}\\
 &= \int_{\Y}h(y_{n+1})\ \dy_{n+1}\int_{\HH}\predictive(x_{n}, y_{n})\frac{g(y_n\mid x_n)}{\int \eta_n\vert_\X(z) g(y_n\mid z) \dz}K(x_n, x_{n+1})\ \dx_n \dy_n\\
&= \int_{\X}\predictiveX(x_n)K(x_n, x_{n+1})\int_{\Y}h(y_n)\frac{g(y_n\mid x_n)}{\int \eta_n\vert_\X(z) g(y_n \mid z) \dz}\ \dy_n\ \dx_n
\end{align*}
which, with the given identifications, satisfies the EMS recursion~\eqref{eq:ems}.
\end{proof}
To facilitate the theoretical analysis we separate the contribution of the mutation kernels~\eqref{eq:mutation} and of the potential functions~\eqref{eq:potential}, in particular, we denote the weighted distribution obtained from $\eta_n$ by
$\bg(\predictive)(x_{n}, y_{n}) := \predictive(x_{n}, y_{n})G_n(x_n,y_n)\big/\predictive(G_n)$.

\subsection{A particle method for EMS}
Having shown that the EMS recursion describes a sequence of densities satisfying~\eqref{eq:smc}, it is possible to use SMC techniques to approximate this recursion. This involves replacing the true density at each step with a
sample approximation obtained at the previous iteration, giving rise to Algorithm~\ref{alg:fpsmc}, which describes the case in which only a fixed number of samples from $h$ are available and in line 1-2 we draw $Y_n^i$ from their empirical distribution; when sampling freely from $h$ is feasible one could instead draw these samples from it.

The resulting SMC scheme is \emph{not} a standard particle approximation of~\eqref{eq:smc}, because of the definition of the potential~\eqref{eq:potential}. Indeed, $\weight$ cannot be computed exactly, because $\predictiveX$ is not known.
The SMC scheme provides an approximation for $\predictiveX$ at time $n$.
Let us denote by $\predictiveNX$ the particle approximation of the marginal $\predictiveX$ in~\eqref{eq:marginal}
\begin{align*}
\predictiveNX := \int_{\Y} \predictiveN\left(\cdot, y_{n}\right)\ \dy_{n}=\frac{1}{N} \sum_{i=1}^N \delta_{X_n^i}.
\end{align*}
We can approximate
\begin{equation*}
\weight(x_n, y_n)= \frac{g(y_n\mid x_n)}{h_n(y_n)} = \frac{g(y_n\mid x_n)}{\int_{\X} \predictiveX(z)g(y_n \mid z)\ \dz}
\end{equation*}
using the particle approximation of the denominator $h_n(y_n):= \int_{\X} f_n(z) g(y \mid z) \dz$,
\begin{equation}
\label{eq:hN}
h^N_n(y_n) := \frac{1}{N} \sum_{i=1}^N g(y_n\mid X_n^i) = \predictiveNX\left(g(y_n \mid \cdot)\right),
\end{equation}
to obtain the approximate potentials
\begin{equation}
\label{eq:potentialN}
\weightN(x_n, y_n):= \frac{g(y_n\mid x_n)}{h^N_n(y_n)}.
\end{equation}

\begin{algorithm}[ht]
\caption{Particle Method for Fredholm Equations of the First Kind}
\label{alg:fpsmc}
\nonl At time $n=1$\\
Sample $\widetilde{X}_1^i \sim f_1$, $\widetilde{Y}_1^i$ uniformly from $\mathbf{Y}$ for $i=1,\dots,N$ and set $W_1^i=\frac{1}{N}$\\
\nonl At time $n>1$\\
Sample $X_n^i\sim K(\widetilde{X}_{n-1}^i, \cdot)$ and $Y_n^i$ uniformly from $\mathbf{Y}$ for $i=1,\dots,N$\\
Compute the approximated potentials $\weightN(X_n^i, Y_n^i)$ in~\eqref{eq:potentialN} and obtain the normalized weights
$
W_n^i = {\weightN(X_n^i, Y_n^i)} \big/ {\sum_{j=1}^N \weightN(X_n^j, Y_n^j)}
$\\
(Re)Sample $\left\lbrace(X_{n}^i, Y_{n}^i), W_n^i\right\rbrace$ to get $\left\lbrace(\widetilde{X}_{n}^i,\widetilde{Y}_{n}^i), \frac{1}{N}\right\rbrace$ for $i=1, \ldots, N$\\
Estimate $f_{n+1}(x)$ as in~\eqref{eq:smc_kde1}
\end{algorithm}

The use of $\weightN$ within the importance weighting step corresponds
to an additional approximation which is not found in standard SMC
algorithms.
In particular,~\eqref{eq:potentialN} are biased estimators of the true potentials~\eqref{eq:potential}. As a consequence, it is not possible to use arguments based on extensions of the state space (as in particle filters using unbiased estimates of the potentials \citep{liu1998sequential,delmoraldoucetjasraBayescomp2006,fearnhead2008particle}) to provide theoretical guarantees for this SMC scheme.
If $\weight$ itself were available then it would be
preferable to make use of it; in practice this will never be the case but the
idealized algorithm which employs such a strategy is of use for
theoretical analysis.

At time $n+1$, we estimate $f_{n+1}(x)$ by computing a kernel density estimate (KDE) of the weighted particle approximation
\begin{align*}
\bgN(\predictiveN)\vert_{\X} := \sum_{i=1}^N\frac{ \weightN(X_n^i, Y_n^i)}{\sum_{j=1}^N \weightN(X_n^j, Y_n^j)}\delta_{X_n^i},
\end{align*}
and \emph{then} applying the EMS smoothing kernel $K$.
This approach may seem counter-intuitive but the KDE kernel and the EMS kernel are fulfilling different roles. The KDE gives a good smooth approximation of the density associated with the EMS recursion at a point in that recursion which we expect to be under-smoothed and is driven by the usual considerations of KDE when obtaining a smooth density approximation from an empirical distribution; going on to apply the EMS smoothing kernel is simply part of the EMS regularization procedure. One could instead apply kernel density estimation after step 2 of the subsequent iteration of the algorithm but this would simply introduce additional Monte Carlo variance, with the described approach corresponding to a Rao-Blackwellisation of that slightly simpler strategy. Using the kernel of Fredholm equations of the \emph{second} kind to extract smooth approximations of their solution from Monte Carlo samples has also been found empirically to perform well \citep{doucet2010solving}. Depending on the intended use of the approximation, the KDE step can be omitted entirely; the empirical distribution provides a good (in the sense of Proposition~\ref{prop:asw}) approximation to that given by the EMS recursion but one which does not admit a density.

We consider standard $d_{\X}$-dimensional kernels for KDE, $\bwN^{-d_{\X}}\vert \Sigma\vert^{-1/2} S\left(\left(\bwN^2\Sigma \right)^{-1/2}u\right)$,
where $\bwN$ is the smoothing bandwidth and $S$ is a continuous bounded symmetric density \citep{silverman1986density}.
To account for the dependence between samples, when computing the bandwidth, $\bwN$, instead of $N$ we use the effective sample size \citep{kong1994sequential}
\begin{equation}
\label{eq:ess}
\ESS = \left(\sum_{i=1}^N \weightN(X_n^i, Y_n^i)^2\right)^{-1} \left(\sum_{j=1}^N \weightN(X_n^j, Y_n^j)\right)^2 .
\end{equation}
The resulting estimator,
\begin{equation}
\label{eq:smc_kde1}
f^N_{n+1}(x) = \int_{\X}K(x^\prime, x)\sum_{i=1}^N\frac{ \weightN(X_n^i, Y_n^i)}{\sum_{j=1}^N \weightN(X_n^j, Y_n^j)}\bwN^{-d_{\X}}\vert \Sigma\vert^{-1/2} S\left(\left(\bwN^2\Sigma \right)^{-1/2}(X_n^i - x^\prime)\right)\ \dx^\prime,
\end{equation}
satisfies the standard KDE convergence results in $\mathbb{L}_1$ and in $\mathbb{L}_2$ (see Section~\ref{sec:convergence_kde}). 

As the EMS recursion~\eqref{eq:ems} aims at finding a fixed point, after a certain number of iterations the approximation of $f$ provided by the SMC scheme stabilizes. We could therefore average over approximations obtained at different iterations to get more stable reconstructions. When the storage cost is prohibitive, a thinned set of iterations could be used.

In principle, one could reduce the variance of associated estimators by using a different proposal distribution within Algorithm~\ref{alg:fpsmc} just as in standard particle methods (see, e.g., \citet[Section 25.4.1]{doucet2011tutorial}) but this proved unnecessary in all of the examples which we explored as we obtained good performances with this simple generic scheme (the effective sample size was above $70\%$ in 
all the examples considered). Another strategy to reduce the variance of the estimators would be to implement the quasi-Monte Carlo
version of SMC \citep{gerber2015sequential} which is particularly efficient in the relatively low-dimensional settings typically found in the context of Fredholm equations.

\subsubsection{Algorithmic Setting}
Algorithm~\ref{alg:fpsmc} requires specification of a number of parameters.
The initial density, $f_1$, must be specified but we did not find performance to be sensitive to this choice (see Appendix~\ref{sec:analytically_tractable}). We advocate choosing $f_1$ to be a diffuse distribution with support intended to include that of $f$ because the resampling step allows SMC to more quickly forget overly diffuse initializations than overly concentrated ones. For problems with bounded domains, choosing $f_1$ to be uniform over \X is a sensible default choice.

We propose to stop the iteration in Algorithm~\ref{alg:fpsmc} when the difference between successive approximations, measured through the $\mathbb{L}_2$ norm of
the reconstruction of $h$ obtained by convolution of $f^N_{n}$ with $g$, $\hat{h}^N_{n}(y) := \int_{\X}f^N_{n}(x)g(y \mid x) \dx$, is smaller than the variability due to the Monte Carlo approximation of~\eqref{eq:ems}
\begin{align}
    \label{eq:stop}
     \int_{\Y}\left\lvert\hat{h}^N_{n+1}(y)-\hat{h}^N_{n}(y)\right\rvert^2\dy < \textrm{var}\left(\zeta(f_{k}^N); k = n+1-m, \dots n+1 \right),
\end{align}
where $\zeta$ is some function of the estimator $f_{n+1}^N$ and we consider its variance over the last $m$ iterations. 
The term on the left-hand side is an indicator of whether the EMS recursion~\eqref{eq:ems} has reached a fixed point, while the variance takes into account the error introduced by approximating~\eqref{eq:ems} through Monte Carlo. For given $N$ there is a point at which further increasing $n$ does not improve the estimate because Monte Carlo variability dominates. 
We employ this stopping rule in the PET example in Section~\ref{sec:pet}.

The amount of regularization introduced by the smoothing step is controlled by the smoothing kernel $K$. In principle, any density $T$ can be used to specify $K$ as in~\ref{a:k}; we opted for isotropic Gaussian kernels since in this case the integral in~\eqref{eq:smc_kde1} can be computed analytically with an appropriate choice of $S$.
In this case, the amount of smoothing is controlled by the variance $\varepsilon^2$. If the expected smoothness of the fixed point of the EMS recursion~\eqref{eq:ems} is known, $\varepsilon$ should be chosen so that~\eqref{eq:smc_kde1} matches this knowledge. If no information is known on the expected smoothness, the level of smoothing introduced could be picked by cross validation, comparing, e.g., the reconstruction accuracy or smoothness.
In addition, one could allow extra flexibility by letting $K$ change at each iteration: e.g., allowing larger moves in early iterations can be beneficial in standard SMC settings to improve stability and ergodicity; alternatively one could choose the smoothing parameter adaptively using information on the smoothness of the current estimate.

We end this section by identifying a further degree of freedom which can be exploited to improve performance: a variance reduction can be achieved by averaging over several $Y_n^i$ when computing the approximated potentials $\weightN$. At time $n$, draw $M$ samples $Y_n^{ij}$, $j=1, \ldots, M$ without replacement for each particle $i=1, \ldots, N$ and compute the approximated potentials by averaging over the $M$ replicates
\begin{equation*}
G_n^{N, M}(X_n^i, Y_n^i) = \frac1M \sum_{j=1}^M \frac{g(Y_n^{ij}\mid X_{n}^i)}{h^N_n(Y_n^{ij})}.
\end{equation*}
This incurs an $O(MN)$ computational cost and can be justified by further extending the state space to $\X\times \Y^M$.
Unfortunately, the results on the optimal choice of $M$ obtained for pseudo-marginal methods (e.g. \cite{pitt2012some}) cannot be applied here, as the estimates of $\weight$ given by~\eqref{eq:potentialN} are \emph{not} unbiased.
In the examples shown in Section~\ref{sec:examples} we resample without replacement $M$ samples from $\mathbf{Y}$ where $M$ is the smallest between $N$ and the size of $\mathbf{Y}$, but smaller values of $M$ could be considered (see Appendix~\ref{sec:analytically_tractable}).

\subsubsection{Comparison with EMS}
The discretized EMS~\eqref{eq:ems_discrete} and Algorithm~\ref{alg:fpsmc} both approximate the EMS recursion~\eqref{eq:ems}.
There are two main aspects under which the SMC implementation of EMS is an improvement with respect to the one obtained by brute-force discretization: the information on $h$ which is needed to run the algorithm and the scaling with the dimensionality of the domain of $f$.

The discretized EMS~\eqref{eq:ems_discrete} requires the value of $h$ on each of the $D$ bins of the space discretization of $\Y$, when we only have a sample $\mathbf{Y}$ from $h$, as it is the case in most applications \citep{delaigle2008alternative, miao2018identifying, goldstein2009reconstructing, gostic2020practical, marschner2020back, hall2005nonparametric}, $h$ should then be approximated through a histogram or a kernel density estimator as in the Iterative Bayes algorithm \citep{ma2011indirect}. On the contrary, Algorithm~\ref{alg:fpsmc} does not require this additional approximation and naturally deals with samples from $h$.
In Section~\ref{sec:indirect_density_estimation} we show on a one dimensional example that the brute-force discretization~\eqref{eq:ems_discrete} struggles at recovering the shape of a bimodal distribution while the SMC implementation achieves much better performances in terms of accuracy. In addition, increasing the number of bins for EMS has a milder effect on the accuracy than increasing the number of particles in the SMC implementation.

Similar considerations apply when $\X, \Y$ are higher dimensional (i.e. $d_{\X}\geq 2$). The number of bins $B$ in the EMS recursion~\eqref{eq:ems_discrete} necessary to achieve reasonable accuracy increases exponentially with $d_{\X}$, resulting in higher runtime which quickly exceed those needed to run Algorithm~\ref{alg:fpsmc}. On the contrary the convergence rate for SMC remains $N^{-1/2}$, and although the associated constants may grow with $d_{\X}$, its performance is shown to scale better with dimension than EMS in Appendix~\ref{app:pdim}.

\subsection{Convergence properties}
\label{sec:results}

As the potentials~\eqref{eq:potential} cannot be computed exactly but need to be estimated, the convergence results for standard SMC (e.g., \citet{del2013mean}) do not hold. We present here a strong law of large numbers (SLLN) and $\lp$ error estimates for our particle approximation of the EMS and also provide theoretical guarantees for the estimator~\eqref{eq:smc_kde1}.

\subsubsection{Strong law of large numbers}
\label{sec:slln}
For simplicity, we only consider multinomial resampling \citep{gordon1993novel}. Lower variance resampling schemes can be employed but considerably complicate the theoretical analysis (\citet[page 336]{douc2014nonlinear}, \cite{ gerber2019negative}).
Compared to the SLLN proof for standard SMC methods, we need to analyze here the contribution of the additional approximation introduced by using $\weightN$ instead of $\weight$ and then combine the results with existing arguments for standard SMC; see, e.g., \cite{miguez2013convergence}.

The SSLN is stated in Corollary~\ref{prop:slln}. This result follows from the $\lp$ inequality in Proposition~\ref{prop:lp}, the proof of which is given in Appendix~\ref{app:lp} and follows the inductive argument of \citet{crisan2002survey, miguez2013convergence}.
Both results are proved for bounded measurable test functions $\testfn$, a set we denote $\measurable(\HH)$.

As a consequence of~\ref{a:g}, the potentials $\weight$ and $\weightN$ are bounded and bounded away from 0 (see Lemma~\ref{lemma:potential} in Appendix~\ref{app:lp}), a strong mixing condition that is common in the SMC literature and is satisfied in most of the applications which we have considered. 

\begin{prop}[$\lp$-inequality]
\label{prop:lp}
Under~\ref{a:space}, \ref{a:g} and~\ref{a:k}, for every $n\geq 1$ and every $p\geq 1$ there exist finite constants $\widehat{C}_{p, n}, \widetilde{C}_{p,n}$ such that
\begin{align}
\label{eq:lp4kde}
  \Exp\left[\vert\bgN(\predictiveN)(\testfn) - \bg(\predictive)(\testfn)\vert^p\right]^{1/p} \leq& \widehat{C}_{p,n}\frac{\supnorm{\testfn}}{\sqrt{N}}\\
  \textrm{ and } \qquad\qquad
\label{eq:lpthm}
\Exp\left[\vert\predictiveN(\testfn) - \predictive(\testfn)\vert^p\right]^{1/p} \leq& \widetilde{C}_{p,n}\frac{\supnorm{\testfn}}{\sqrt{N}},
\end{align}
for every bounded measurable function $\testfn\in \measurable(\HH)$, where the expectations are taken with respect to the law of all random variables generated within the SMC algorithm.
\end{prop}

The SLLN follows from the $\lp$-inequality using a standard Borel-Cantelli argument (see, e.g. \citet[Appendix D]{boustati2020generalised} for a reference in the context of SMC):
\begin{corollary}[Strong law of large numbers]
\label{prop:slln}
Under~\ref{a:space}, \ref{a:g} and~\ref{a:k}, for all $n\geq1$ and for every $ \testfn\in \measurable(\HH)$, we have almost surely as $N\to\infty$:
\begin{align*}
\bgN(\predictiveN)(\testfn) \rightarrow \bg(\predictive)(\testfn)
  \qquad \textrm{ and }  \qquad
 \predictiveN(\testfn) \rightarrow \predictive(\testfn).
\end{align*}
\end{corollary}
A standard approach detailed in Appendix~\ref{app:wcm} yields convergence of the sequence $\lbrace\eta_n^N\rbrace_{n\geq 1}$ itself, showing that the particle approximations of the distributions converge to the sequence in~\eqref{eq:fkflow}, whose marginal over $x$ satisfies the EMS recursion~\eqref{eq:ems}.
\begin{prop}
\label{prop:asw}
Under~\ref{a:space}, \ref{a:g} and~\ref{a:k}, for all $n\geq1$, $\predictiveN$ converges weakly to $\predictive$ with probability 1.
\end{prop}

\subsubsection{Convergence of kernel density estimator}
\label{sec:convergence_kde}

Under standard assumptions on the bandwidth $\bwN$ we can show that the estimator $f_{n+1}^N(x)$ converges in $\mathbb{L}_1$ to $f_{n+1}(x)$ and its mean integrated square error (MISE) goes to 0 as $N$ goes to infinity as shown in Appendix~\ref{sec:kde_misepf}: 
\begin{prop}\label{cor:kde_mise}
Under~\ref{a:space}, \ref{a:g} and~\ref{a:k}, if $\bwN\rightarrow0$ as $N\rightarrow \infty$, $f^N_{n+1}$ converges almost surely to $f_{n+1}$ in $\mathbb{L}_1$ for every $n\geq 1$:
\begin{equation}
\label{eq:kde_l11}
\lim_{N\to\infty} \int_{\X} \vert f_{n+1}^N(x) - f_{n+1}(x)\vert \dx \overset{\textrm{a.s.}}{=} 0;
\end{equation}
and the MISE satisfies
\begin{equation}
\label{eq:kde_mise1}
\lim_{N\to\infty} \mise(f_{n+1}^N) \equiv \lim_{N\to\infty} \Exp\left[\int_{\X} \vert f_{n+1}^N(x) - f_{n+1}(x)\vert^2\ \dx\right] {=} 0.
\end{equation}
\end{prop}

\section{Examples}
\label{sec:examples}

This section shows the results obtained using the SMC implementation of the recursive formula~\eqref{eq:ems} on some common examples. Two additional examples are investigated in Appendix~\ref{sec:additional_examples}.
We consider a simple density estimation problem and a realistic example of image restoration in positron emission tomography \citep{webb2017introduction}. In the first example, the analytic form of $h$ is known and is used to implement the discretized EM and EMS. IB and SMC are implemented using a fixed sample $\mathbf{Y}$ drawn from $h$.
For image restoration problems we consider the observed distorted image as the empirical distribution of a sample $\mathbf{Y}$ from $h$ and resample from it at each iteration of line 2 in Algorithm~\ref{alg:fpsmc}.

The initial distribution $f_1$ is uniform over $\X$ and the number of iterations is either fixed to $n=100$ (we observed that convergence to a fixed point occurs in a smaller number of steps for all algorithms; see Appendix~\ref{sec:analytically_tractable}) or determined using the stopping criterion~\eqref{eq:stop}.
For the smoothing kernel $K$, we use isotropic Gaussian kernels with marginal variance $\varepsilon^2$. The bandwidth $\bwN$ is the plug-in optimal bandwidth for Gaussian distributions where the effective sample size~\eqref{eq:ess} is used instead of the sample size $N$ \citep[page 45]{silverman1986density}.

The deterministic discretization of EM and EMS (\eqref{eq:em_discrete} and~\eqref{eq:ems_discrete} respectively) is obtained by considering $B$ equally spaced bins for $\X$ and $D$ for $\Y$. The number of bins, and the number of particles, $N$, for SMC vary between examples.
In the first example, the choice of $D, B$ and $N$ is motivated by a comparison of error and runtime.
For the image restoration problems, $D, B$ are the number of pixels in each image, while
the number of particles $N$ is chosen to achieve a good trade-off between reconstruction accuracy and runtime.

For the SMC implementation, we use the adaptive multinomial resampling scheme described by \citet[page 35]{liu2001monte}. At each iteration the effective sample size~\eqref{eq:ess} is evaluated and multinomial resampling is performed if $\ESS < N/2 $.
This choice is motivated by the fact that up to adaptivity (which we anticipate could be addressed by the approach of \cite{delmoral2012adaptive}) this is the setting considered in the theoretical analysis of Section~\ref{sec:results} and we observed only modest improvements when using lower variance resampling schemes (e.g. residual resampling, see \citet{liu2001monte}) instead of multinomial resampling.
The accuracy of the reconstructions is measured through the integrated square error
\begin{equation}
\label{eq:ise}
\ise(f_{n+1}^N) = \int_{\X} \left( f(x) - f^N_{n+1}(x)\right)^2 \ \dx.
\end{equation}

Although the density estimation example of Section~\ref{sec:indirect_density_estimation} and the example considered in  Appendix~\ref{sec:analytically_tractable} do not satisfy conditions~\ref{a:space} or 
\ref{a:g} under which our theoretical guarantees hold; we nonetheless observe good results in terms of reconstruction accuracy and smoothness, demonstrating that assumption~\ref{a:g} is not necessary and could be relaxed (see also Appendix~\ref{app:lb}). The other examples \emph{do} satisfy all of our theoretical assumptions.
\subsection{Indirect density estimation}
\label{sec:indirect_density_estimation}
The first example is the Gaussian mixture model used in \cite{ma2011indirect} to compare the Iterative Bayes (IB) algorithm with EM. Take $\X = \Y = \real$ (although note that $|1 - \int_0^1 f(x) \dx| < 10^{-30}$ and restricting out attention to $[0,1]$ would not significantly alter the results) and
\begin{align*}
& f(x)    = \frac{1}{3}\N(0.3, 0.015^2) + \frac{2}{3}\N(0.5, 0.043^2),\\
& g(y \mid x)  = \N(x, 0.045^2),\\
& h(y)    = \frac{1}{3}\N(0.3, 0.045^2 + 0.015^2) + \frac{2}{3}\N(0.5, 0.045^2 + 0.043^2).
\end{align*}
The initial distribution $f_1$ is Uniform on $[0,1]$ and the bins for the discretized EMS are $B$ equally spaced intervals in $[0,1]$, noting that discretization schemes essentially require known compact support and this interval contains almost all of the probability mass.
We run Algorithm~\ref{alg:fpsmc} assuming that we have a sample $\mathbf{Y}$ of size $10^3$ from $h$ from which we re-sample $M=\min(N, 10^3)$ times without replacement at each iteration of line 2.
We analyze the influence of the number of bins $B$ and of the number of particles $N$ on the integrated square error and on the runtime for the deterministic discretization of EMS~\eqref{eq:ems_discrete} and for the SMC implementation of EMS (Figure~\ref{fig:gaussian_mixture_runtimes}). We compare the two implementations of EMS with a class of estimators for deconvolution problems, deconvolution kernel density estimators with cross validated bandwidth (DKDE-cv;  \cite{stefanski1990deconvolving}) and plug-in bandwidth (DKDE-pi; \cite{delaigle2004practical})\footnote{MATLAB code is available on the authors' web page: \url{https://researchers.ms.unimelb.edu.au/~aurored/links.html\#Code}}. These estimators take as input a sample from $h$ of size $N$ and output a kernel density estimator for $f$. 

\begin{figure}[t]
\centering
\resizebox{0.6\textwidth}{!}{%
\begin{tikzpicture}[every node/.append style={font=\normalsize}]
\node (img1) {\includegraphics[width=0.67\textwidth]{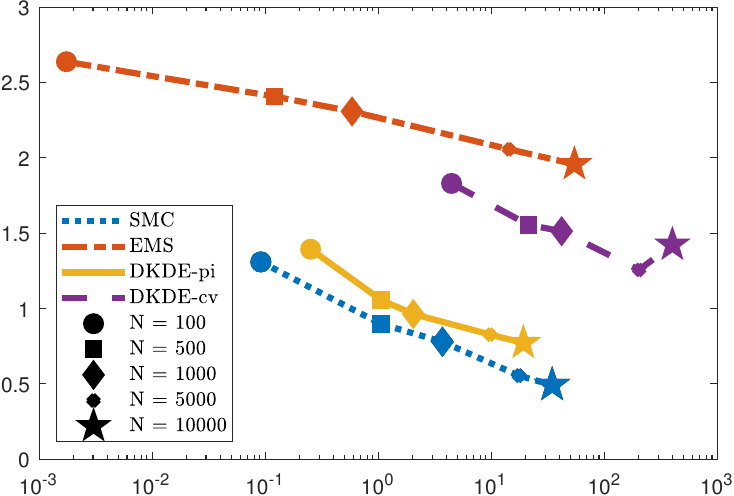}};
\node[below=of img1, node distance = 0, yshift = 1cm] { Runtime / s};
  \node[left=of img1, node distance = 0, rotate=90, anchor = center, yshift = -0.7cm] {$\ise(f_{n+1}^N)$};
\end{tikzpicture}
}
\caption{Average $\ise(f_{n+1}^N)$ and runtime for 1,000 repetitions of discretized EMS, SMC and DKDE. The number of bins and of particles/samples $N$ ranges between $10^2$ and $10^4$.}
\label{fig:gaussian_mixture_runtimes}
\end{figure}

The discretized EMS has the lowest runtime for fixed $N$, however $\ise(f_{n+1}^N)$ is the highest and finer discretizations for EMS do not significantly improve accuracy.
The runtime of DKDE are closer to those of the SMC implementation, however, the SMC implementation gives better results in terms of $\ise(f_{n+1}^N)$ for any particle size and, indeed, for given computational cost.
We set $\varepsilon=10^{-3}$, for both EMS and SMC, somewhat arbitrarily, based on the support of the target in this example; where that is not possible cross validation could be used --- and might be expected to provide better reconstructions --- at the expense of some additional computational cost. We did not find solutions overly sensitive to the precise value of $\varepsilon$ (see Appendix~\ref{sec:analytically_tractable}). A significant portion of the runtime of DKDE-cv is needed to obtain the bandwidth through cross validation and in this sense the comparison may not be quite fair, but the use of the much cheaper plug-in estimates of bandwidth within DKDE-pi also led to poorer estimates at given cost than those provided by the SMC-EMS algorithm.

Secondly, we compare the reconstructions provided by the proposed SMC scheme with those given by deterministic discretization of the EM iteration~\eqref{eq:em_discrete} with exact $h$ and when only samples are available (IB) and deterministic discretization of the EMS iteration~\eqref{eq:ems_discrete}.

Having observed a small decrease in $\ise(f_{n+1}^N)$ for large $B$, we fix the number of bins $B=D=100$.
For the SMC scheme, we compare $N = 500$, $N = 1,000$ and $N = 5,000$. We discard $N=10,000$, as it shows little improvement in $\ise(f_{n+1}^N)$ with respect to $N=5,000$, and $N=100$, because of the higher $\ise(f_{n+1}^N)$.
We draw a sample $\mathbf{Y}$ from $h$ of size $10^3$ and we use this sample to get a kernel density estimator for the IB algorithm, compute the DKDE and (re)sample points at line 2 of Algorithm~\ref{alg:fpsmc}.

We set $\varepsilon = 10^{-3}$ and compare the smoothing matrix obtained by discretization of the Gaussian kernel (EMS ($\Kop$)) with the three-point smoothing proposed in \citet[Section 3.2.2]{silverman1990smoothed}, where the value $f_b^{(n+1)}$ is obtained by a weighted average over the values $f_{\kappa}^{(n)}$ of the two nearest neighbors (the third point is $f_b^{(n)}$), with weights proportional to the distance $\vert \kappa - b\vert$
\begin{align*}
\Kop_{b\kappa} = 2^{-2}\binom{3 - 1}{\kappa - b + (3-1)/2}.
\end{align*}

The reconstruction process is repeated 1,000 times and the reconstructions are compared by computing their means and variances, the integrated squared error~\eqref{eq:ise} and the Kullback--Leibler divergence between $h$ and the reconstruction of $h$ obtained by convolution of $f^N_{n+1}$ with $g$, $\int_{\X}f^N_{n+1}(x)g(y \mid x)\ \dx$, (Table~\ref{tab:gaussian_mixture}).
To characterize the roughness of $f^N_{n+1}$, we evaluate both $f^N_{n+1}$ and $f$ at the 100 bin centers $x_c$ and for each bin center we approximate (with 1,000 replicates) the mean squared error (MSE)
\begin{equation}
\label{eq:mse}
\mse(x_c) = \Exp\left[\left( f(x_c) - f_{n+1}^N(x_c)\right)^2\right].
\end{equation}
Table~\ref{tab:gaussian_mixture} shows the 95th percentile w.r.t. the 100 bin centers $x_c$.

\begin{table}[t]
\centering
\caption{Estimates of mean, variance, ISE, 95th-percentile of MSE, KL-divergence and runtime for 1,000 repetitions of EM, EMS, IB, SMC and DKDE for the Gaussian mixture example. The mean of $f$ is 0.43333, the variance is 0.010196. \textbf{Bold} indicates best values.}
\footnotesize{
\begin{tabular}{lccccccc}
\multirow{2}{*} &Mean & Variance & $\ise(f_{n+1}^N)$ & $\mse(x_c)$  & $\KL$ & $\log_{10}$ \\
& & & & (95th) & & Runtime / s\\
\hline\noalign{\smallskip}
EM &0.36667&\textbf{0.010}&3.26&16.32&\textbf{2299}&\textbf{-6.01}\\
EMS ($\Kop$) &0.36646&0.012&2.41&8.20&2355&-5.90\\
EMS (3-point) &0.3668&0.011&1.58&13.04&2303&-5.88\\
IB &\textbf{0.43304}&0.011&1.71&10.17&2489&-5.29\\
SMC (500) &0.43303&0.011&0.90&3.42&2484&0.63\\
SMC (1000) &0.43302&0.011&0.78&3.33&2483&1.87\\
SMC (5000) &0.43302&0.011&\textbf{0.55}&\textbf{2.17}&2485&3.47\\
DKDE-pi &0.43288&0.012&0.96&3.38&2483&0.81\\
DKDE-cv &0.43287&0.014&1.5&4.76&2503&4.15\\
\end{tabular}
}
\label{tab:gaussian_mixture}
\end{table}

The discretized EM~\eqref{eq:em_discrete} gives the best results in terms of Kullback--Leibler divergence (restricting to the $[0,1]$ interval and computing by numerical integration). This is not surprising, as IB is an approximation of EM when the analytic form of $h$ is not known, and the EMS algorithms (both those with the deterministic discretization~\eqref{eq:ems_discrete} and those with the stochastic one given by the SMC scheme) do not seek to minimize the $\KL$ divergence, but to provide a more regular solution.
The solutions recovered by EM and IB have considerably higher $\ise$ than that given by the other algorithms and are considerably worse than the other algorithms at recovering the smoothness of the solution.

SMC is generally better at recovering the global shape of the solution ($\ise$ is at least two times smaller than EM and EMS ($\Kop$) and about half than EMS (3-point) and IB) and the smoothness of the solution (the 95th-percentile for $\mse(x_c)$ is at least two times smaller).
For the discretized EMS~\eqref{eq:ems_discrete} and the SMC implementation the estimates of the variance are higher than those of EM, this is a consequence of the addition of the smoothing step and can be controlled by selecting smaller values of $\varepsilon$.
DKDEs behave similarly to SMC, however their reconstruction accuracy and smoothness are slightly worse than those of SMC (even when both algorithms use the same sample size $N=1,000$). In particular, DKDE-cv has runtime of the same order of that of SMC but achieves considerably worse results.
IB, SMC and DKDE give similar values for the $\KL$ divergence.
The slight increase observed for the SMC scheme with $N=5,000$ is apparently due to the sensitivity of this divergence to tail behaviors; taking a bandwidth independent of $N$ eliminated this effect (results not shown).

\subsection{Positron emission tomography}
\label{sec:pet}

Positron Emission Tomography (PET) is a medical diagnosis technique used to analyze internal biological processes from radial projections to detect medical conditions such as schizophrenia, cancer, Alzheimer's disease and coronary artery disease \citep{phelps2000positron}.

The data distribution of the radial projections $h(\phi, \xi)$ is defined on $\Y=[0, 2\pi]\times[-R, R]$ for $R>0$ and is linked to the cross-section image of the organ of interest $f(x, y)$ defined on the 2D square $\X=[-r, r]^2$ for $r>0$ through the kernel $g$ describing the geometry of the PET scanner.
The Markov kernel $g(\phi, \xi \mid x, y)$ gives the probability that the projection onto $(\phi, \xi)$ corresponds to point $(x, y)$ \citep{vardi1985statistical} and is modeled as a zero-mean Gaussian distribution with small variance (we use $\sigma^2 = 0.02^2$) to mimic the alignment between projections and recovered emissions (see Appendix~\ref{app:pet}). As $g$ is defined on $\X\times \Y$ where $\X=[-r, r]^2$ and $\Y=[0, 2\pi]\times[-R, R]$, assumption~\ref{a:g} is satisfied. 

The data used in this work are obtained from the reference image in the final panel of Figure~\ref{fig:pet_reconstruction}, a simplified imitation of the brain's metabolic activity (e.g. \cite{vardi1993image}).
The collected data are the values of $h$ at 128 evenly spaced projections over $360^\circ$ and 185 values of $\xi$ in $[-92, 92]$ to which Poisson noise is added.
Figure~\ref{fig:pet_reconstruction} shows the reconstructions obtained with the SMC scheme with smoothing parameter $\varepsilon = 10^{-3}$ and number of particles is $N= 20,000$. Convergence to a fixed point occurs in less than 100 iterations, in fact the criterion~\eqref{eq:stop} with $\zeta(f^N_n)=\int_{\X}\vert f_{n}^N(x)\vert^2\dx$  and $m=15$ stops the iteration at $n=15$. The $\ise$ between the original image and the reconstructions stabilizes around 0.08. Additional results and model details are given in the supplementary material.

The results above show that the SMC implementation of the EMS recursion achieves convergence in a small number of steps ($\approx$ 12 minutes on a standard laptop) and that, contrary to EM \citep[Section 4.2]{silverman1990smoothed}, these reconstructions are smooth and do not deteriorate with the number of iterations.
In addition, contrary to standard reconstruction methods, e.g. filtered back-projection, ordered-subset EM, Tikhonov regularization (see, e.g., \cite{tong2010image}) the SMC implementation does not require that a discretization grid is fixed in advance.

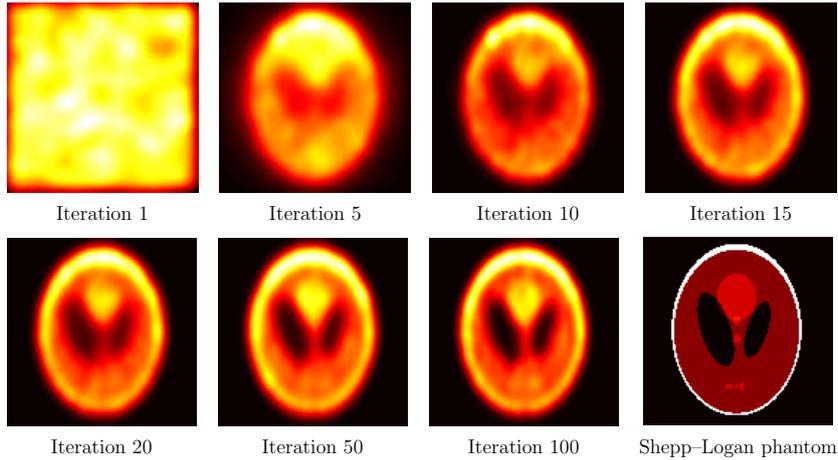
\begin{figure}
\centering
\resizebox{0.7\textwidth}{!}{%
\begin{tikzpicture}[baseline, every node/.append style={font=\normalsize}]
\node (img1) {\includegraphics[width=0.25\textwidth]{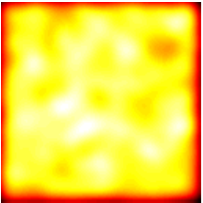}};
\node[below=of img1, node distance = 0, yshift = 1cm] (label1){Iteration 1};
\end{tikzpicture}
\begin{tikzpicture}[baseline, every node/.append style={font=\normalsize}]
\node[right=of img1] (img2) {\includegraphics[width=0.25\textwidth]{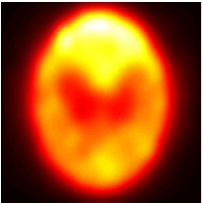}};
\node[below=of img2, node distance = 0, yshift = 1cm] {Iteration 5};
\end{tikzpicture}
\begin{tikzpicture}[baseline, every node/.append style={font=\normalsize}]
\node[right=of img2] (img3) {\includegraphics[width=0.25\textwidth]{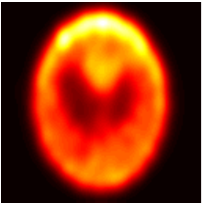}};
\node[below=of img3, node distance = 0, yshift = 1cm] {Iteration 10};
\end{tikzpicture}
\begin{tikzpicture}[baseline, every node/.append style={font=\normalsize}]
\node[right=of img3] (img4) {\includegraphics[width=0.25\textwidth]{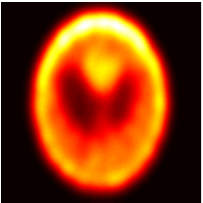}};
\node[below=of img4, node distance = 0, yshift = 1cm] {Iteration 15};
\end{tikzpicture}
}
\resizebox{0.7\textwidth}{!}{%
\begin{tikzpicture}[baseline, every node/.append style={font=\normalsize}]
\node (img5) {\includegraphics[width=0.25\textwidth]{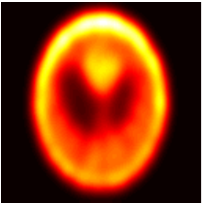}};
\node[below=of img5, node distance = 0, yshift = 1cm] {Iteration 20};
\end{tikzpicture}
\begin{tikzpicture}[baseline, every node/.append style={font=\normalsize}]
\node[right=of img5] (img6) {\includegraphics[width=0.25\textwidth]{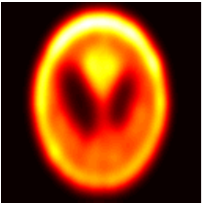}};
\node[below=of img6, node distance = 0, yshift = 1cm] {Iteration 50};
\end{tikzpicture}
\begin{tikzpicture}[baseline, every node/.append style={font=\normalsize}]
\node[right=of img6] (img7) {\includegraphics[width=0.25\textwidth]{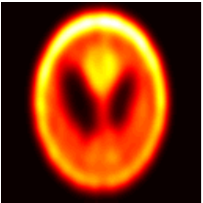}};
\node[below=of img7, node distance = 0, yshift = 1cm] {Iteration 100};
\end{tikzpicture}
\begin{tikzpicture}[baseline, every node/.append style={font=\normalsize}]
\node[right=of img7] (img8) {\includegraphics[width=0.25\textwidth]{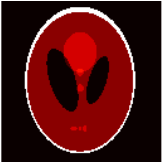}};
\node[below=of img8, node distance = 0, yshift = 1cm] {Shepp--Logan phantom};
\end{tikzpicture}
}
\caption{Reconstruction of the Shepp--Logan phantom with $N = 20,000$ particles and $\varepsilon=0.001$. The stopping criterion~\eqref{eq:stop} is satisfied at iteration 15.}
\label{fig:pet_reconstruction}
\end{figure}

\section{Conclusion}
We have proposed a novel particle algorithm to solve a wide class of Fredholm equations of the first kind.  This algorithm has been obtained by identifying a close connection between the continuous EMS recursion and the dynamics~\eqref{eq:smc}. It performs a stochastic discretization of the EMS recursion and can be naturally implemented when only samples from the distorted signal $h$ are available. Additionally, it does not require the assumption of piecewise constant solutions common to deterministic discretization schemes.

Having established that the continuous EMS recursion admits a fixed point, we have studied the asymptotic properties of the proposed particle scheme, showing that the empirical measures obtained by this scheme converge almost surely in the weak topology to those given by the EMS recursion as the number of particles $N$ goes to infinity. This result is a consequence of the $\lp$ convergence of expectations and the strong law of large numbers which we extended to the particle scheme under study. We have also provided theoretical guarantees on the proposed estimator for the solution $f$ of the Fredholm integral equation.
This algorithm outperforms the state of the art in this area in several examples.

\section*{Supplementary Material}
The supplementary material contains the analysis of the EMS map, proofs of all results and additional examples.
MATLAB code to reproduce all examples is available online at \if1\version{\url{https://anonymous.4open.science/r/31a22d2c-310f-48db-98a8-2e5744da40cb/}}
\else
{\url{https://github.com/FrancescaCrucinio/smcems}}\fi.

\bibliographystyle{agsm}
\bibliography{smcfe_biblio}

\newpage
\appendix
{\Large \center{\ourtitle:\\Supplementary Material\\}}
\if1\version{}
\else
{
\medskip
\centerline{\large Francesca R. Crucinio, Arnaud Doucet and Adam M. Johansen}
\medskip
}\fi

\section{Notation}
\label{app:notation}
For the convenience of the reader, we summarize the notation adopted in the following arguments. A slightly more technical presentation is adopted than within the main manuscript as a little care is required in order to obtain rigorous results.

We work on a probability space $(\Omega,\mathcal{A},\pr)$ rich enough to allow the definition of the particle system introduced in Section~\ref{sec:mf} for all $N\in\mathbb{N}$. All expectations and probabilities which are not explicitly associated with some other measure are taken with respect to $\pr$.

For any $\HH \subseteq \real^d$, we consider the Borel $\sigma$-algebra $B(\HH)$ with respect to the Euclidean norm, and we endow any product space with the product Borel $\sigma$-algebra.

Let the Banach space of real-valued bounded measurable functions on $\HH$, endowed with the supremum norm, $\supnorm{\testfn} = \sup_{u\in \HH} \left\vert\testfn(u)\right\vert$, be denoted by $\measurable(\HH)$.

Let $\mathcal{M}(\HH)$ be the Banach space of signed finite measures on $(\HH, B(\HH))$ endowed with the bounded Lipschitz norm (e.g. \citetapp[page 394]{dudley2002real})
\begin{equation}\label{eq:defTV}
\beta(\eta) := \sup_{\norm{\testfn}{BL}\leq 1}\left\lvert \int_{\HH} \eta(\dx)\testfn(x)\right\rvert,
\end{equation}
where $\norm{\cdot}{BL}$ denotes the bounded Lipschitz norm for bounded Lipschitz functions $\testfn$
\begin{align*}
    \norm{\testfn}{BL}:=\supnorm{\testfn}+\sup_{x\neq y}\frac{\vert \testfn(x)-\testfn(y)\vert}{\norm{x-y}{2}}.
\end{align*}

For ease of notation, for every measure $\eta\in\mathcal{M}(\HH)$ and every $\testfn\in \measurable(\HH)$ we denote the integral of $\testfn$ with respect to $\eta$ by $\eta(\testfn) := \int_{\HH} \eta(\du)\testfn(u)$.

We denote by $\mathcal{M}^+(\HH)\subset \mathcal{M}(\HH)$ the set of (unsigned) measures of nonzero mass and by $\mathcal{P}(\HH) \subset \mathcal{M}^+(\HH)$ the set of all probability measures on $(\HH, B(\HH))$.
For every $\eta \in \mathcal{P}(\HH)$ we have 
\begin{align*}
  \beta(\eta)=\sup_{\norm{\testfn}{BL}\leq 1}\left\lvert \int_{\HH} \eta(\dx)\testfn(x)\right\rvert  \leq \sup_{\norm{\testfn}{BL}\leq 1} \supnorm{\testfn}\eta(\HH)\leq 1.
\end{align*}

The $\beta$ norm metrizes weak convergence \citep[Theorem 11.3.3]{dudley2002real} in $\mathcal{M}(\X)$: for every $\mu \in \mathcal{M}(\HH)$, and sequence $\lbrace\mu_n\rbrace_{ n\geq 1}$ taking values in $\mathcal{M}(\HH)$, $\beta(\mu_n, \mu)\rightarrow 0$ is equivalent to
$\mu_n(\testfn)\rightarrow \mu(\testfn)$ for all continuous bounded functions $\testfn\in C_b(\HH)$.

For any $\eta\in \mathcal{M}^+(\HH)$ and any positive function $G$ integrable with respect to $\eta$ we denote by $\Psi_G(\eta)(\dx)$ the Boltzmann-Gibbs transform
\begin{equation*}
\Psi_G(\eta)(\dx) = \frac{1}{\eta(G)} G(x)\eta(\dx).
\end{equation*}

A Markov kernel $M$ from $\HH$ to $\HH$ induces two operators. One acts upon measures in $\mathcal{M} (\HH)$ and takes values in $\mathcal{M} (\HH)$ and is defined by
\begin{equation*}
\forall \eta \in \mathcal{M}(\HH) \quad \eta M(\cdot) = \int_{\HH} \eta(\du) M(u,\cdot)
\end{equation*}
and the other acts upon functions in $\measurable(\HH)$ and takes values in $\measurable(\HH)$ and may be defined as
\begin{equation*}
\forall u \in \HH \quad \forall \testfn \in \measurable(\HH) \quad M(\testfn)(u) = \int_{\HH}
M(u,\dv) \testfn(v).
\end{equation*}

For each $\omega\in\Omega$, we obtain a realization of the particle system with $N$ particles at time $n$ and a corresponding random measure denoted by $\predictiveN:\omega\in \Omega \mapsto \predictiveN(\omega)\in\mathcal{P}(\HH)$
\begin{align*}
\predictiveN(\omega)(\cdot) = \frac{1}{N}\sum_{i=1}^N \delta_{(X_n^i(\omega), Y_n^i(\omega))}(\cdot),
\end{align*}
where we suppress from the notation the dependence of $X_n^i(\omega)$ and $Y_n^i(\omega)$ upon $N$, as we shall throughout in the interest of readability.

\section{Existence of the Fixed Point}
\label{app:ems}
Let us formally define the EMS map as a map from the set of unsigned measures of nonzero mass to the set of probability measures, $\emsmap :\mathcal{M}^+(\X) \rightarrow \mathcal{P}(\X)$, such that
\begin{equation*}
\emsmap : \eta \mapsto \emsmap \eta := \int_{\X} \eta(\dx^\prime) K(x^\prime, \cdot)  \int_{\Y} \frac{g(y \mid x^\prime) h(y)}{\int_{\X} \eta(\dz)g(y \mid z)  }\ \dy
\end{equation*}
and the EM map, $\emmap :\mathcal{M}^+(\X) \rightarrow \mathcal{P}(\X)$ as in~\eqref{eq:emop}, slightly more formally as:
\begin{align}
\label{eq:emop2}
  \emmap(\eta)(\dx) &= \frac{1}{\eta(\bar{G}_{\eta})}\bar{G}_{\eta}(x) \eta(\dx),
\end{align}
where the normalizing constant $\eta(\bar{G}_{\eta})\equiv 1$ is introduced to highlight the connection with the particle methods introduced in Section~\ref{sec:mf}.
We introduce the smoothing operator, $\Kop:\mathcal{P}(\X) \rightarrow \mathcal{P}(\X)$, corresponding to the smoothing kernel in~\ref{a:k}
\begin{equation}\label{eq:smoothop}
\Kop: \eta \mapsto \eta K := \int_{\X} \eta(\dv) K(v, \cdot)
\end{equation}
and observe that $\emsmap \eta = \Kop\left( \emmap(\eta)\right) = \left(\emmap\eta\right)K$.

In order to prove that the EMS map admits a fixed point, a number of properties of the EM map, of the smoothing operator $\Kop$ and of the EMS map itself must be established.
We show that $\emsmap$ is a compact operator on $\mathcal{M}^+(\X)$ (Corollary~\ref{cor:compact}). To do so, we show that $\emmap$ is continuous and bounded (Proposition~\ref{prop:emop}) then we prove that $\Kop$ is compact (Proposition~\ref{prop:Kop}).
Compactness is needed to prove existence of a fixed point.

\subsection{Properties of the Continuous EMS Map}
\begin{prop}
\label{prop:emop}
Under~\ref{a:space} and~\ref{a:g}, the EM map $\emmap$ in~\eqref{eq:emop2} is a continuous and bounded operator on $\mathcal{M}^+(\X)$ endowed with the weak topology.
\end{prop}
\begin{proof}

Let $\eta\in\mathcal{M}^+(\X)$ and $\lbrace \eta_n\rbrace_{n\geq 1}$ be a sequence of measures in $\mathcal{M}^+(\X)$ converging to $\eta$ in the weak topology as $n\rightarrow \infty$.
For any $\testfn\in C_b(\X)$ consider
\begin{align*}
&\left\lvert\int_{\X}\emmap(\eta_n)(\dx) \testfn(x) - \int_{\X}\emmap (\eta)(\dx) \testfn(x)\right\rvert \\			
&= \left\lvert \int_{\X}\eta_n(\dx) \testfn(x) \int_{\Y} \frac{g(y \mid x)h(\dy)}{\eta_n\left(g(y \mid \cdot)\right)} - \int_{\X} \eta(\dx)\testfn(x) \int_{\Y} \frac{g(y \mid x)h(\dy)}{\eta\left(g(y \mid \cdot)\right)}\right\rvert\\
&= \left\lvert \int_{\X}\int_{\Y}\testfn(x)g(y \mid x)h(\dy)  \left[\frac{\eta_n(\dx) }{\eta_n\left(g(y \mid \cdot)\right)} - \frac{\eta(\dx)}{\eta\left(g(y \mid \cdot)\right)}\right]\right\rvert\\
&\leq  \left\lvert \int_{\X}\int_{\Y}\frac{\eta_n(\dx)\testfn(x)g(y \mid x)h(\dy)  }{\eta_n\left(g(y \mid \cdot)\right)\eta\left(g(y \mid \cdot)\right)}\left[\eta\left(g(y \mid \cdot)\right) -\eta_n\left(g(y \mid \cdot)\right)\right]\right\rvert\\
& + \left\lvert \int_{\X}\int_{\Y}\left(\eta_n(\dx) - \eta(\dx)\right) \frac{\testfn(x)g(y \mid x)h(\dy) }{\eta\left(g(y \mid \cdot)\right)} \right\rvert,
\end{align*}
where the second equality follows from Fubini's Theorem since $g, \testfn$ are bounded functions.

The first term can be bounded by
\begin{align*}
&\left\lvert \int_{\X}\int_{\Y}\frac{\eta_n(\dx)\testfn(x)g(y \mid x)h(\dy)  }{\eta_n\left(g(y \mid \cdot)\right)\eta\left(g(y \mid \cdot)\right)}\left[\eta\left(g(y \mid \cdot)\right) -\eta_n\left(g(y \mid \cdot)\right)\right]\right\rvert\\
 \leq &  \supnorm{\testfn} \int_{\Y}\frac{h(\dy)\int_{\X}\eta_n(\dx)g(y \mid x)  }{\eta_n\left(g(y \mid \cdot)\right)\eta\left(g(y \mid \cdot)\right)}\left\lvert\eta\left(g(y \mid \cdot)\right) -\eta_n\left(g(y \mid \cdot)\right)\right\rvert\\
 \leq & \supnorm{\testfn} \int_{\Y}\frac{h(\dy)  }{\eta\left(g(y \mid \cdot)\right)}\left\lvert\eta\left(g(y \mid \cdot)\right) -\eta_n\left(g(y \mid \cdot)\right)\right\rvert.
\end{align*}
Under~\ref{a:g}, $g$ is bounded below by $1/m_g$ and we have
\begin{align*}
    \eta\left(g(y \mid \cdot)\right)=\int_{\X}\eta(\dx)g(y \mid x)\geq\frac{1}{m_g}\int_{\X}\eta(\dx)=\frac{1}{m_g}\eta(\X) > 0
\end{align*}
since $\eta\in\mathcal{M}^+(\X)$ is an unsigned measure with nonzero mass. Therefore we obtain
\begin{align*}
\supnorm{\testfn} \int_{\Y}\frac{h(\dy)  }{\eta\left(g(y \mid \cdot)\right)}\left\lvert\eta\left(g(y \mid \cdot)\right) -\eta_n\left(g(y \mid \cdot)\right)\right\rvert  &\leq \supnorm{\testfn}\frac{m_g}{\eta(\X)}\int_{\Y}h(\dy) \left\lvert\eta\left(g(y \mid \cdot)\right) -\eta_n\left(g(y \mid \cdot)\right)\right\rvert 
\end{align*}
For fixed $y$, $g(y \mid \cdot)\in C_b(\X)$, we have that
\begin{align*}
\left\lvert\eta\left(g(y \mid \cdot)\right) - \eta_n\left(g(y \mid \cdot)\right) \right\rvert \rightarrow 0
\end{align*}
as $n\rightarrow \infty$ since $\eta_n$ converges to $\eta$ in the weak topology.
Since $g$ is uniformly bounded by $m_g$, the Dominated Convergence Theorem then gives
\begin{align*}
    \int_{\Y}h(\dy) \left\lvert\eta\left(g(y \mid \cdot)\right) -\eta_n\left(g(y \mid \cdot)\right)\right\rvert  \rightarrow 0
\end{align*}
as $n\rightarrow \infty$, from which we obtain
\begin{align}
\label{eq:emop_1}
\left\lvert \int_{\X}\int_{\Y}\frac{\eta_n(\dx)\testfn(x)g(y \mid x)h(\dy)  }{\eta_n\left(g(y \mid \cdot)\right)\eta\left(g(y \mid \cdot)\right)}\left[\eta\left(g(y \mid \cdot)\right) -\eta_n\left(g(y \mid \cdot)\right)\right]\right\rvert \rightarrow 0
\end{align}
as $n\rightarrow\infty$.

For the second term, consider the function
\begin{align}
\label{eq:emop_function}
x \mapsto \int_{\Y}\frac{\testfn(x)g(y \mid x)h(\dy) }{\eta\left(g(y \mid \cdot)\right)}.
\end{align}
This function is bounded by $m_g^2\supnorm{\testfn}/\eta(\X)$; to see that it is also continuous, recall that $\testfn$, $g$ are continuous functions while the continuity of $y\mapsto\eta\left(g(y \mid \cdot)\right)$ follows from the continuity of $g$ and the Dominated Convergence Theorem.
The Dominated Convergence theorem and the fact that $g$ is continuous, bounded above and below give continuity of~\eqref{eq:emop_function}.  

Using Fubini's Theorem, whose applicability is granted by the boundedness of $g, \testfn$, we obtain
\begin{align}
\label{eq:emop_2}
&\left\lvert \int_{\X}\int_{\Y}\left(\eta_n(\dx) - \eta(\dx)\right) \frac{\testfn(x)g(y \mid x)h(\dy) }{\eta\left(g(y \mid \cdot)\right)} \right\rvert \notag\\
&\qquad \qquad =\left\lvert \int_{\X}\left( \eta_n(\dx) - \eta(\dx) \right)\int_{\Y}\frac{\testfn(x)g(y \mid x)h(\dy) }{\eta\left(g(y \mid \cdot)\right)} \right\rvert \rightarrow 0
\end{align}
as $n\rightarrow \infty$. 

Combining~\eqref{eq:emop_1} and~\eqref{eq:emop_2} we obtain  convergence of $\emmap\eta_n(\testfn)$ to $\emmap\eta(\testfn)$ for every $\testfn\in C_b(\X)$, and thus convergence in the weak topology of $\emmap\eta_n$ to $\emmap\eta$ \citep[Theorem 11.3.3]{dudley2002real} whenever $\eta_n$ converges weakly to $\eta$, proving that the EM map is continuous in $\mathcal{M}^+(\X)$.
Finally, consider boundedness. A non-linear operator is bounded if and only if it maps bounded sets into bounded sets (e.g. \citetapp[page 757]{zeidler1985nonlinear}).
The EM map maps the space of positive finite measures $\mathcal{M}^+(\X)$ into the space of probability measures $\mathcal{P}(\X)$, whose elements have $\beta$ norm uniformly bounded by 1; in particular $\emmap$ maps any bounded subset of $\mathcal{M}^+(\X)$ into a uniformly bounded subset of $\mathcal{P}(\X)$, showing that $\emmap$ is a bounded operator.
\end{proof}

\begin{prop}
\label{prop:Kop}
Under~\ref{a:space} and~\ref{a:k}, the smoothing operator $\Kop$ defined in~\eqref{eq:smoothop} is compact on $\mathcal{P}(\X)$ endowed with the weak topology.
\end{prop}
\begin{proof}
To prove that $\Kop$ is compact we need to prove that it maps bounded subsets into relatively compact subsets \citepapp[Definition 2.17]{kress2014linear}.
It is sufficient to observe that $\X$ is a complete subset of $\real^{d_{\X}}$ (as it is a compact subset of a metric space) from which it follows that $\mathcal{P}(\X)$ is complete by Prokhorov's Theorem (e.g. \citetapp[Corollary 11.5.5]{dudley2002real}) and therefore $\mathcal{P}(\X)$ is relatively compact \citepapp[Theorem 11.5.4]{dudley2002real}.

\end{proof}

\begin{corollary}[Compactness of $\emsmap$]
\label{cor:compact}
Under~\ref{a:space}, \ref{a:g} and~\ref{a:k}, the EMS map $\emsmap$ is compact on $\mathcal{M}^+(\X)$ endowed with the weak topology.
\end{corollary}
\begin{proof}
The EMS map is the composition of the continuous and bounded operator $\emmap$ (by Proposition~\ref{prop:emop}) which maps bounded sets into bounded sets with the compact smoothing operator $\Kop$ (by Proposition~\ref{prop:Kop}) which maps bounded sets into relatively compact sets. It follows that $\emsmap$ is continuous and maps bounded sets into relatively compact sets, hence $\emsmap$ is compact (e.g. \citetapp[page 54]{zeidler1985nonlinear}).
\end{proof}

\subsection{Proof of Proposition~\ref{prop:eu}}
The proposition may be established straightforwardly using the technical results obtained in the previous section.
\begin{proof}

Since $\X$ is a compact metric space (and therefore complete), the set of probability measures  $\mathcal{P}(\X)\subset \mathcal{M}(\X)$ is complete by Prokhorov's Theorem (e.g. \citet[Corollary 11.5.5]{dudley2002real}) and therefore $\mathcal{P}(\X)$ is closed.
Moreover, $\mathcal{P}(\X)$ is non-empty, bounded (since all of its elements have $\beta$ norm bounded by 1) and convex: take $\mu, \nu\in \mathcal{P}(\X)$ and $t\in\left[ 0, 1\right]$, then for every $A\in {B}(\X)$
\begin{align*}
& t\mu(A) + (1 - t)\nu(A) \geq 0
& t\mu(\X) + (1 - t)\nu(\X) = 1,
\end{align*}
showing that $t\mu+(1-t)\nu\in\mathcal{P}(\X)$ for all $t\in[0, 1]$ and all $\mu,\nu\in\mathcal{P}(\X)$.

These properties and the compactness of the EMS map (Corollary~\ref{cor:compact}) give the existence of a fixed point by Schauder's fixed point theorem see, e.g., \citetapp[Theorem 2.A]{zeidler1985nonlinear}.

\end{proof}

\section{Convergence of the SMC Approximation}
The theoretical characterization of the particle method approximating the EMS recursion is carried out by decomposing Algorithm~\ref{alg:fpsmc} into three steps: mutation, reweighting and resampling. This decomposition is standard in the study of SMC algorithms \citepapp{crisan2002survey, chopin2004central, miguez2013convergence} and allows us to examine the novelty of the particle approximation introduced in Section~\ref{sec:mf} by directly considering the contribution to the overall approximation error of the use of approximate weights $\weightN$.

First, consider the following decomposition of the dynamics in~\eqref{eq:smc} with potentials~\eqref{eq:potential} and Markov kernels~\eqref{eq:mutation}. In the selection step, the current state is weighted according to the potential function $\weight$
\begin{align*}
\update(x_{1:n}, y_{1:n})\equiv \bg(\predictive)(x_{1:n}, y_{1:n})=\frac{1}{\predictive(\weight)}\weight(x_n, y_n)\predictive(x_{1:n}, y_{1:n});
\end{align*}
then, in the mutation step, a new state is proposed according to $M_{n+1}$
\begin{align*}
\eta_{n+1}(x_{1:n+1}, y_{1:n+1})\propto\update(x_{1:n}, y_{1:n}) M_{n+1}(x_{n+1}\mid x_{n}).
\end{align*}
Each step of the evolution above is then compared to its particle approximation: the weighted distribution $\bg(\predictiveN)$ is compared with $\bg(\predictive)$, the resampled distribution $\updateN$ is compared with $\update$ and finally $\predictiveN$ is compared with $\predictive$. 

The proof of the $\lp$-inequality in Proposition~\ref{prop:lp} follows the inductive approach of \citetapp{crisan2002survey, miguez2013convergence} and consists of 4 Lemmata. Lemmata~\ref{lp:lemma2},~\ref{lp:lemma4} and~\ref{lp:lemma1} are due to \citetapp{crisan2002survey, miguez2013convergence} and establish $\lp$-error estimates for the reweighting step performed with the exact potential $\weight$ (exact reweighting), the multinomial resampling step and the mutation step.
Lemma~\ref{lp:lemma3} compares the exact reweighting with the reweighting obtained by using the approximated potentials $\weightN$ and is the main element of novelty in the proof.

In the following we commit the usual abuse of notation and we denote by $\predictive$ both a measure and its density with respect to the Lebesgue measure.

\subsection{Proof of Proposition~\ref{prop:lp}}
\label{app:lp}

Before proceeding to the proof of Proposition~\ref{prop:lp} we introduce the following auxiliary Lemma giving some properties of the approximated potentials $\weightN$:
\begin{lemma}
\label{lemma:potential}
Under~\ref{a:space} and~\ref{a:g}, the approximated and exact potentials are positive functions, bounded and bounded away from 0
\begin{align*}
\supnorm{\weight} \leq m_g^2<\infty\qquad\text{and}\qquad \inf_{(x, y)}\vert \weight(x, y)\vert \geq \frac{1}{m_g^2}>0\\
\supnorm{\weightN} \leq m_g^2<\infty\qquad\text{and}\qquad \inf_{(x, y)}\vert \weightN(x, y)\vert \geq \frac{1}{m_g^2}>0.
\end{align*}
We have the following decomposition
\begin{align*}
\weightN(x, y) - \weight(x, y) &=  \weight(x, y)\frac{ \predictiveX(g(y \mid \cdot)) - \predictiveNX(g(y \mid \cdot))}{ \predictiveNX(g(y \mid \cdot))}\\
& =  \weightN(x, y)\frac{\predictiveX(g(y \mid \cdot)) - \predictiveNX(g(y \mid \cdot))}{\predictiveX(g(y \mid \cdot))}
\end{align*}
for fixed $(x, y)\in \HH$.
\end{lemma}
\begin{proof}
The boundedness of $\weight$ and $\weightN$ follows from definitions~\eqref{eq:potential} and~\eqref{eq:potentialN} and the boundedness of $g$.
The second assertion is proved by considering the relative errors between the exact and the approximated potential:
\begin{align*}
\frac{ \weightN(x, y) - \weight(x, y)}{ \weight(x, y)} &=\frac{h_n(y)}{g(y \mid x)}\left[ \frac{g(y \mid x)}{h_n^N(y)} - \frac{g(y \mid x)}{h_n(y)} \right]\\
&= h_n(y) \left[\frac{1}{h_n^N(y)} - \frac{1}{h_n(y)}\right] \\
& = \frac{ h_n(y) - h_n^N(y)}{ h_n^N(y) }\\
& = \frac{ \predictiveX(g(y \mid \cdot)) - \predictiveNX(g(y \mid \cdot))}{ \predictiveNX(g(y \mid \cdot)) }
\end{align*}
and
\begin{align*}
\frac{\weightN(x, y) - \weight(x, y)}{ \weightN(x, y)} & = \frac{ \predictiveX(g(y \mid \cdot)) - \predictiveNX(g(y \mid \cdot))}{ \predictiveX(g(y \mid \cdot)) }
\end{align*}
respectively.

\end{proof}

\begin{lemma}[Exact reweighting]
\label{lp:lemma2}
Assume that for any $\testfn\in \measurable\left(\HH\right)$ and for some $p\geq1$
\begin{equation*}
\Exp\left[\vert\predictiveN(\testfn) - \predictive(\testfn)\vert^p\right]^{1/p} \leq \widetilde{C}_{p,n} \frac{\supnorm{ \testfn}}{N^{1/2}}
\end{equation*}
holds for some finite constant $\widetilde{C}_{p,n} $, then
\begin{equation*}
\Exp\left[\vert\bg(\predictiveN)(\testfn) - \bg(\predictive)(\testfn)\vert^p\right]^{1/p} \leq \bar{C}_{p,n} \frac{\supnorm{ \testfn}}{N^{1/2}}
\end{equation*}
for any $\testfn\in \measurable\left(\HH\right)$ for some finite constant $\bar{C}_{p,n} $.
\end{lemma}
\begin{proof}
The proof follows that of \citetapp[Lemma 4]{crisan2002survey} by exploiting the boundedness of $\weight$.
\end{proof}

\begin{lemma}[Approximate reweighting]
\label{lp:lemma3}
Assume that for any $\testfn\in \measurable\left(\HH\right)$ and for some $p\geq1$
\begin{equation*}
\Exp\left[\vert\predictiveN(\testfn) - \predictive(\testfn)\vert^p\right]^{1/p} \leq \widetilde{C}_{p,n} \frac{\supnorm{ \testfn}}{N^{1/2}}
\end{equation*}
holds for some finite constant $\widetilde{C}_{p,n} $, then
\begin{equation*}
\Exp\left[\vert\bgN(\predictiveN)(\testfn) - \bg(\predictiveN)(\testfn)\vert^p\right]^{1/p} \leq \ddot{C}_{p,n} \frac{\supnorm{ \testfn}}{N^{1/2}}
\end{equation*}
for any $\testfn\in \measurable\left(\HH\right)$ and for some finite constant $\ddot{C}_{p,n} $.
\end{lemma}

\begin{proof}
Apply the definition of $\bg$ and $\bgN$ and consider the following decomposition
\begin{align*}
\vert\bgN(\predictiveN)(\testfn) - \bg(\predictiveN)(\testfn)\vert & = \left\lvert \frac{\predictiveN( \weightN\testfn)}{\predictiveN( \weightN)} - \frac{\predictiveN( \weight\testfn)}{\predictiveN( \weight)} \right\rvert\\
&\leq \left\lvert \frac{\predictiveN( \weightN\testfn)}{\predictiveN( \weightN)} - \frac{\predictiveN( \weightN\testfn)}{\predictiveN( \weight)} \right\rvert\\
& + \left\lvert \frac{\predictiveN( \weightN\testfn)}{\predictiveN( \weight)} - \frac{\predictiveN( \weight\testfn)}{\predictiveN( \weight)} \right\rvert.
\end{align*}
Then, for the first term
\begin{align*}
\left\lvert \frac{\predictiveN( \weightN\testfn)}{\predictiveN( \weightN)} - \frac{\predictiveN( \weightN \testfn)}{\predictiveN( \weight)} \right\rvert & = \left\lvert \frac{\predictiveN( \weightN\testfn)}{\predictiveN( \weightN)}\right\rvert \left\lvert \frac{\predictiveN(\weight) - \predictiveN(\weightN)}{\predictiveN( \weight)} \right\rvert\\
&\leq \frac{\supnorm{\testfn}}{\vert \predictiveN( \weight) \vert} \predictiveN(\vert \weight - \weightN\vert).
\end{align*}
For the second term
\begin{align*}
\left\lvert \frac{\predictiveN( \weightN\testfn)}{\predictiveN( \weight)} - \frac{\predictiveN( \weight\testfn)}{\predictiveN( \weight)} \right\rvert & = \frac{1}{\vert \predictiveN( \weight) \vert}  \vert \predictiveN( \weightN\testfn) - \predictiveN( \weight\testfn)\vert\\
& \leq \frac{\supnorm{\testfn}}{\vert \predictiveN( \weight) \vert}\predictiveN(\vert  \weightN -  \weight\vert).
\end{align*}
Hence,
\begin{align*}
\vert\bgN(\predictiveN)(\testfn) - \bg(\predictiveN)(\testfn)\vert & \leq 2\frac{\supnorm{\testfn}}{\vert \predictiveN( \weight) \vert}\predictiveN(\vert  \weightN -  \weight\vert) \leq 2 m_g^2\supnorm{\testfn}\predictiveN(\vert  \weightN -  \weight\vert).
\end{align*}
By applying Minkowski's inequality and the decomposition of the potentials in Lemma~\ref{lemma:potential}
\begin{align*}
&\Exp\left[\left\lvert\predictiveN(\vert  \weightN -  \weight\vert)\right\rvert^p\right]^{1/p} \\
&= \Exp\left[\left\lvert \frac{1}{N}\sum_{i=1}^N \left\lvert  \weightN(X_n^i, Y_n^i) -  \weight(X_n^i, Y_n^i)\right\rvert \right\rvert^p\right]^{1/p}\\
&\leq  \frac{1}{N}\sum_{i=1}^N\Exp\left[\left\lvert \weightN(X_n^i, Y_n^i) -  \weight(X_n^i, Y_n^i) \right\rvert^p\right]^{1/p}\\
&\leq  \frac{1}{N}\sum_{i=1}^N \Exp\left[ \left\lvert \frac{\weightN(X_n^i, Y_n^i)}{\predictiveX\left( g(Y_n^i \mid \cdot)\right)}\right\rvert^{p} \vert \predictiveX\left( g(Y_n^i \mid \cdot)\right) - \predictiveNX\left( g(Y_n^i \mid \cdot)\right)\vert^p\right]^{1/p}\notag \\
&\leq \frac{1}{N}\sum_{i=1}^N  m_g^3\Exp\left[\vert \predictiveX\left( g(Y_n^i \mid \cdot)\right) - \predictiveNX\left( g(Y_n^i \mid \cdot)\right)\vert^p\right]^{1/p}.
\end{align*}
Then, consider $\mathcal{S}_n^N:= \sigma\left( Y_n^i: i\in\lbrace 1,\ldots, N\rbrace\right)$, the $\sigma$-field generated by all the $Y_n^i$ at time $n$.
By construction, the evolution of $X_n^i$ for $i=1, \ldots, N$ is independent of $\mathcal{S}_n^N$ (this is due to the definition of the mutation kernel~\eqref{eq:mutation}).
Conditionally on $\mathcal{S}_n^N$, the $Y_n^i$ are fixed for $i=1, \ldots, N$ and we can use the fact that the integrals of functions from $\X$ to $\mathbb{R}$ with respect to $\predictive$ and $\predictiveX$ coincide as do their integrals with respect to $\predictiveN$ and $\predictiveNX$, thus for fixed $y$:
\begin{align*}
\Exp\left[\vert \predictiveX\left( g(y \mid \cdot)\right) - \predictiveNX\left( g(y \mid \cdot)\right)\vert^p \right]^{1/p} &= \Exp\left[\vert \predictive\left( g(y \mid \cdot)\right) - \predictiveN\left( g(y \mid \cdot)\right)\vert^p \right]^{1/p}\\
&\leq \frac{m_g \widetilde{C}
_{p, n}}{N^{1/2}}
\end{align*}
where the last inequality follows from the hypothesis of the Lemma because $g(y \mid \cdot)$ is a bounded and measurable function for all fixed $y\in\mathbb{Y}$.

Hence, since $Y_n^i$ is $\mathcal{S}_n^N$-measurable and independent of $\predictiveN\vert_{\X}$, we have
\begin{align*}
\Exp\left[\left\lvert\predictiveN(\vert  \weightN -  \weight\vert)\right\rvert^p\right]^{1/p} &\leq m_g^3\frac{1}{N}\sum_{i=1}^N \Exp\left[\vert \predictiveX\left( g(Y_n^i \mid\cdot)\right) - \predictiveNX\left( g(Y_n^i \mid \cdot)\right)\vert^p\right]^{1/p}\\
&\leq m_g^3 \frac{1}{N}\sum_{i=1}^N\Exp\left[\Exp\left[\vert \predictiveX\left( g(Y_n^i \mid \cdot)\right) - \predictiveNX\left( g(Y_n^i \mid \cdot)\right)\vert^p| \mathcal{S}_n^N\right]\right]^{1/p}\\
&\leq  \frac{m_g^4 \widetilde{C}_{p, n}}{N^{1/2}}.
\end{align*}

Therefore,
\begin{align*}
\Exp\left[\vert\bgN(\predictiveN)(\testfn) - \bg(\predictiveN)(\testfn)\vert^p\right]^{1/p} & \leq 2\widetilde{C}_{p,n}  m_g^6\frac{\supnorm{\testfn}}{N^{1/2}},
\end{align*}
with the constant $\ddot{C}_{p,n}  = 2\widetilde{C}_{p,n} m_g^6$.
\end{proof}

\begin{lemma}[Multinomial resampling]
\label{lp:lemma4}
Assume that for any $\testfn\in \measurable\left(\HH\right)$ and for some $p\geq 1$
\begin{equation*}
\Exp\left[\vert\bgN(\predictiveN)(\testfn) - \update(\testfn)\vert^p\right]^{1/p} = \Exp\left[\vert\bgN(\predictiveN)(\testfn) - \bg(\predictive)(\testfn)\vert^p\right]^{1/p} \leq \widehat{C}_{p, n}\frac{\supnorm{\testfn}}{N^{1/2}}
\end{equation*}
holds for some finite constant $\widehat{C}_{p, n}$, then after the resampling step performed through multinomial resampling
\begin{equation*}
\Exp\left[\vert\updateN(\testfn) - \hat{\eta}_n(\testfn)\vert^p\right]^{1/p} \leq C_{p,n}\frac{\supnorm{\testfn}}{N^{1/2}}
\end{equation*}
for any $\testfn\in \measurable\left(\HH\right)$ for some finite constant $C_{p,n} $.
\end{lemma}
\begin{proof}
The proof follows that of \citetapp[Lemma 5]{crisan2002survey} using the Marcinkiewicz-Zygmund type inequality in \citetapp[Lemma 7.3.3]{del2004feynman} and the hypothesis.

\end{proof}

\begin{lemma}[Mutation]
\label{lp:lemma1}
Assume that for any $\testfn\in \measurable(\HH)$ and for some $p\geq 1$
\begin{equation*}
\Exp\left[\vert\hat{\eta}^N_{n}(\testfn) - \hat{\eta}_{n}(\testfn)\vert^p\right]^{1/p} \leq C_{p,n} \frac{\supnorm{\testfn}}{N^{1/2}}
\end{equation*}
holds for some finite constant $C_{p,n} $, then, after the mutation step
\begin{equation*}
\Exp\left[\vert\eta_{n+1}^N(\testfn) - \eta_{n+1}(\testfn)\vert^p\right]^{1/p} \leq \widetilde{C}_{p,n+1} \frac{\supnorm{\testfn}}{ N^{1/2}}
\end{equation*}
for any $\testfn\in \measurable(\HH)$ for some finite constant $\widetilde{C}_{p,n+1} $.
\end{lemma}
\begin{proof}
The proof follows that of \citetapp[Lemma 3]{crisan2002survey}, where after applying Minkowski's inequality
\begin{align*}
\Exp\left[\vert\eta_{n+1}^N(\testfn) - \eta_{n+1}(\testfn)\vert^p\right]^{1/p} & = \Exp\left[\vert\eta_{n+1}^N(\testfn) - \hat{\eta}_{n}M_{n+1}(\testfn)\vert^p\right]^{1/p}\\
& \leq \Exp\left[\vert\eta_{n+1}^N(\testfn) - \hat{\eta}^N_{n}M_{n+1 }(\testfn)\vert^p\right]^{1/p} \\
&+ \Exp\left[\vert\hat{\eta}^N_{n}M_{n+1}(\testfn)- \hat{\eta}_{n} M_{n+1}(\testfn) \vert^p\right]^{1/p},
\end{align*}
we can bound the first term with the Marcinkiewicz-Zygmund type inequality in \citetapp[Lemma 7.3.3]{del2004feynman} and the second term with the hypothesis.

\end{proof}

The proof of the $\lp$-inequality in Proposition~\ref{prop:lp} is based on an inductive argument which uses Lemmata~\ref{lp:lemma2}-\ref{lp:lemma1}:
\begin{proof}[Proof of Proposition~\ref{prop:lp}]
At time $n=1$, the particles $(X_1^i, Y_1^i)_{i=1}^N$ are sampled i.i.d. from $\eta_1\equiv\hat{\eta}_1$ thus $\Exp\left[\testfn(X_1^i, Y_1^i)\right] = \eta_1(\testfn)$ for $i=1,\ldots,N$.
We can define the sequence of functions $\Delta_1^i: \X\times\Y \mapsto \mathbb{R}$ for $i=1,\ldots, N$
\begin{equation*}
\Delta_1^i(x, y) :=  \testfn(x, y) -  \Exp\left[\testfn(X_{1}^{i}, Y_{1}^{i})\right]
\end{equation*}
so that,
\begin{align*}
\eta_1^N(\testfn) - \eta_1(\testfn) = \frac{1}{N}\sum_{i=1}^N \Delta_1^i(X_1^i, Y_1^i),
\end{align*}
and apply Lemma 7.3.3 in \citeapp{del2004feynman} to get
\begin{equation*}
\Exp\left[\vert \eta_1^N(\testfn) - \eta_1(\testfn)\vert^p\right]^{1/p} \leq 2b(p)^{1/p}\frac{\supnorm{\testfn}}{N^{1/2}},
\end{equation*}
with $b(p) < \infty$, for every $p\geq 1$.

Then, assume that the result holds at time $n$: for every $\testfn \in \measurable(\HH)$, every $p\geq1$ and some finite constant $\widetilde{C}_{p,n}$
\begin{equation*}
\Exp\left[\vert\predictiveN(\testfn) - \predictive(\testfn)\vert^p\right]^{1/p} \leq \widetilde{C}_{p,n} \frac{\supnorm{\testfn}}{ N^{1/2}}.
\end{equation*}

The $\lp$-inequality in~\eqref{eq:lp4kde} is obtained by combining the results of Lemma~\ref{lp:lemma2} and Lemma~\ref{lp:lemma3} using Minkowski's inequality
\begin{equation*}
\Exp\left[\vert\bgN(\predictiveN)(\testfn) - \bg(\predictive)(\testfn)\vert^p\right]^{1/p} \leq (\bar{C}_{p,n} +\ddot{C}_{p,n} )\frac{\supnorm{ \testfn}}{N^{1/2}}
\end{equation*}
for every $\testfn \in \measurable(\HH)$ and some finite constants $\bar{C}_{p,n}, \ddot{C}_{p,n}$. Thus, $\widehat{C}_{p, n} = \bar{C}_{p,n} +\ddot{C}_{p,n}$.

Lemma~\ref{lp:lemma4} gives
\begin{equation*}
\Exp\left[\vert\updateN(\testfn) - \update(\testfn)\vert^p\right]^{1/p} \leq C_{p,n} \frac{\supnorm{\testfn}}{N^{1/2}}
\end{equation*}
for every $\testfn \in \measurable(\HH)$ and some finite constants $C_{p,n}$, and Lemma~\ref{lp:lemma1} gives
\begin{equation*}
\Exp\left[\vert\eta_{n+1}^N(\testfn) - \eta_{n+1}(\testfn)\vert^p\right]^{1/p} \leq \widetilde{C}_{p,n+1} \frac{\supnorm{\testfn}}{ N^{1/2}}
\end{equation*}
for every $\testfn \in \measurable(\HH)$ and some finite constant $\widetilde{C}_{p,n+1}$.

The result follows for all $n\in \mathbb{N}$ by induction.
\end{proof}

\subsection{Proof of Proposition \protect{\ref{prop:asw}}} \label{app:wcm}
Using standard techniques following \citetapp[Chapter 11, Theorem 11.4.1]{dudley2002real} and \citetapp{berti2006almost} and given in detail for the context of interest by \citetapp[Theorem 4]{schmon2018large}, the result of Corollary~\ref{prop:slln} can be strengthened to the convergence of the measures in the weak topology.

\begin{proof}[Proof of Proposition~\ref{prop:asw}]
Consider $\BL(\HH)\subset \measurable(\HH)$, the Banach space of bounded Lipschitz functions. As shown in \citetapp[Theorem 11.4.1]{dudley2002real}, see also \citetapp[Proposition 5]{schmon2018large} for a more accessible presentation, $\BL(\HH)$ admits a countable dense subclass $C\subset \BL(\HH)$.

For every $\testfn\in C$ define $A_\testfn := \lbrace \omega\in \Omega: \predictiveN(\omega)(\testfn) \rightarrow\predictive(\testfn)\ N\rightarrow\infty\rbrace$. Then $\pr\left(A_\testfn\right) = 1$ $\forall\testfn\in C$ by Corollary~\ref{prop:slln} and
\begin{align*}
\pr\left(\lbrace\omega\in\Omega: \predictiveN(\omega)(\testfn) \rightarrow\predictive(\testfn)\ N\rightarrow\infty\ \forall\testfn\in C\rbrace \right)= \pr\left(\bigcap_{\testfn\in C} A_\testfn \right) = 1.
\end{align*}
The result follows from the fact that $C$ is dense in $\BL(\HH)$ and the Portmanteau Theorem (e.g. \citetapp[Theorem 11.1.1]{dudley2002real}).
\end{proof}

\section{Convergence of Density Estimates}
\label{sec:kde_misepf}

\subsection{Auxiliary Results}\label{sec:lp_kdepf}
Using a version of the Dominated Convergence Theorem for weakly converging measures \citepapp{serfozo1982convergence, feinberg2020fatou2}, standard results on kernel density estimation (e.g. \citetapp{parzen1962estimation, cacoullos1966estimation}) and an argument based on compactness as in \citeapp{newey1991uniform} we can establish the following result
\begin{prop}
\label{prop:lp_kde}
Under~\ref{a:space}, \ref{a:g} and~\ref{a:k}, if $\bwN\rightarrow0$ as $N\rightarrow \infty$, the estimator $f^{N}_{n+1}(x)$  in~\eqref{eq:smc_kde1} converges uniformly to $f_{n+1}(x)$ with probability 1 for all $n\geq 1$.
\end{prop}

\begin{proof}
Let us define for $N\in\mathbb{N}$
\begin{align*}
\testfn^N(t, x) := \int_{\X}K(x^\prime, x)\bwN^{-d_{\X}}\vert \Sigma\vert^{-1/2}S\left( (\bwN^2 \Sigma)^{-1/2}(t-x^\prime)\right)\dx^\prime,
\end{align*}
and note that the estimator~\eqref{eq:smc_kde1} is given by $f_{n+1}^N(x) = \bgN(\predictiveN) \left( \testfn^N(\cdot, x) \right)$ for any fixed $x\in\X$.
Standard results in the literature on kernel density estimation show that $\testfn^N(\cdot, x)$ converges to $K(\cdot, x)$ pointwise for all $ x\in\X$ (e.g. \citetapp[Theorem 2.1]{cacoullos1966estimation}).
Because $\X$ is compact, Assumption~\ref{a:k} ensures that $K$ is \emph{uniformly} continuous on $\X$ (e.g. \citetapp[Theorem~4.19]{rudin1964principles}), then, as argued in \citetapp[Theorem 3.A]{parzen1962estimation}, the sequence $\testfn^N(\cdot, x)$ converges uniformly to $K(\cdot, x)$ in $\X$
(see also \citetapp[Theorem 3.3]{cacoullos1966estimation}).
As a consequence, the sequence $\lbrace\testfn^N(\cdot, x)\rbrace_{N\in\mathbb{N}}$ is uniformly equicontinuous and uniformly bounded (e.g. \citetapp[Theorem 7.25]{rudin1964principles}). It follows that $\lbrace\testfn^N(\cdot, x)\rbrace_{N\in\mathbb{N}}$ is (asymptotically) uniformly integrable in the sense of \citetapp[Definition 2.6]{feinberg2020fatou2}.

Using an argument analogous to that in Proposition~\ref{prop:asw} we can establish that $\bgN(\predictiveN)$ converges to $\bg(\predictive)$ almost surely in the weak topology, then using the fact that the sequence $\lbrace\testfn^N(\cdot, x)\rbrace_{N\in\mathbb{N}}$ is asymptotically uniformly integrable and equicontinuous with continuous limit $K(\cdot, x)$, the Dominated Convergence theorem for weakly converging measures (\citetapp[Corollary 5.2]{feinberg2020fatou2}; see also \citetapp[Theorem 3.3]{serfozo1982convergence}) implies that 
\begin{align}
\label{eq:as_estimator}
    f_{n+1}^N(x)=\bgN(\predictiveN)\left(\testfn^N(\cdot, x)\right)\rightarrow \bg(\predictive)  \left(K(\cdot, x) \right)= f_{n+1}(x)
\end{align}
almost surely as $N\rightarrow\infty$ for any fixed $x\in\X$.

To turn the result above into almost sure uniform convergence, i.e. 
\begin{align*}
    \pr\left( \limsup_{N\to\infty} \left\lbrace  \ \sup_{x\in\X}\vert f_{n+1}^N(x)-f_{n+1}(x)\vert >\varepsilon \right\rbrace\right)=0
\end{align*}
for every $\varepsilon>0$, we exploit assumption~\ref{a:space} and the resulting continuity properties of $K$.

Under~\ref{a:space}--\ref{a:k}, $K$ is uniformly continuous and we have that for any $\varepsilon>0$, there exists some $\delta_\varepsilon > 0$ such that
 \begin{align*}
   \vert f_{n+1}(x)-f_{n+1}(x^\prime)\vert& = \vert  \bg(\predictive)  \left(K(\cdot, x) -K(\cdot, x^\prime)\right)\vert \\
   &\leq \sup_{z\in\X}\vert K(z, x) -K(z, x^\prime)\vert \leq\frac{\varepsilon}{3}
\end{align*}
whenever $\norm{x-x^\prime}{2}<\delta_{\varepsilon}$. 
Using the definition of $\testfn^N$ and exploiting again the uniform continuity of $K$ we also have that for every $\varepsilon>0$
\begin{align*}
    \vert \testfn^N(t, x) -\testfn^N(t, x^\prime)\vert & \leq  \int_{\X}\vert K(u, x)-K(u, x^\prime)\vert \bwN^{-d_{\X}}\vert \Sigma\vert^{-1/2}S\left( (\bwN^2 \Sigma)^{-1/2}(t-u)\right)\du\\
    &\leq \frac{\varepsilon}{3} \int_{\X}\bwN^{-d_{\X}}\vert \Sigma\vert^{-1/2}S\left( (\bwN^2 \Sigma)^{-1/2}(t-u)\right)\du\leq \frac{\varepsilon}{3}
\end{align*}
if  $\norm{x-x^\prime}{2}<\delta_{\varepsilon}$.
It follows that $f^N_{n+1}$ is uniformly continuous: for any $\varepsilon>0$ 
\begin{align*}
       \vert f_{n+1}^N(x)-f_{n+1}^N(x^\prime)\vert& = \vert  \bgN(\predictiveN)  \left(\testfn^N(\cdot, x) -\testfn^N(\cdot, x^\prime)\right)\vert \\
   &\leq \sup_{z\in\X}\vert \testfn^N(z, x) -\testfn^N(z, x^\prime)\vert \leq\frac{\varepsilon}{3}
\end{align*}
whenever $\norm{x-x^\prime}{2}<\delta_{\varepsilon}$.

Let $B(x, \delta_{\varepsilon}):=\{x^\prime\in \X:\norm{x-x^\prime}{2}< \delta_\varepsilon\}$ denote the ball in $\X$ centred around $x$ of radius $\delta_\varepsilon$.
Under~\ref{a:space}, $\X$ is compact and therefore there exists a finite subcover $\{B(x^j)\}_{j=1}^J$ of $\{B(x,\delta_{\varepsilon})\}_{x\in \X}$.
Using the uniform continuity above and the following decomposition, we obtain for all $x\in B(x^j)$, $j=1, \ldots, J$ and for all $N$,
\begin{align*}
    \vert f_{n+1}^N(x)-f_{n+1}(x)\vert& \leq \vert f_{n+1}^N(x)-f_{n+1}^N(x^j)\vert + \vert f_{n+1}^N(x^j)-f_{n+1}(x^j)\vert + \vert f_{n+1}(x^j)-f_{n+1}(x)\vert\\
    &\leq \frac{\varepsilon}{3} +\vert f_{n+1}^N(x^j)-f_{n+1}(x^j)\vert+\frac{\varepsilon}{3} \\
    &\leq \frac{2}{3}\varepsilon +\max_{j=1, \ldots, J}\vert f_{n+1}^N(x^j)-f_{n+1}(x^j)\vert,
\end{align*}
from which follows
\begin{align*}
   \sup_{x\in\X} \vert f_{n+1}^N(x)-f_{n+1}(x)\vert& \leq \frac{2}{3}\varepsilon +\max_{j=1, \ldots, J}\vert f_{n+1}^N(x^j)-f_{n+1}(x^j)\vert.
\end{align*}

Therefore, to obtain almost sure uniform convergence, it is sufficient to show that
\begin{align*}
\pr\left( \left\lbrace \omega\in \Omega: \max_{j=1, \ldots, J}\vert f_{n+1}^N(\omega)(x^j)-f_{n+1}(x^j)\vert \rightarrow 0\ N\rightarrow\infty \right\rbrace\right) =1.
\end{align*}
Let us define $A_j:= \lbrace \omega\in \Omega: f_{n+1}^N(\omega)(x^j) \rightarrow f_{n+1}(x^j)\ N\rightarrow\infty\rbrace$. As a consequence of~\eqref{eq:as_estimator} we have $\pr(A_j)=1$ for all $j=1,\ldots, J$ and
\begin{align*}
    \pr\left( \left\lbrace \omega\in \Omega: \max_{j=1, \ldots, J}\vert f_{n+1}^N(\omega)(x^j)-f_{n+1}(x^j)\vert \rightarrow 0\ N\rightarrow\infty \right\rbrace\right) =\pr\left( \bigcap_{j=1,\ldots, J} A_j\right)=1,
\end{align*}
which gives the result.
\end{proof}

\subsection{Proof of Proposition~\ref{cor:kde_mise}}
\begin{proof}
A direct consequence of Proposition~\ref{prop:lp_kde} is the almost sure pointwise convergence of $f_{n+1}^N$ to $f_{n+1}$.
As both $f_{n+1}^N(x)$ and $f_{n+1}(x)$ are probability densities on $\X$, we can extend them to $\real^{d_{\X}}$ by taking
\begin{align*}
\psi_{n+1}^N(x) := \begin{cases}
f_{n+1}^N(x)\qquad x\in \X\\
0 \qquad \text{otherwise}
\end{cases}
\qquad \text{and}\qquad \psi_{n+1}(x) := \begin{cases}
f_{n+1}(x)\qquad x\in \X\\
0 \qquad \text{otherwise}
\end{cases}
\end{align*}
respectively.
Both $\psi_{n+1}^N(x)$ and $\psi_{n+1}(x)$ are probability densities on $\real^{d_{\X}}$ and are measurable functions. Moreover, $\psi_{n+1}^N(x)$ converges almost surely to $\psi_{n+1}(x)$ for all $x\in\real^{d_{\X}}$.  Hence, we can apply Glick's extension to Scheff\'e's Lemma (e.g. \citeapp{devroye1979l1}) to obtain
\begin{equation*}
\int_{\real^d} \vert \psi_{n+1}^N(x) - \psi_{n+1}(x)\vert \ \dx \overset{\textrm{a.s.}}{\rightarrow} 0
\end{equation*}
from which we can conclude
\begin{equation*}
\int_{\X} \vert f_{n+1}^N(x) - f_{n+1}(x)\vert \ \dx = \int_{\X} \vert \psi_{n+1}^N(x) - \psi_{n+1}(x)\vert \ \dx \leq \int_{\real^{d_{\X}}} \vert \psi_{n+1}^N(x) - \psi_{n+1}(x)\vert \ \dx \rightarrow 0
\end{equation*}
almost surely as $N\rightarrow \infty$.

Convergence of the MISE is a consequence of Proposition~\ref{prop:lp_kde}, \ref{a:space} and the Dominated Convergence Theorem
\begin{align*}
\Exp\left[\int_{\X}  \vert f_{n+1}^N(x) - f_{n+1}(x)\vert^2\dx\right]  \leq \lambda(\X)\Exp\left[ \supnorm{ f_{n+1}^N - f_{n+1}}^2\right] \dx  \rightarrow 0
\end{align*}
as $N\rightarrow \infty$, where $\lambda(\X)<\infty$ denotes the Lebesgue measure of $\X$.

\end{proof}

\FloatBarrier
\section{Additional Examples}
\label{sec:additional_examples}
\subsection{Analytically tractable example}
\label{sec:analytically_tractable}
Here we consider a toy example involving Gaussian densities for which both the EM recursion~\eqref{eq:em} and the EMS recursion~\eqref{eq:ems} can be solved at least implicitly.
The Fredholm integral equation we consider is
\begin{equation*}
\N(y; \mu, \sigma_f^2 + \sigma_g^2) = \int_{\X} \N(x; \mu, \sigma_f^2)\N(y; x, \sigma_g^2)\ \dx,\qquad y\in\Y
\end{equation*}
where $\X = \Y = \real$.
The initial distribution $f_1(x)$ is $\N(x; \mu, \sigma_{\textrm{EMS}, 1}^2)$ for some $\sigma_{\textrm{EMS}, 1}^2 > 0$.

\begin{figure}[t]
\centering
\resizebox{0.9\textwidth}{!}{%
\begin{tikzpicture}[every node/.append style={font=\normalsize}]
\node (img1) {\includegraphics[width=0.40\textwidth]{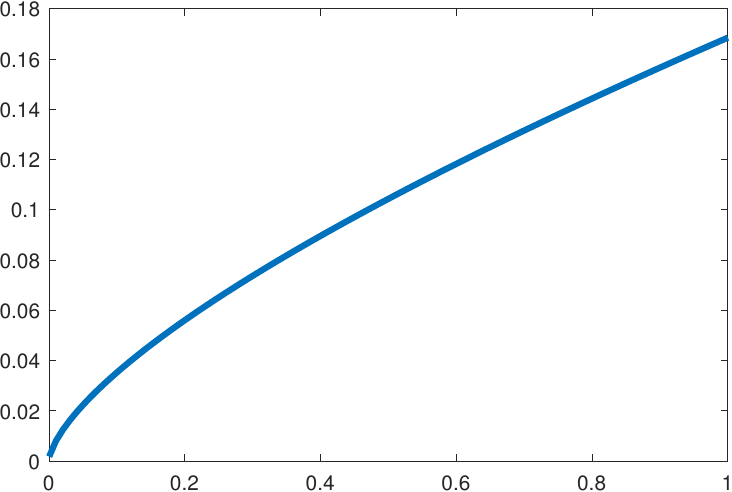}};
\node[below=of img1, node distance = 0, yshift = 1cm] { $\varepsilon$};
  \node[left=of img1, node distance = 0, rotate=90, anchor = center, yshift = -0.7cm] {$\sigma_{\textrm{EMS}}^2$};
\end{tikzpicture}
\begin{tikzpicture}[every node/.append style={font=\normalsize}]
\node[right= of img2, yshift = 0.1cm] (img3){\includegraphics[width=0.40\textwidth]{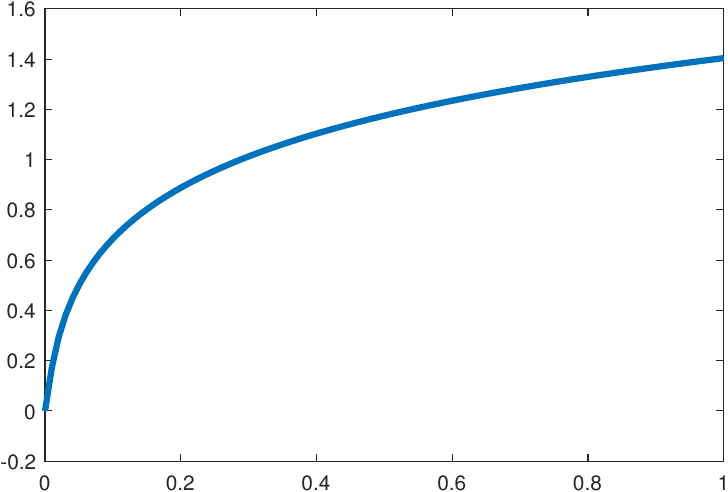}};
\node[below=of img3, node distance = 0, yshift = 1cm] {$\varepsilon$};
  \node[left=of img3, node distance = 0, rotate=90, anchor = center, yshift = -0.7cm] {$\KL$};
\end{tikzpicture}
}
\caption{Functional dependence of the variance of the resulting approximation $\sigma_{\textrm{EMS}}^2$ (left) and the Kullback--Leibler divergence~\eqref{eq:kl_gaussian} (right) on the smoothing parameter $\varepsilon$.}
\label{fig:gaussian_smoothing}
\end{figure}
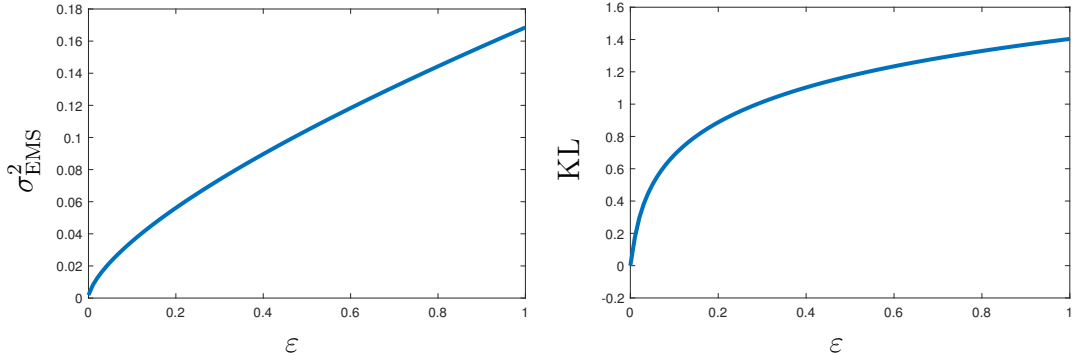

The fixed point $f_{EMS}$ of the EMS recursion~\eqref{eq:ems} with Gaussian smoothing kernel $K(x^\prime, x)=\N(x; x^\prime, \varepsilon^2)$ is a Gaussian with mean $\mu$ and variance $\sigma_{\textrm{EMS}}^2$ solving
\begin{equation}
\label{eq:gaussian_fixedpoint}
\sigma_{\textrm{EMS}}^6 + \sigma_{\textrm{EMS}}^4(\sigma_g^2 - \sigma_h^2) -2 \sigma_{\textrm{EMS}}^2 \varepsilon^2\sigma^2_g -2\varepsilon^2\sigma^2_g = 0.
\end{equation}
We can compute the Kullback--Leibler divergence achieved by $f_{EMS}$:
\begin{equation}
\label{eq:kl_gaussian}
\KL\left( h, \int_{\X} f_{EMS}(x)g(y \mid \cdot)\ \dx\right) = \frac{1}{2}\log \frac{\sigma^2_{\textrm{EMS}} + \sigma^2_g}{\sigma^2_h} + \frac{\sigma^2_h}{2(\sigma^2_{\textrm{EMS}} + \sigma^2_g)} - \frac{1}{2},
\end{equation}
as $\int_{\X} f_{EMS}(x)g(y \mid \cdot)\ \dx$ is the Gaussian density $\N(y; \mu, \sigma^2_{\textrm{EMS}} + \sigma^2_g)$.
The fixed point for the EM recursion~\eqref{eq:em} is obtained setting $\varepsilon=0$. The corresponding value of the Kullback--Leibler divergence is 0.
Figure~\ref{fig:gaussian_smoothing} shows the dependence of $\sigma^2_{\textrm{EMS}}$ and of the $\KL$ divergence on $\varepsilon$.

The conjugacy properties of this model allow us to obtain an exact form for the potential~\eqref{eq:potential}
\begin{equation}
\label{eq:gaussian_potentialexact}
G_n(x_n, y_n) = \frac{g(y_n \mid x_n)}{h_n(y_n)} = \frac{\N(y_n; x_n, \sigma_g^2)}{\N(y_n; \mu, \sigma_g^2 + \sigma_{\textrm{EMS}, n}^2)}
\end{equation}
where $\sigma_{\textrm{EMS}, n}^2$ is the variance of $f_n(x)$.

We use this example to show that the maximum likelihood estimator obtained with the EM iteration~\eqref{eq:em_discrete} does not enjoy good properties, and to motivate the addition of a smoothing step in the iterative process (Figure~\ref{fig:gaussian_motivation}).

Taking $\sigma_f^2=0.043^2$ and $\sigma_g^2=0.045^2$ we have $|1 - \int_0^1 f(x) \dx| < 10^{-30}$, thus we can restrict our attention to $[0, 1]$ and implement the discretized EM and EMS by taking $B=D=100$ equally spaced intervals in this interval. The number of iterations $n=100$ is fixed for EM, EMS and SMC. The number of particles for SMC is $N=10^4$ and $\varepsilon = 10^{-2}$.
The smoothing matrix for EMS is obtained by discretization of the smoothing kernel $K(x^\prime, x)=\N(x; x^\prime, \varepsilon^2)$.

\begin{figure}[t]
\centering
\resizebox{0.8\textwidth}{!}{%
\begin{tikzpicture}[every node/.append style={font=\normalsize}]
\node (img1) {\includegraphics[width=0.75\textwidth]{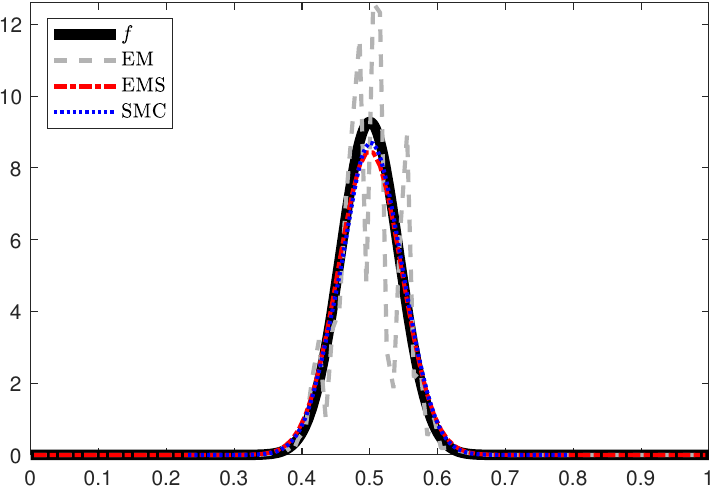}};
\node[below=of img1, node distance = 0, yshift = 1cm] { $x$};
  \node[left=of img1, node distance = 0, rotate=90, anchor = center, yshift = -0.7cm] {$f(x)$};
\end{tikzpicture}
}
\caption{Comparison of EM, EMS and SMC with exact potential $\weight$ for the analytically tractable example.}
\label{fig:gaussian_motivation}
\end{figure}

Figure~\ref{fig:gaussian_motivation} clearly shows that the EM estimate, despite identifying the correct support of the solution, cannot recover the correct shape and is not smooth. On the contrary, both EMS and SMC give good reconstruction of $f$ while preserving smoothness.

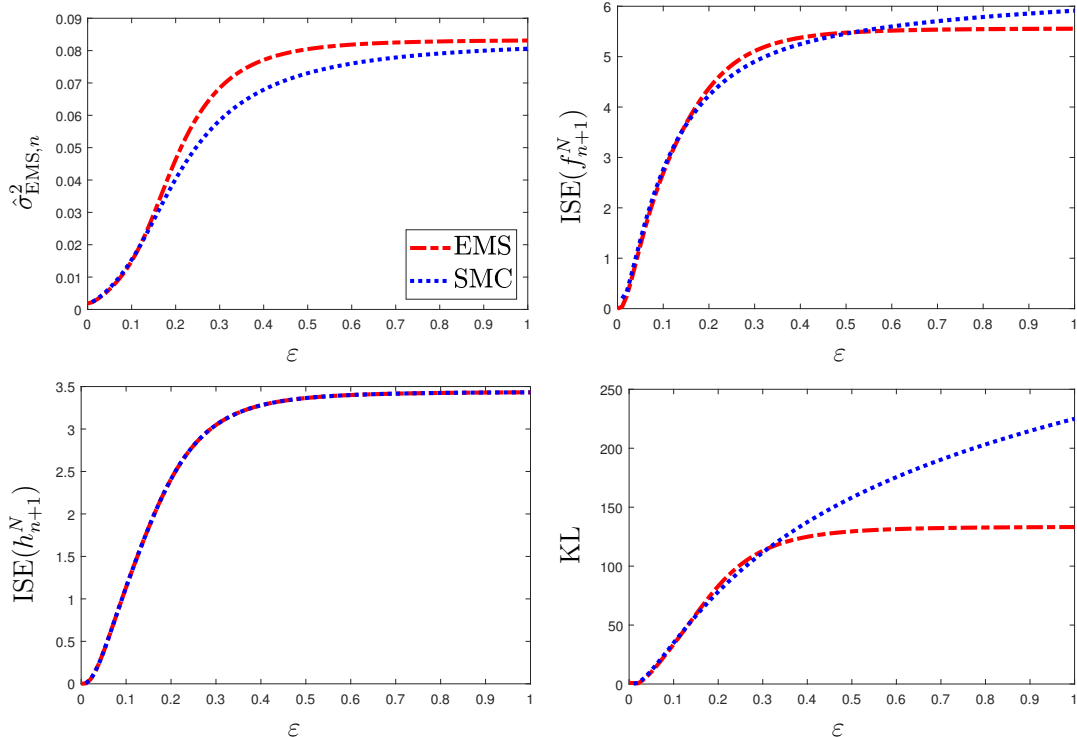
\begin{figure}[t]
\centering
\resizebox{0.9\textwidth}{!}{%
\begin{tikzpicture}[every node/.append style={font=\normalsize}]
\node (img1) {\includegraphics[width=0.45\textwidth]{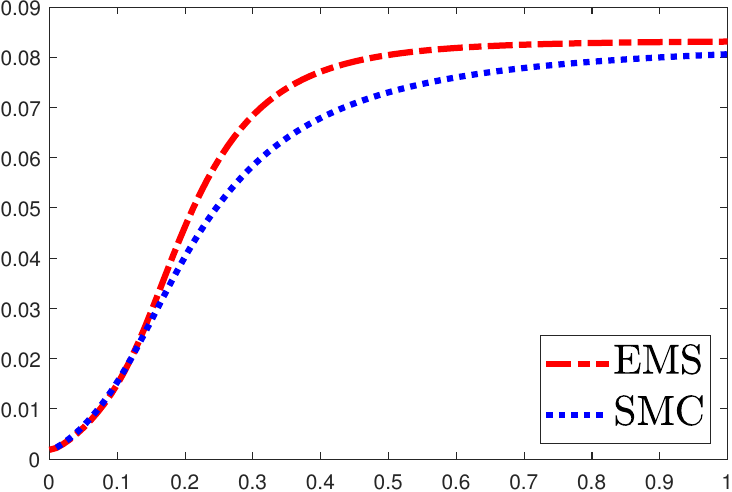}};
\node[below=of img1, node distance = 0, yshift = 1cm] { $\varepsilon$};
  \node[left=of img1, node distance = 0, rotate=90, anchor = center, yshift = -0.7cm] {$\hat{\sigma}_{\textrm{EMS}, n}^2$};
\end{tikzpicture}
\begin{tikzpicture}[every node/.append style={font=\normalsize}]
\node[right= of img1, yshift = 0.1cm] (img2){\includegraphics[width=0.45\textwidth]{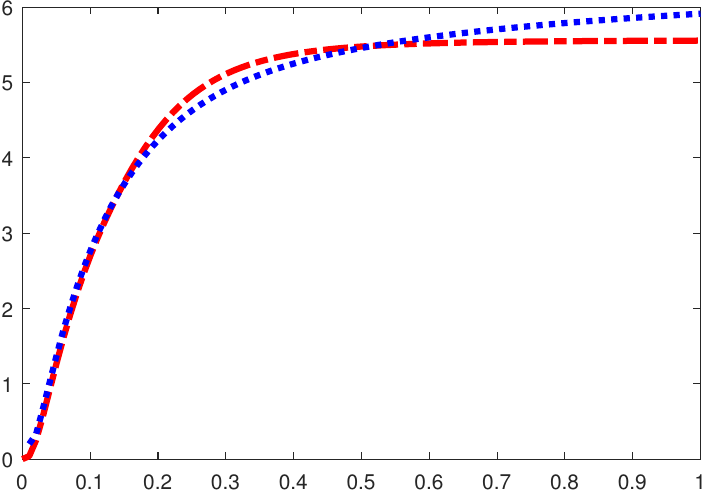}};
\node[below=of img2, node distance = 0, yshift = 1cm] {$\varepsilon$};
  \node[left=of img2, node distance = 0, rotate=90, anchor = center, yshift = -0.7cm] {$\ise(f_{n+1}^N)$};
\end{tikzpicture}
}
\resizebox{0.9\textwidth}{!}{%
\begin{tikzpicture}[every node/.append style={font=\normalsize}]
\node[below= of img2, yshift = 0.1cm] (img3){\includegraphics[width=0.45\textwidth]{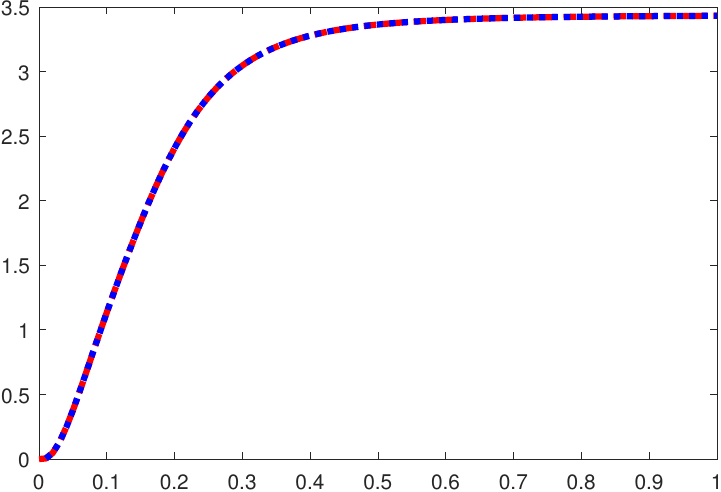}};
\node[below=of img3, node distance = 0, yshift = 1cm] {$\varepsilon$};
  \node[left=of img3, node distance = 0, rotate=90, anchor = center, yshift = -0.7cm] {$\ise(h_{n+1}^N)$};
\end{tikzpicture}
\begin{tikzpicture}[every node/.append style={font=\normalsize}]
\node[right= of img3, yshift = 0.1cm] (img4){\includegraphics[width=0.45\textwidth]{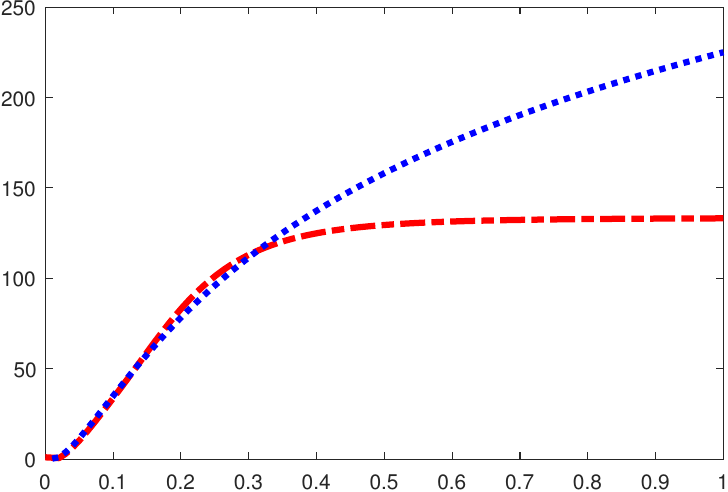}};
\node[below=of img4, node distance = 0, yshift = 1cm] {$\varepsilon$};
  \node[left=of img4, node distance = 0, rotate=90, anchor = center, yshift = -0.7cm] {$\KL$};
\end{tikzpicture}
}
\caption{Estimated variance (top left), $\ise(f_{n+1}^N)$ (top right), $\ise(h_{n+1}^N)$ (bottom left) and Kullback--Leibler divergence (bottom right) as functions of the smoothing parameter $\varepsilon$ for the analytically tractable example. The deterministic discretization~\eqref{eq:ems_discrete} (red) and the stochastic discretization via SMC with the exact potentials~\eqref{eq:gaussian_potentialexact} (blue) are compared.}
\label{fig:gaussian_comparison_ems_smc}
\end{figure}

Then we compare the deterministic discretization~\eqref{eq:ems_discrete} of the EMS recursion~\eqref{eq:ems} with the stochastic one given by SMC with the exact potential~\eqref{eq:gaussian_potentialexact}. To do so, we consider the variance of the obtained reconstructions, their integrated square error~\eqref{eq:ise}, the mean integrated square error for between $h$ and
\begin{equation*}
\hat{h}^N_{n+1}(y) = \int_{\X}f^N_{n+1}(x)g(y \mid x) \dx
\end{equation*}
and the Kullback--Leibler divergence $\KL(h, \hat{h}^N_{n+1})$ (restricting to the $[0,1]$ interval and computing by numerical integration) as the value of the smoothing parameter $\varepsilon$ increases (Figure~\ref{fig:gaussian_comparison_ems_smc}).
We consider one run of discretized EMS and compare it with 1,000 repetitions of SMC for each value of $\varepsilon$ (this choice follows from the fact that discretized EMS is a deterministic algorithm). The number of particles for SMC is $N=10^3$ and for each run we draw a sample $\mathbf{Y}$ of size $10^4$ from $h$ and resample from it $M=\min(N, 10^4)$ particles in line 2 of Algorithm~\ref{alg:fpsmc}.
Both algorithms correctly identify the mean for every value of $\varepsilon$ while the estimated variances increase from that obtained with the EM algorithm ($\varepsilon = 0$) to the variance of a Uniform distribution over $[0, 1]$ (Figure~\ref{fig:gaussian_comparison_ems_smc} top left).
Unsurprisingly, the $\ise$ for both $f_{n+1}^N$ and $h_{n+1}^N$ increases with $\varepsilon$ (Figure~\ref{fig:gaussian_comparison_ems_smc} top right and bottom left), showing that an excessive amount of smoothing leads to poor reconstructions. In particular for values of $\varepsilon \geq 0.5$ the reconstructions of $f$ become flatter and tend to coincide with a Uniform distribution in the case of EMS and with a normal distribution centered at $\mu$ and with high variance ($\geq 0.08$) in the case of SMC.
This difference reflects in the behavior of the Kullback--Leibler divergence, which stabilizes around 133 for EMS while keeps increasing for SMC (Figure~\ref{fig:gaussian_comparison_ems_smc} bottom right).

We now consider the effect of the use of the approximated potentials $\weightN$ in place of the exact ones $\weight$ in the SMC scheme. We compare the $\ise$ for $f_{n+1}^N$ given by the SMC scheme with exact and approximated potentials for values of the number $M$ of samples $Y^{ij}_n$ drawn from $h$ at each time step between 1 and $10^3$ with 1,000 repetitions for each $M$.
Through this comparison we also address the computational complexity $O(MN)$ of the algorithm, with focus on the choice of the value of $M$.
Figure~\ref{fig:gaussian_M} shows the results for $N = 10^3$ and $\varepsilon = 10^{-2}$. The behavior for different values of $N$ and $\varepsilon$ is similar.
The plot of $\ise(f_{n+1}^N)$ shows a significant improvement when $M>1$ but little further improvement for $M>10$.

\begin{figure}[t]
\centering
\resizebox{0.9\textwidth}{!}{%
\begin{tikzpicture}[every node/.append style={font=\normalsize}]
\node (img1) {\includegraphics[width=0.46\textwidth]{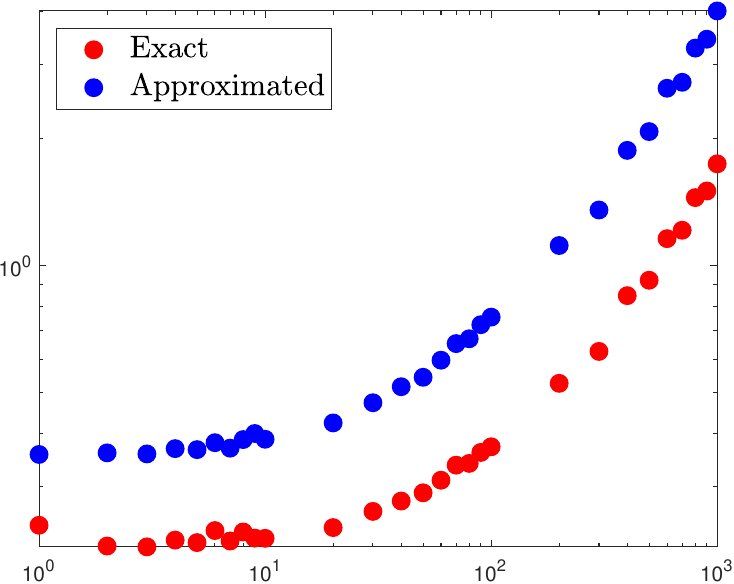}};
\node[below=of img1, node distance = 0, yshift = 1cm] { $M$};
  \node[left=of img1, node distance = 0, rotate=90, anchor = center, yshift = -0.7cm] { Runtime (s)};
\end{tikzpicture}
\begin{tikzpicture}[every node/.append style={font=\normalsize}]
\node[right= of img1, yshift = 0.1cm] (img2){\includegraphics[width=0.45\textwidth]{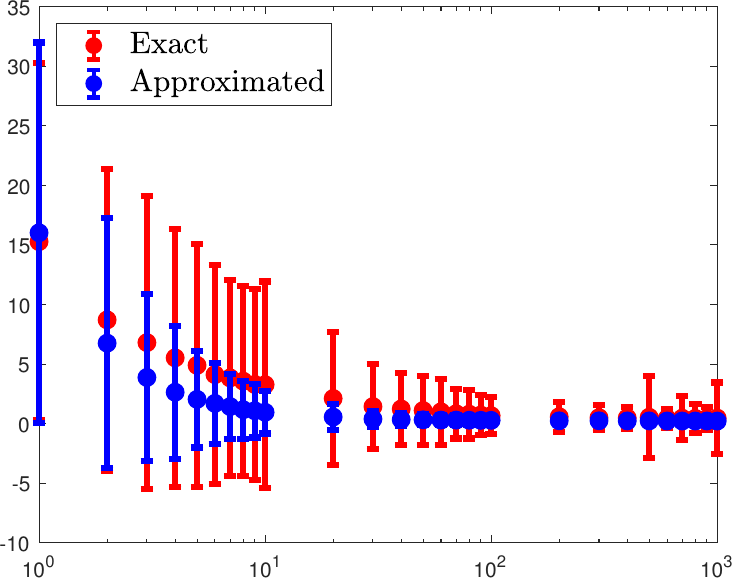}};
\node[below=of img2, node distance = 0, yshift = 1cm] {$M$};
  \node[left=of img2, node distance = 0, rotate=90, anchor = center, yshift = -0.7cm] {$\ise(f_{n+1}^N)$};
\end{tikzpicture}
}
\caption{Dependence of runtime and $\ise(f_{n+1}^N)$ on the value of $M$, the number of samples drawn from $h$ at each iteration, for the SMC scheme run with the exact potential (blue) and the approximated potential (red). The error bars represent twice the standard deviation of $\ise(f_{n+1}^N)$.}
\label{fig:gaussian_M}
\end{figure}

To further investigate the choice of $M$ we compare the reconstructions obtained using the exact and the approximated potentials for $M=10$, $M=10^2$ and $M=N=10^3$.
Figure~\ref{fig:gaussian_reconstruction} shows pointwise means and pointwise MSE~\eqref{eq:mse} for 1,000 reconstructions. The means of the reconstructions with the exact potentials (blue) coincide for the three values of $M$, the means of the reconstructions with the approximated potentials (red) also coincide but have heavier tails than those obtained with the exact potentials.
The MSE is similar for reconstructions with exact and approximated potentials with the same value of $M$. In particular, the little improvement of the MSE from $M=10^2$ to $M=10^3$ suggests that $M=10^2$ could be used instead of $M=N=10^3$ if the computational resources are limited. Using $M=10^2$ instead of $M=10^3$ reduces the average runtime by $\approx 80\%$ for both the algorithm using the exact potentials and that using the approximated potentials.

\begin{figure}[t]
\centering
\resizebox{0.9\textwidth}{!}{%
\begin{tikzpicture}[every node/.append style={font=\normalsize}]
\node (img1) {\includegraphics[width=0.45\textwidth]{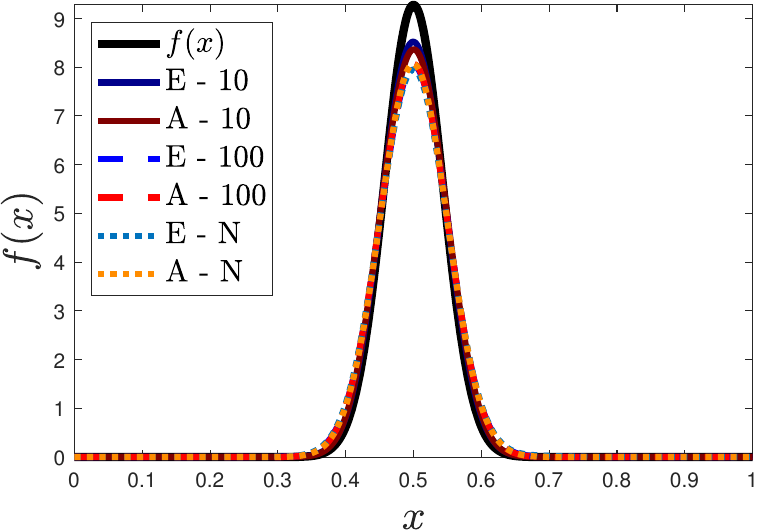}};
\node[below=of img1, node distance = 0, yshift = 1cm] { $x$};
  \node[left=of img1, node distance = 0, rotate=90, anchor = center, yshift = -0.7cm] { $\widehat{\Exp}\left[f^N_n(x)\right]$};
\end{tikzpicture}
\begin{tikzpicture}[every node/.append style={font=\normalsize}]
\node[right= of img1, yshift = 0.1cm] (img2){\includegraphics[width=0.45\textwidth]{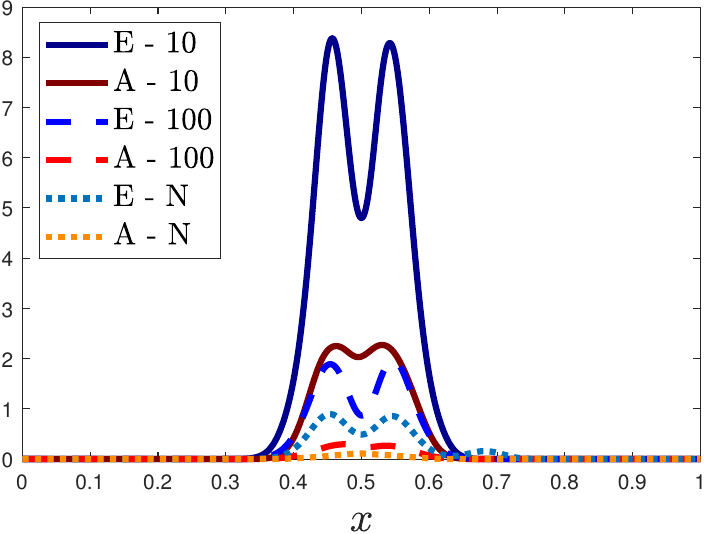}};
\node[below=of img2, node distance = 0, yshift = 1cm] {$x$};
  \node[left=of img2, node distance = 0, rotate=90, anchor = center, yshift = -0.7cm] { $\widehat{\text{MSE}}(x_c)$};
\end{tikzpicture}
}
\caption{Reconstruction of $f(x) = \N(x; 0.5, 0.043^2)$ from data distribution $h(y) = \N(y; 0.5, 0.043^2 + 0.045^2)$. The number of particles $N$ is $10^3$ and the smoothing parameter $\varepsilon=10^{-2}$. $M=10$, $M=10^2$ and $M=N$ are compared through the pointwise means of the reconstructions and the pointwise mean squared error (MSE).}
 \label{fig:gaussian_reconstruction}
\end{figure}
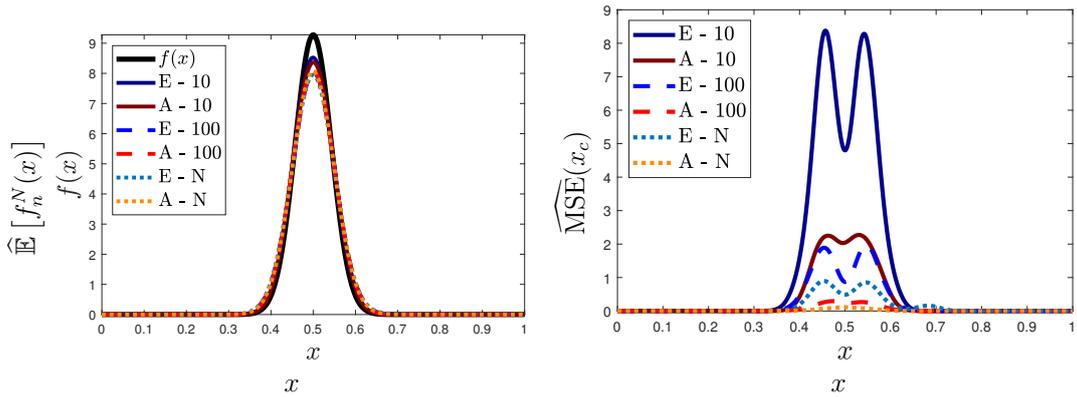

\citetapp[Section 5.4]{silverman1990smoothed} conjectured that under suitable assumptions the EMS map~\eqref{eq:ems} has a unique fixed point. This conjecture is empirically confirmed by the results in Figure~\ref{fig:gaussian_uniqueness}.
We run EM, EMS and SMC with approximated potentials for $n=100$ iterations starting from three initial distributions $f_1(x)$: a Uniform on $[0, 1]$, a Dirac $\delta$ centered at 0.5 and the solution $\N(x; \mu, \sigma_f^2)$.
The number of particles is set to $N=10^3$ and the smoothing parameter $\varepsilon = 10^{-1}$.
Both EMS and SMC converge to the same value of the Kullback--Leibler divergence regardless of the starting distribution. The speed of convergence of the three algorithms is similar, in each case little further change is observed once 4 iterations have occurred.

\begin{figure}[t]
\centering
\resizebox{0.65\textwidth}{!}{%
\begin{tikzpicture}[every node/.append style={font=\normalsize}]
\node (img1) {\includegraphics[width=0.6\textwidth]{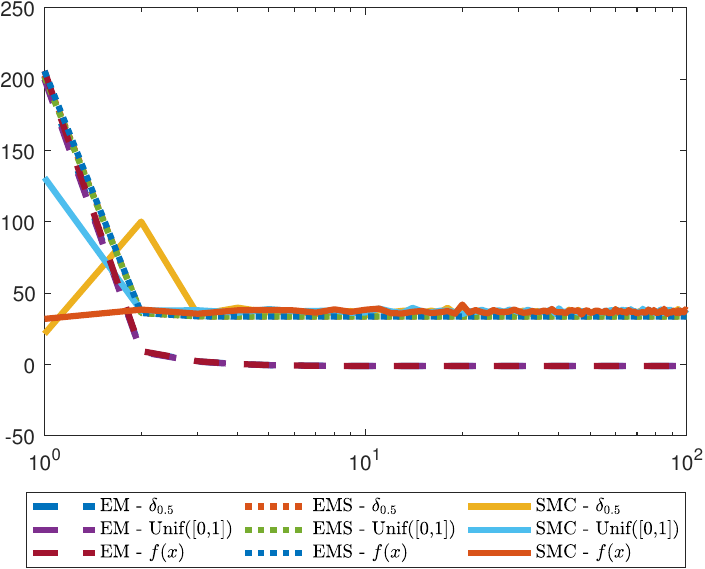}};
\node[below=of img1, node distance = 0, yshift = 1cm] { $\log_{10}($iteration$)$};
  \node[left=of img1, node distance = 0, rotate=90, anchor = center, yshift = -0.7cm] { $\KL(h, \hat{h}_n^N)$};
\end{tikzpicture}
}
\caption{Kullback--Leibler divergence between $h$ and $\hat{h}_n^N$ as function of the number of iterations. Three starting distributions are considered: Uniform$\left([0, 1]\right)$, $\delta_{0.5}$, $\N(x; \mu, \sigma_f^2)$. The behavior of EM (dashed lines), EMS (dotted lines) and SMC (solid lines) is compared.}
 \label{fig:gaussian_uniqueness}
\end{figure}

\FloatBarrier
\subsection{Motion deblurring}
Consider a simple example of motion deblurring where the observed picture $h$ is obtained while the object of interest is moving with constant speed $b$ in the horizontal direction \citepapp{vardi1993image, lee1994experiments}.
The constant motion in the horizontal direction is modeled by multiplying the density of a uniform random variable on $[-b/2, b/2]$ describing the motion in the horizontal direction and a Gaussian, $\N(v; y, \sigma^2)$, with small variance, $\sigma^2=0.02^2$, describing the relative lack of motion in the vertical direction
\begin{align*}
g(u, v \mid x, y) = \N(v; y, \sigma^2)\text{Uniform}_{[x-b/2, x+b/2]}(u).
\end{align*}

We obtain the corrupted image in Figure~\ref{fig:bcblur} from the reference image in Figure~\ref{fig:bc} using the model above with constant speed $b = 128$ pixels and adding multiplicative noise as in \citetapp[Section 6.2]{lee1994experiments}.
Figure~\ref{fig:bcblur} is a noisy discretization of the unknown $h(u, v)$ on a $300\times 600$ grid.
The addition of multiplicative noise makes the model~\eqref{eq:fe} misspecified, but still suitable to describe the deconvolution problem when the amount of noise is low. For higher levels of noise, the noise itself should be taken into account when modeling the generation of the data corresponding to $h$.

\begin{figure}[t]
\centering
\subfloat[Original sharp image \label{fig:bc}]{\includegraphics[width=0.37\textwidth]{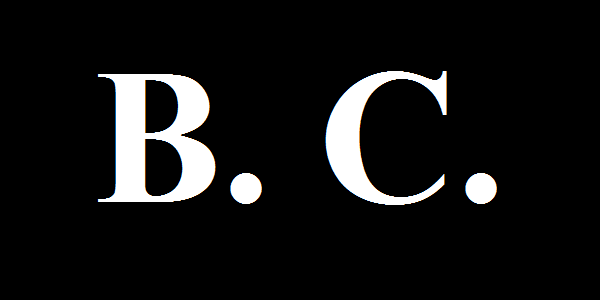}}\quad
\subfloat[Blurred image with 0.5\% multiplicative noise \label{fig:bcblur}]{\includegraphics[width=0.37\textwidth]{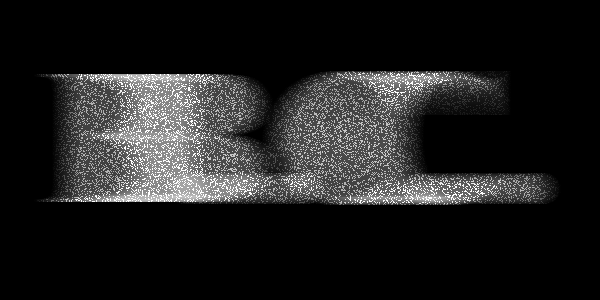}}
\\
\subfloat[Reconstruction with RL \label{fig:bcrl}]{\includegraphics[width=0.37\textwidth]{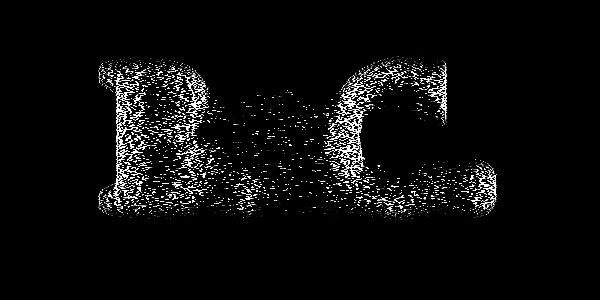}}\quad
\subfloat[Reconstruction with SMC \label{fig:bcsmc}]{\includegraphics[width=0.37\textwidth]{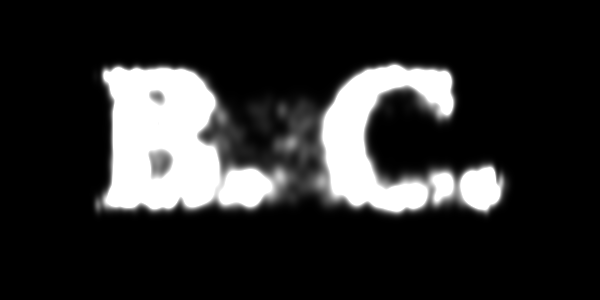}}
\caption{Reference image, blurred noisy data distribution and reconstructions for the motion deblurring example. Each scheme used 100 iterations; the SMC scheme used $N=5,000$ particles.}
\end{figure}

We compare the reconstruction obtained using the SMC scheme with that given by the \texttt{deconvlucy} function in MATLAB\textsuperscript{\textcopyright} \citepapp{rlmatlab}, an efficient implementation of the Richardson--Lucy algorithm  (i.e. EM for Poisson counts) for image processing which considers the data image as a discretization of the unknown density $h$ into bins. The same image is used to draw the samples necessary for the SMC implementation.

The smoothing parameter is $\varepsilon = 10^{-3}$, and the number of particles is $N=5,000$. These values are chosen to achieve a trade-off between smoothing and accuracy of the reconstruction and to keep the runtime under three minutes on a standard laptop.

The distance between the reconstructions and the original image is evaluated using both the ISE~\eqref{eq:ise} and the match distance, i.e. the $\mathbb{L}_1$ norm of the cumulative histogram of the image, a special case of the Earth Mover's Distance for gray-scale images \citepapp{rubner2000earth}.
SMC gives visibly smoother images and is better at recovering the shape of the original image ($\ise(f_{n+1}^N)$ is 1.4617 for SMC and 2.0863 for RL). In contrast, the RL algorithm performs better in terms of match distance (0.0054 for RL and 0.0346 for SMC).

\FloatBarrier
\section{Additional Results for PET Example}
\label{app:pet}

The reconstruction of cross-sectional images from projections given by PET scanners is modeled by the Radon transform \citepapp{radon1986determination}
\begin{align}
\label{eq:radon}
    h(\phi, \xi) = \int_{-\infty}^{+\infty} f(\xi\cos \phi-t\sin\phi, \xi\sin \phi+t\cos\phi)\dt,
\end{align}
for $(\phi, \xi)\in \Y=[0, 2\pi]\times[-R, R]$, where the right hand side is the line integral along the line with equation $x\cos\phi +y \sin\phi =\xi$.
We rewrite~\eqref{eq:radon} as a Fredholm integral equation~\eqref{eq:fe} modelling the alignment between the projections onto $(\phi, \xi)$ and the corresponding location $(x, y)$ in the reference image using a Gaussian distribution with small variance (in the experiments we use $\sigma^2= 0.02^2$)
\begin{align*}
h(\phi, \xi)
    = \int_{\X}\mathcal{N}\left( x\cos(\phi) +y \sin(\phi) -\xi; 0, \sigma^2\right)f(x, y)\dx\dy,
\end{align*}
where $\X=[-r, r]^2$.
The kernel $g(\phi, \xi \mid x, y) = \mathcal{N}\left( x\cos(\phi) +y \sin(\phi) -\xi; 0, \sigma^2\right)$ is not a Markov kernel (in the sense that it does not integrate to 1 for fixed $(x, y)$), however, we can use the re-normalization described in \citetapp[Section 6]{chae2018algorithm} to obtain the Markov kernel
\begin{align*}
   \tilde{g}(\phi, \xi \mid x, y) = \frac{g(\phi, \xi \mid x, y)}{C(x, y, \sigma^2)}
\end{align*}
where $C(x, y, \sigma^2)$ is the normalizing constant for each fixed $(x, y)\in \X$
\begin{align*}
    C(x, y, \sigma^2)=&\int_{\Y}g(\phi, \xi \mid x, y)\textrm{d}\phi\ \textrm{d}\xi\\
    =&\int_{0}^{2\pi}\frac{1}{2}\left[ \textrm{erf}\left(\frac{R-x\cos\phi-y\sin\phi}{\sqrt{2}\sigma}\right)+\textrm{erf}\left(\frac{R+x\cos\phi+y\sin\phi}{\sqrt{2}\sigma}\right)\right]\textrm{d}\phi
\end{align*}
with $\textrm{erf}$ the error function.
Recalling that $R=92$, $\phi\in[0, 2\pi]$ and $(x, y)\in[-64, 64]^2$ (i.e. we want to reconstruct a $128\times 128$ pixels image) and selecting $\sigma=0.02$ gives 
\begin{align*}
   \left\lvert \frac{1}{2}\left[ \textrm{erf}\left(\frac{R-x\cos\phi-y\sin\phi}{\sqrt{2}\sigma}\right)+\textrm{erf}\left(\frac{R+x\cos\phi+y\sin\phi}{\sqrt{2}\sigma}\right)\right] - 1\right\rvert< 10^{-17}
\end{align*}
for all $\phi\in[0, 2\pi]$ and $(x, y)\in[-64, 64]^2$.
The above shows that, for $\sigma^2$ sufficiently small (e.g. $\sigma^2=0.02^2$ as we use in our experiments), i.e. if the Gaussian distribution appropriately describes the alignment onto $x\cos\phi +y \sin\phi =\xi$,
\begin{align*}
  \left\lvert C(x, y, \sigma^2)- 2\pi\right\rvert <10^{-17}
\end{align*}
for all $(x, y)\in \X$.
Therefore we obtain, up to a negligible approximation, an integral equation satisfying~\ref{a:space}--\ref{a:g} dividing $h$ by $2\pi$:
\begin{align*}
\frac{h(\phi, \xi)}{2\pi}
    = \int_{\X}\frac{\mathcal{N}\left( x\cos(\phi) +y \sin(\phi) -\xi; 0, \sigma^2\right)}{2\pi}f(x, y)\dx\dy.
\end{align*}

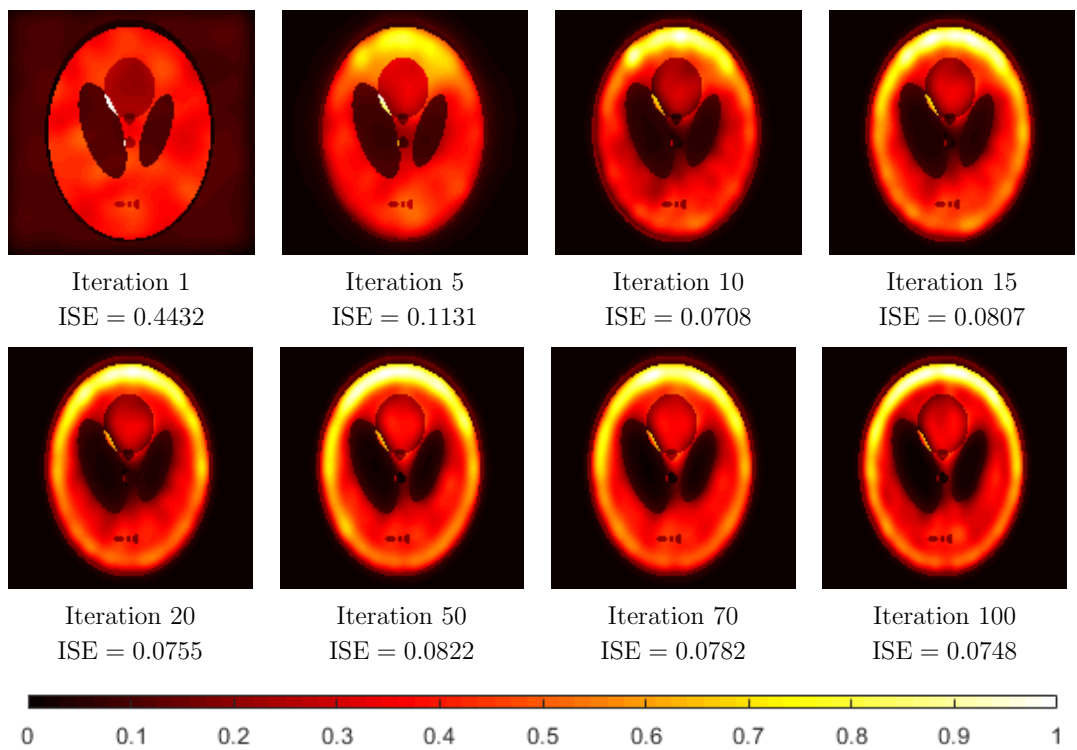
\begin{figure}[b]
\centering
\resizebox{0.9\textwidth}{!}{%
\begin{tikzpicture}[baseline, every node/.append style={font=\normalsize}]
\node (img1) {\includegraphics[width=0.25\textwidth]{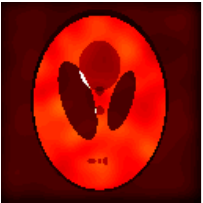}};
\node[below=of img1, node distance = 0, yshift = 1cm] (label1){Iteration 1};
\node[below=of label1, node distance = 0, yshift = 1cm] {$\ise = 0.4432$};
\end{tikzpicture}
\begin{tikzpicture}[baseline, every node/.append style={font=\normalsize}]
\node[right=of img1] (img2) {\includegraphics[width=0.25\textwidth]{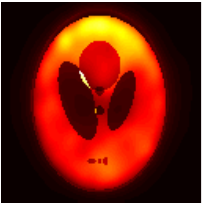}};
\node[below=of img2, node distance = 0, yshift = 1cm] (label2) {Iteration 5};
\node[below=of label2, node distance = 0, yshift = 1cm] {$\ise = 0.1131$};
\end{tikzpicture}
\begin{tikzpicture}[baseline, every node/.append style={font=\normalsize}]
\node[right=of img2] (img3) {\includegraphics[width=0.25\textwidth]{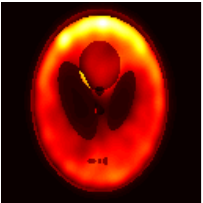}};
\node[below=of img3, node distance = 0, yshift = 1cm] (label3){Iteration 10};
\node[below=of label3, node distance = 0, yshift = 1cm] {$\ise = 0.0708$};
\end{tikzpicture}
\begin{tikzpicture}[baseline, every node/.append style={font=\normalsize}]
\node[right=of img3] (img4) {\includegraphics[width=0.25\textwidth]{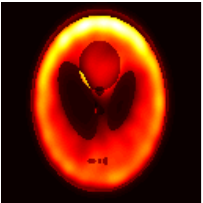}};
\node[below=of img4, node distance = 0, yshift = 1cm] (label4){Iteration 15};
\node[below=of label4, node distance = 0, yshift = 1cm] {$\ise = 0.0807$};
\end{tikzpicture}
}
\resizebox{0.9\textwidth}{!}{%
\begin{tikzpicture}[baseline, every node/.append style={font=\normalsize}]
\node (img5) {\includegraphics[width=0.25\textwidth]{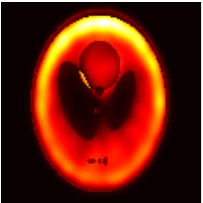}};
\node[below=of img5, node distance = 0, yshift = 1cm] (label5){Iteration 20};
\node[below=of label5, node distance = 0, yshift = 1cm] {$\ise = 0.0755$};
\end{tikzpicture}
\begin{tikzpicture}[baseline, every node/.append style={font=\normalsize}]
\node[right=of img5] (img6) {\includegraphics[width=0.25\textwidth]{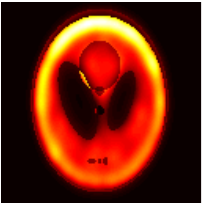}};
\node[below=of img6, node distance = 0, yshift = 1cm] (label6){Iteration 50};
\node[below=of label6, node distance = 0, yshift = 1cm] {$\ise = 0.0822$};
\end{tikzpicture}
\begin{tikzpicture}[baseline, every node/.append style={font=\normalsize}]
\node[right=of img6] (img7) {\includegraphics[width=0.25\textwidth]{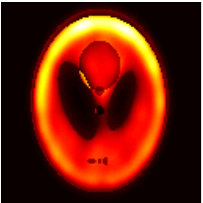}};
\node[below=of img7, node distance = 0, yshift = 1cm] (label7){Iteration 70};
\node[below=of label7, node distance = 0, yshift = 1cm] {$\ise = 0.0782$};
\end{tikzpicture}
\begin{tikzpicture}[baseline, every node/.append style={font=\normalsize}]
\node[right=of img7] (img8) {\includegraphics[width=0.25\textwidth]{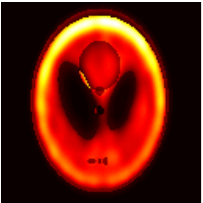}};
\node[below=of img8, node distance = 0, yshift = 1cm] (label8){Iteration 100};
\node[below=of label8, node distance = 0, yshift = 1cm] {$\ise = 0.0748$};
\end{tikzpicture}
\begin{tikzpicture}[baseline, every node/.append style={font=\normalsize}]
\end{tikzpicture}
}
\resizebox{0.9\textwidth}{!}{%
\begin{tikzpicture}[baseline, every node/.append style={font=\normalsize}]
\node {\includegraphics{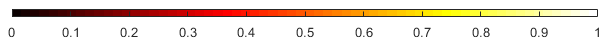}};
\end{tikzpicture}
}
\caption{Relative error for the reconstructions in Figure~\ref{fig:pet_reconstruction}. The $\ise$ at each iteration is given in the captions and stabilises below 0.08.}
\label{fig:pet_relative_error}
\end{figure}

Figure~\ref{fig:pet_relative_error} shows relative error and $\ise$ for the reconstructions in Figure~\ref{fig:pet_reconstruction}; the $\ise$ between the original image and the reconstructions at iteration 50 to 100 stabilizes below 0.08.
The stopping criterion~\eqref{eq:stop} is a trade-off between Monte Carlo error and convergence to a fixed point. In particular, when $\zeta( f^N_{k})= \int_{\X}\vert f^N_{k}(x)\vert^2\dx$, larger values of $N$ will make the r.h.s. of~\eqref{eq:stop} smaller which corresponds to a smaller tolerance to assess the convergence to the fixed point. On the other hand, small values of $N$ will give poorer reconstructions and might require more iterations $n$ to satisfy the stopping criterion~\eqref{eq:stop}. For instance, for $N=1,000$ the stopping criterion is not satisfied in 100 iterations despite the r.h.s. of~\eqref{eq:stop} being of order $10^{-3}$ against the $ 10^{-5}$ order when $N=20,000$.

\FloatBarrier
\section{Effect of Lower Bound on Gaussian Mixture Example}
\label{app:lb}
Consider the example in Section~\ref{sec:indirect_density_estimation} and instead of defining the integrals on $\X=\Y=\real$ take $\X=\Y=[0.4-a, 0.4+a]$ with $a\rightarrow\infty$ so that we obtain the integral equation
\begin{align*}
    \tilde{h}(y) = \int_{\X} \tilde{f}(x)\tilde{g}(y \mid x)\dx
\end{align*}
with
\begin{align*}
    \tilde{h}(y)&= \frac{h(y)}{\frac{1}{3}C(a, 0.3, 0.045^2 + 0.015^2)+\frac{2}{3}C(a, 0.5, 0.045^2 + 0.043^2)}\\
    \tilde{g}(y \mid x) &= \frac{g(y \mid x)}{C(a, x, 0.045^2)}\\
    \tilde{f}(x) &= \frac{f(x)C(a, x, 0.045^2)}{\frac{1}{3}C(a, 0.3, 0.045^2 + 0.015^2)+\frac{2}{3}C(a, 0.5, 0.045^2 + 0.043^2)}
\end{align*}
where
\begin{align*}
    C(a, \mu, \sigma):= \int_{0.4-a}^{0.4+a}\N(x; \mu, \sigma^2)\dx=\frac{1}{2} \left(\textrm{erf} \left(\frac{(a + 0.4 - \mu)}{\sqrt{2\sigma^2}}\right) - \textrm{erf}\left(\frac{(0.4 -a - \mu)}{\sqrt{2\sigma^2}}\right)\right).
\end{align*}

In any of the intervals $[0.4-a, 0.4+a]$ assumption~\ref{a:g} is satisfied, in particular $\tilde{g}$ is bounded below. We study the behaviour of the reconstructions as $a\rightarrow\infty$ to check the influence of the lower bound on $g$ on the accuracy of the reconstructions measured through the average $\ise$ in~\eqref{eq:ise} over 100 repetitions. The algorithmic set up is the same of Section~\ref{sec:indirect_density_estimation}.
Figure~\ref{fig:mixture_lb} show that for $a\in[0.2, 1]$ (which corresponds to $\X=\Y=[0.2, 0.6]$ up to $\X=\Y=[-0.6, 1.4]$) the average reconstruction error is not influenced by the lower bound on $g$, the behaviour for larger values of $a$ is equivalent since $|1 - \int_{-0.6}^{1.4} f(x) \dx| < 10^{-30}$.

\begin{figure}[t]
\centering
\resizebox{0.65\textwidth}{!}{%
\begin{tikzpicture}[every node/.append style={font=\normalsize}]
\node (img1) {\includegraphics[width=0.6\textwidth]{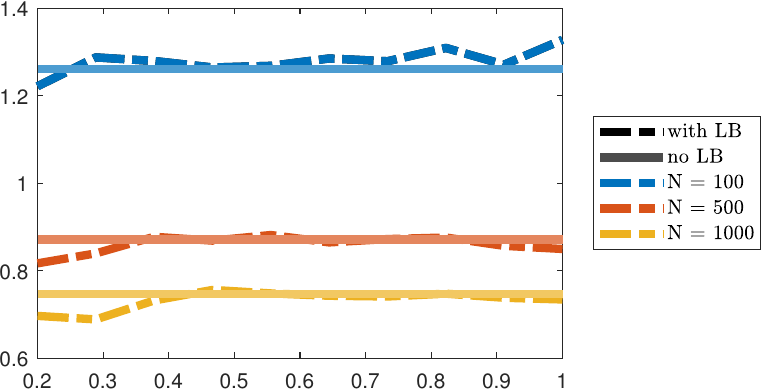}};
\node[below=of img1, node distance = 0, yshift = 1cm, xshift = -1cm] { $a$};
  \node[left=of img1, node distance = 0, rotate=90, anchor = center, yshift = -0.7cm] { $\ise(f_{n+1}^N)$};
\end{tikzpicture}
}
\caption{Influence of the lower bound (LB) on $g$ on average reconstruction accuracy over 100 repetitions. The solid lines represent the average $\ise(f_{n+1}^N)$ for $a=\infty$ while the dashed lines the average $\ise(f_{n+1}^N)$ for finite $a$.}
 \label{fig:mixture_lb}
\end{figure}
\section{Scaling with dimension}
\label{app:pdim}
To explore the scaling with the dimension $d_{\X}$ of the domain of $f$ of the discretized EMS~\eqref{eq:ems_discrete} and the SMC implementation of EMS we revisit the Gaussian mixture model in Section~\ref{sec:indirect_density_estimation} and extend it to higher dimension
\begin{align*}
f(x) =& \frac{1}{3}\N(x;0.3\cdot\mathbf{1}_{d_\X}, 0.07^2 I_{d_\X}) + \frac{2}{3}\N(x; 0.7\cdot\mathbf{1}_{d_\X}, 0.1^2I_{d_\X}),\\
g\left(y \mid x\right)  =& \N(y; x, 0.15^2 I_{d_\X}),\\
h(y)  =&  \frac{1}{3}\N(y;0.3 \cdot\mathbf{1}_{d_\X}, (0.07^2+0.15^2)I_{d_\X}) + \frac{2}{3}\N(y; 0.7\cdot\mathbf{1}_{d_\X}, (0.1^2+0.15^2) I_{d_\X}),
\end{align*}
where $\X=\Y=\real^{d_{\X}}$ and $\mathbf{1}_{d_\X}$ and $I_{d_\X}$ denote the unit function in $\mathbb{R}^{d_\X}$ and the $d_\X \times d_\X$ identity matrix, respectively. In particular, note that for $d_{\X}$ up to 5 at least $97\%$ of the mass of $f$ is contained in $[0, 1]^{d_{\X}}$.
We do not consider DKDE as these estimators approximate each marginal of $f$ separately and then use the product of the marginals as approximation for $f$. In the particular mixture model we consider, this results in reconstructions with additional modes due to the underlying independence assumption (e.g. reconstructions of the 2-dimensional model in Figure~\ref{fig:p_dim} present two additional modes at $(0.7, 0.3)$ and $(0.3, 0.7)$).

\begin{figure}[h]
    \centering
    \resizebox{0.98\textwidth}{!}{%
\begin{tikzpicture}[baseline, every node/.append style={font=\footnotesize}]
\node (img1) {\includegraphics[width=0.4\textwidth]{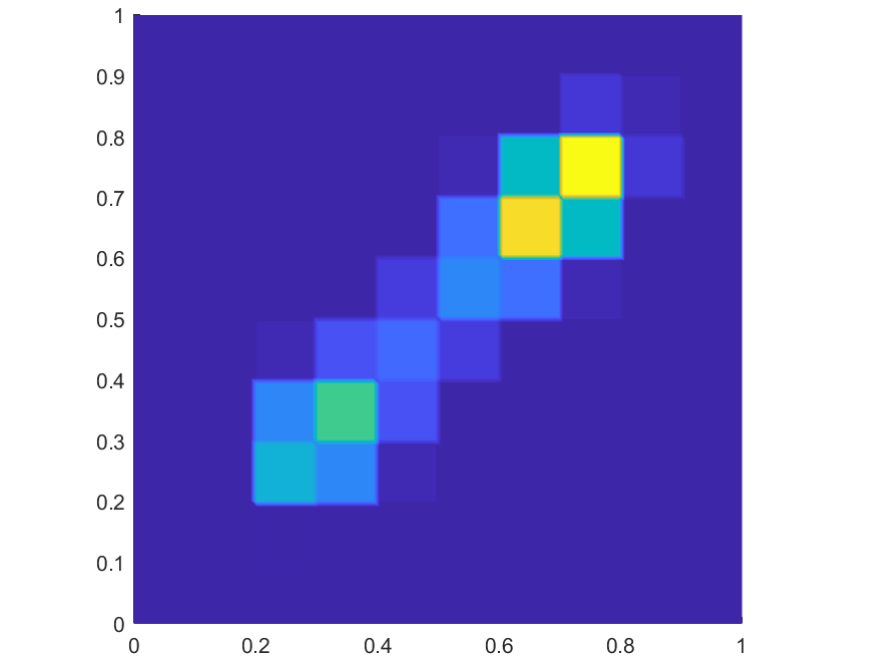}};
\node[left=of img1, node distance = 0, xshift = 1cm] {$\mathbf{10^2}$};
\node[above=of img1, node distance = 0] {\textbf{EMS}};
\node[below=of img1, node distance = 0, yshift = 1cm] {runtime $<$ 1s, $\ise=0.56$};
\node[right=of img1, node distance = 0, xshift = -2cm] (img2) {\includegraphics[width=0.4\textwidth]{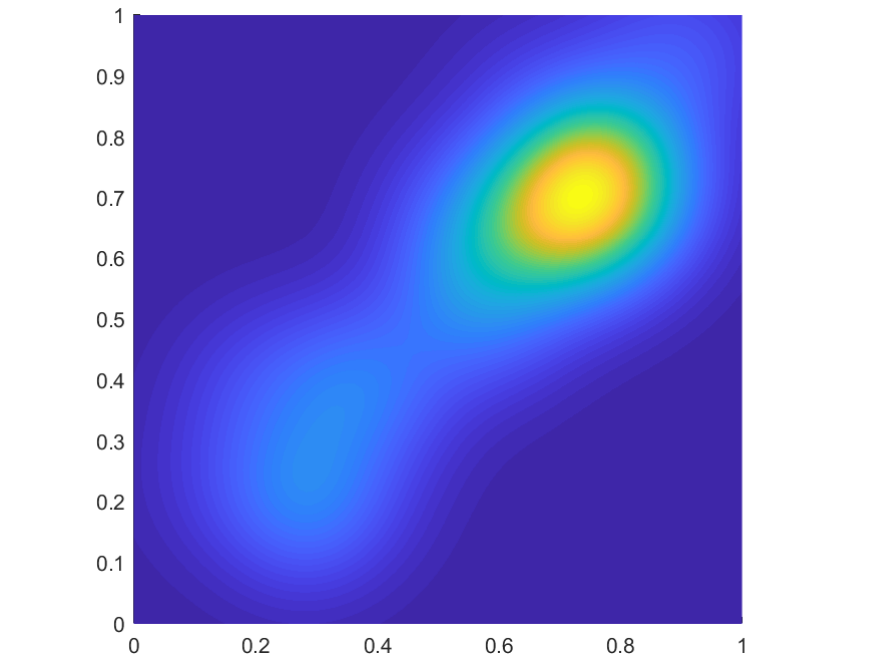}};
\node[above=of img2, node distance = 0,] {\textbf{SMC}};
\node[below=of img2, node distance = 0, yshift = 1cm] {runtime $<$ 1s, $\ise=0.91$};
\node[right=of img2, node distance = 0, xshift = -2cm] (img3) {\includegraphics[width=0.4\textwidth]{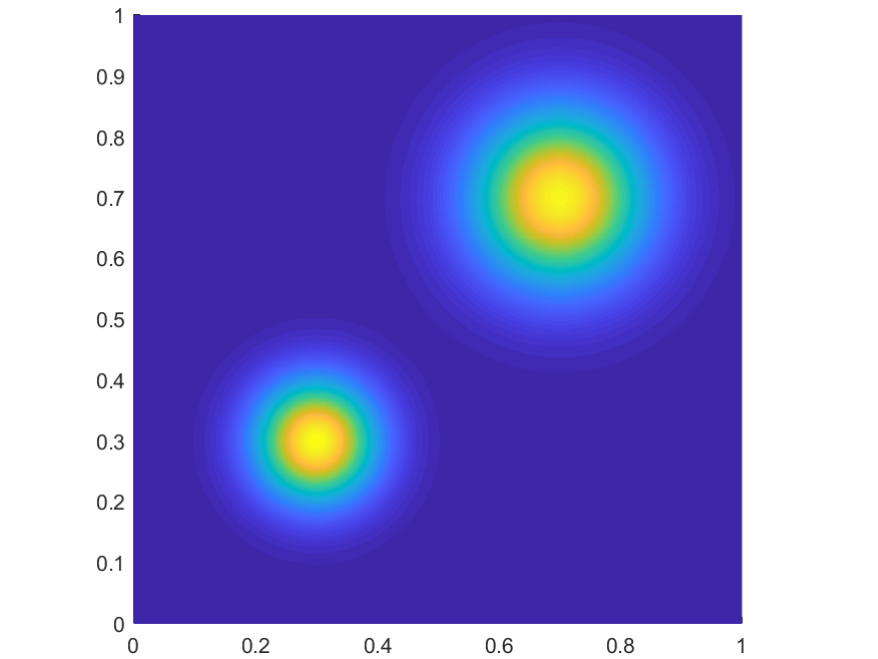}};
\node[above=of img3, node distance = 0] {\textbf{Truth}};
\end{tikzpicture}
}
    \resizebox{0.98\textwidth}{!}{%
\begin{tikzpicture}[baseline, every node/.append style={font=\footnotesize}]
\node (img1) {\includegraphics[width=0.4\textwidth]{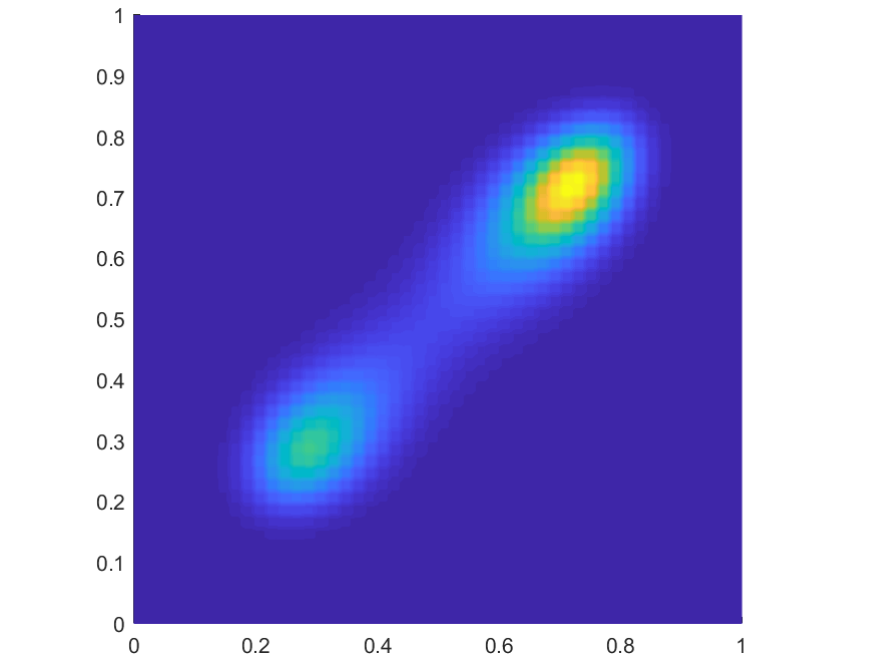}};
\node[left=of img1, node distance = 0, xshift = 1cm] {$\mathbf{50^2}$};
\node[below=of img1, node distance = 0, yshift = 1cm] {runtime $<$ 1m, $\ise=0.65$};
\node[right=of img1, node distance = 0, xshift = -2cm] (img2) {\includegraphics[width=0.4\textwidth]{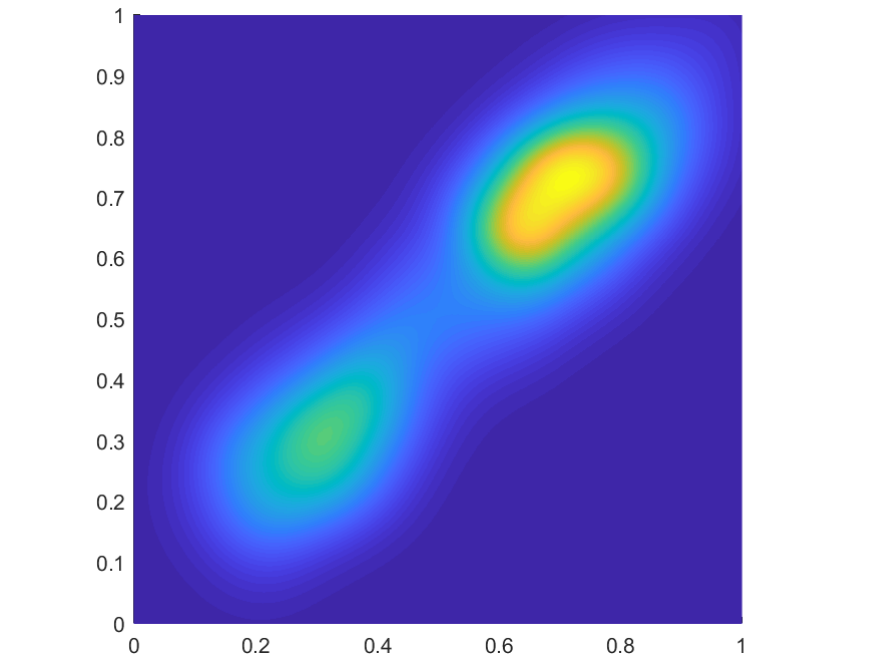}};
\node[below=of img2, node distance = 0, yshift = 1cm] {runtime $<$ 1m, $\ise=0.33$};
\node[right=of img2, node distance = 0, xshift = -2cm] (img3) {\includegraphics[width=0.4\textwidth]{mixture2}};
\end{tikzpicture}
}
    \resizebox{0.98\textwidth}{!}{%
\begin{tikzpicture}[baseline, every node/.append style={font=\footnotesize}]
\node (img1) {\includegraphics[width=0.4\textwidth]{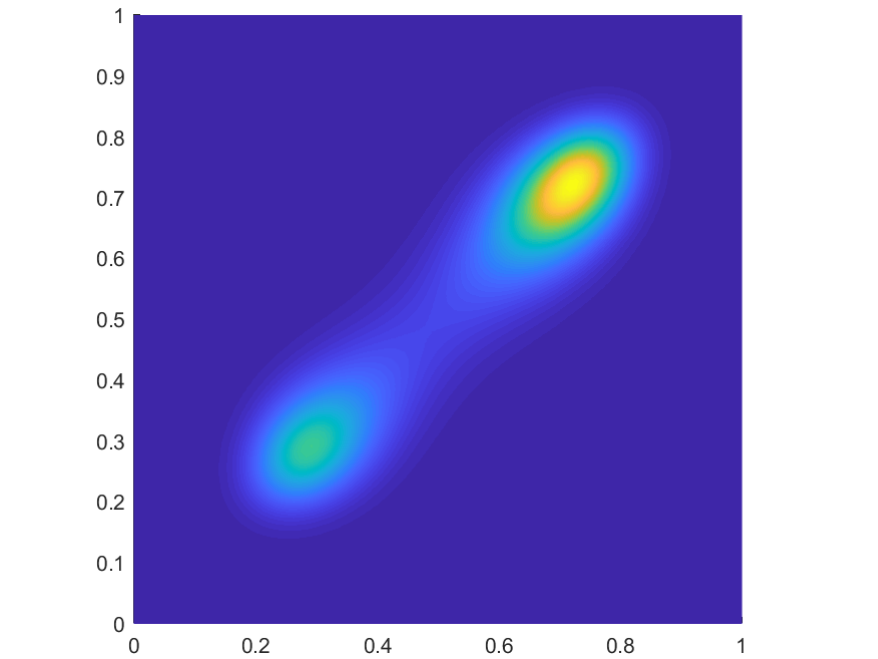}};
\node[left=of img1, node distance = 0, xshift = 1cm] {$\mathbf{100^2}$};
\node[below=of img1, node distance = 0, yshift = 1cm] {runtime $\approx$ 7m, $\ise=0.69$};
\node[right=of img1, node distance = 0, xshift = -2cm] (img2) {\includegraphics[width=0.4\textwidth]{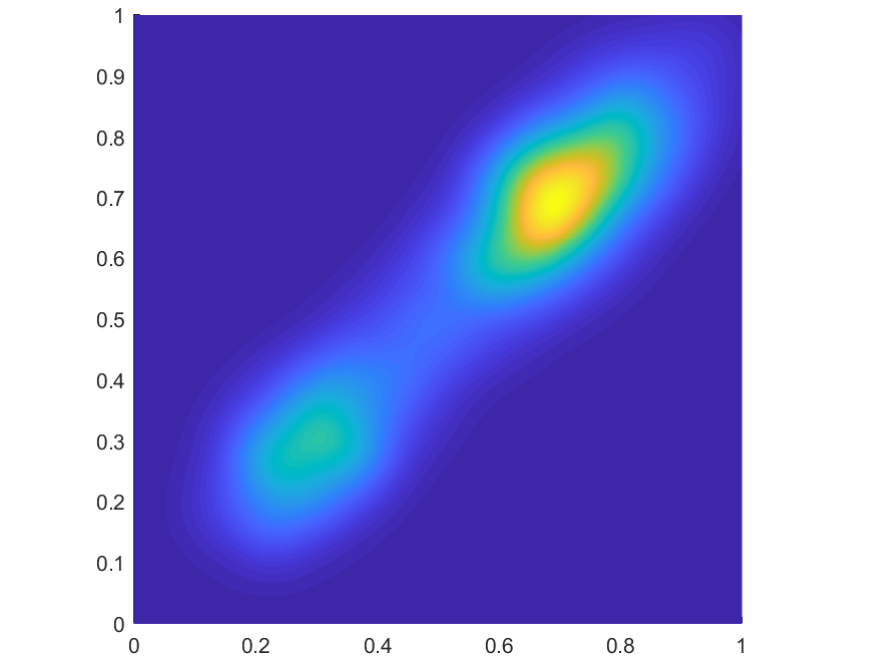}};
\node[below=of img2, node distance = 0, yshift = 1cm] {runtime $\approx$ 5m, $\ise=0.32$};
\node[right=of img2, node distance = 0, xshift = -2cm] (img3) {\includegraphics[width=0.4\textwidth]{mixture2}};
\end{tikzpicture}
}
    \caption{Reconstructions of a 2-dimensional mixture of Gaussian obtained with EMS and SMC. The number of bins/particles increases from $10^2$ to $100^2$. Runtime and accuracy are reported too.}
    \label{fig:p_dim}
\end{figure}

First, we take $d_{\X}=2$ and investigate the minimum number of bins/particles necessary to achieve reasonably good reconstructions.
We consider three particle sizes $N=10^2, 50^2, 100^2$ and set the total number of bins $B\approx N$ so that we obtain $\approx N^{1/2}$ equally spaced bins for each dimension.  We stop iterating after 30 steps since we observed that convergence occurs within 30 iterations, the value of $\varepsilon=10^{-3}$ is fixed and used for both the smoothing kernel and the smoothing matrix. The initial distribution is a uniform over $[0, 1]^2$ and we assume we have a sample $\mathbf{Y}$ of size $10^6$ from $h$, so that $M=N$. This corresponds to the highest computational cost for Algorithm~\ref{alg:fpsmc} but as observed in Appendix~\ref{sec:analytically_tractable} smaller values of $M$ could be considered and would reduce the computational cost of running the SMC implementation of EMS.
For small values of $B$ the runtimes of EMS and SMC are similar, however the reconstructions obtained with EMS are poor and low resolution due to the very coarse discretization (Figure~\ref{fig:p_dim}-left panel); on the contrary, the presence the kernel density estimator~\eqref{eq:smc_kde1} guarantees smooth reconstructions even when the particle size is small (Figure~\ref{fig:p_dim}-middle panel).
In addition, as the number of particles $N$ increases the accuracy of the reconstructions provided by SMC keeps increasing, while the EMS reconstructions do not improve as quickly, a phenomenon we already observed for the one-dimensional example in Section~\ref{sec:indirect_density_estimation}.
For $N=B\geq 100$ the runtime of SMC is roughly $30\%$ less than that of EMS with the accuracy of SMC being always larger than that of EMS.

Since the accuracy of kernel density estimators decreases when the dimension increases \citepapp{silverman1986density} and is primarily used in this work for visualisation and human interpretation (which becomes less informative in higher dimension, with the exception of low dimensional projections), to compare the performances of EMS and SMC in dimension $d_{\X}\geq 2$ we focus on approximating expectations w.r.t. $\eta_{n+1}$ of appropriate test functions $\varphi$, in this case, in fact, Proposition~\ref{prop:lp} gives the rate of convergence in terms of the number of particles $N$.
In particular, we consider mean, variance, the probability of the region $[0, 0.5]^{d_{\X}}$ and the probability of a hyper-sphere of radius 0.3 around the mode at $(0.3,\dots, 0.3)$.
We compare three particle sizes $N=10^2, 10^3, 10^4$ and obtain the number of bins for each dimension as $\lceil N^{1/d_{\X}}\rceil$ so that the total number of bins,  $B=\lceil N^{1/d_{\X}}\rceil^{d_\X}$, where $\lceil \cdot \rceil$ denotes the ceiling function, roughly matches $N$.
This choice allows us to compare EMS and SMC reconstructions which require roughly the same runtime (Table~\ref{tab:pdim}).
The SMC implementation in generally better at recovering the variance and the probability of the region $[0, 0.5]^{d_{\X}}$. For small values of $N$, $B$, both SMC and EMS have larger errors with discretized EMS achieving better crude estimates. However, as $N$, $B$ increase SMC is consistently better at approximating the four quantities considered, in particular, in the case of mean and variance the estimates are at least one order of magnitude more accurate. This is achieved at a computational cost which is always smaller than that of EMS and that could be in principle reduced by considering smaller values of $M$.

\FloatBarrier

\begin{center}
\scriptsize
  \tablelasttail{\hline}
  \topcaption{Mean squared error over 100 repetitions for mean, variance, probability of the lower quadrant and probability of a circle around the mode for the $d_{\X}$-dimensional Gaussian mixture model. Runtimes are reported too. Best values are in \textbf{bold}.}
   \label{tab:pdim}
   \tablefirsthead{%
\hline
 & mean & variance &  $\mathbb{P}(\square)$ & $\mathbb{P}(\bigcirc)$ &$\log_{10} (\textrm{runtime / s})$\\ 
\hline}
\tablehead{%
\hline
\multicolumn{6}{l}{\small\sl continued from previous page}\\
\hline
 & mean & variance &  $\mathbb{P}(\square)$ & $\mathbb{P}(\bigcirc)$ &$\log_{10} (\textrm{runtime / s})$\\ 

\hline}

\begin{supertabular}{lccccc}

 \hline\noalign{\smallskip}
$d_{\X}=2$\\
\hline
EMS - $B=10^2$&\textbf{1.38e-04}&4.96e-05&5.30e-02&\textbf{7.04e-03}&-1.71\\
SMC - $N=10^2$&3.87e-04&\textbf{1.26e-05}&\textbf{4.70e-02}&1.46e-02&\textbf{-2.02}\\
\hline
EMS - $B=32^2$&1.42e-04&5.31e-05&5.17e-02&\textbf{5.86e-03}&1.28\\
SMC - $N=10^3$&\textbf{4.29e-05}&\textbf{5.81e-06}&\textbf{3.02e-02}&8.29e-03&\textbf{0.94}\\
\hline
EMS - $B=100^2$&1.42e-04&5.38e-05&5.15e-02&\textbf{6.11e-03}&5.31\\
SMC - $N=10^4$&\textbf{3.84e-06}&\textbf{4.51e-06}&\textbf{2.77e-02}&8.57e-03&\textbf{5.11}\\
\hline\noalign{\smallskip}
$d_{\X}=3$\\
\hline
EMS - $B=5^3$&\textbf{2.53e-04}&1.26e-04&1.46e-01&8.59e-03&-1.47\\
SMC - $N=10^2$ &3.76e-04&\textbf{3.23e-05}&\textbf{7.41e-02}&\textbf{7.56e-03}&\textbf{-2.06}\\
\hline
EMS - $B=10^3$ &2.00e-04&5.75e-05&9.00e-02&2.42e-03&1.40\\
SMC - $N=10^3$&\textbf{4.62e-05}&\textbf{8.50e-06}&\textbf{7.00e-02}&\textbf{1.54e-03}&\textbf{1.08}\\
\hline
EMS - $B=22^3$&2.04e-04&6.12e-05&8.83e-02&1.64e-03&5.66\\
SMC - $N=10^4$&\textbf{3.53e-0}6&\textbf{6.68e-06}&\textbf{6.61e-02}&\textbf{9.38e-04}&\textbf{5.30}\\
    \hline\noalign{\smallskip}
$d_{\X}=4$\\
\hline
EMS - $B=4^4$ &\textbf{1.98e-04}&\textbf{1.55e-05}&1.22e-01&\textbf{1.16e-03}&-0.65\\
SMC - $N=10^2$&4.77e-04&9.77e-05&\textbf{6.85e-02}&5.48e-03&\textbf{-2.08}\\
\hline
EMS - $B=6^4$&2.43e-04&4.02e-05&1.09e-01&7.80e-04&1.70\\
SMC - $N=10^3$&\textbf{3.45e-05}&\textbf{1.80e-05}&\textbf{8.68e-02}&\textbf{7.21e-04}&\textbf{0.95}\\
\hline
EMS - $B=10^4$&2.60e-04&6.59e-05&1.03e-01&5.54e-04&5.32\\
SMC - $N=10^4$&\textbf{4.10e-06}&\textbf{8.58e-06}&\textbf{8.95e-02}&\textbf{2.22e-04}&\textbf{5.12}\\
\hline\noalign{\smallskip}
$d_{\X}=5$\\
\hline
EMS - $B=3^5$ &\textbf{5.66e-05}&2.67e-04&2.12e-01&\textbf{1.27e-02}&-0.56\\
SMC - $N=10^2$ &6.59e-04&\textbf{1.34e-04}&\textbf{3.89e-02}&1.41e-02&\textbf{-1.96}\\
\hline
EMS - $B=4^5$ &2.42e-04&\textbf{2.08e-05}&1.29e-01&\textbf{7.59e-04}&1.51\\
SMC - $N=10^3$ &\textbf{5.57e-05}&4.54e-05&\textbf{7.49e-02}&9.10e-04&\textbf{1.14}\\
\hline
EMS - $B=7^5$ &2.82e-04&5.71e-05&1.36e-01&2.09e-04&6.63\\
SMC - $N=10^4$ &\textbf{3.39e-06}&\textbf{1.27e-05}&\textbf{8.62e-02}&\textbf{5.73e-05}&\textbf{5.36}\\
\end{supertabular}
\end{center}

\bibliographystyleapp{agsm}
\bibliographyapp{smcfe_biblio}
\end{document}